\documentclass[12pt]{article}
\usepackage{tikz}
\usetikzlibrary{matrix,arrows,shapes}
\usepackage{pgf}
\usepackage{url}
\usepackage{longtable}
  \usepackage{enumerate}
  \usepackage{amsthm}
  \usepackage{amssymb}
  \usepackage{amsmath}
    \usepackage{amscd}
  \usepackage{eucal}
  \usepackage{mathrsfs}
  \usepackage{upgreek}
  \usepackage{mathtools}
   \usepackage{dsfont}
      \usepackage{yfonts}
  \usepackage[T1]{fontenc}
  \usepackage{bbm}
  \usepackage{geometry}
    \usepackage{graphbox}
\usepackage{array}
    \usepackage[retainorgcmds]{IEEEtrantools}
\usepackage{booktabs}
  \usepackage{enumerate}
 \usepackage{feynmf}
  \usetikzlibrary{trees}
\usetikzlibrary{decorations.pathmorphing}
\usetikzlibrary{decorations.markings}
\tikzset{
 photon/.style={decorate, decoration={snake}, draw=red},
    electron/.style={draw=blue, postaction={decorate},
        decoration={markings,mark=at position .55 with {\arrow[draw=blue]{>}}}},
    gluon/.style={decorate, draw=magenta,
        decoration={coil,amplitude=4pt, segment length=5pt}},
    sderiv/.style={postaction={decorate},
        decoration={markings,mark=at position .3 with {\arrow{>}}}},
    tderiv/.style={postaction={decorate},
        decoration={markings,mark=at position .7 with {\arrow{<}}}},
    stderiv/.style={postaction={decorate},
        decoration={markings,mark=at position .7 with {\arrow{<}},mark=at position .3 with {\arrow{>}}}}
}
\usepackage{scalerel}
\usepackage{stackengine}
\newcommand\scaleobject[2]{\hstretch{#1}{\vstretch{#1}{#2}}}

\newcommand\circT{\mathbin{\scaleobject{0.6}{^\bigcirc}\mkern-6mu\scaleobject{0.5}{^{\TT}}}}

\definecolor{see}{RGB}{67,75,179}
\definecolor{darksee}{RGB}{42,44,148}
\definecolor{honey}{RGB}{232,180,129}
\definecolor{lighthoney}{RGB}{255,254,220}
\definecolor{citecol}{rgb}{0.5,0,0} 
\definecolor{blue1}{RGB}{130,150,209}
\DeclareSymbolFont{bbold}{U}{bbold}{m}{n}
\DeclareSymbolFontAlphabet{\mathbbold}{bbold}
\definecolor{see}{RGB}{67,75,179}
\usepackage{hyperref}
\hypersetup
{
bookmarks=true,
bookmarksopen=true,
pdfauthor=Eli Hawkins and Kasia Rejzner,
pdfsubject=Deformation Quantization, 
pdfmenubar=true, 
pdfhighlight=/O, 
colorlinks=true, 
pdfpagemode=None, 
pdfpagelayout=SinglePage, 
pdffitwindow=true, 
linkcolor=see, 
citecolor=red!75!darkgray, 
urlcolor=see 
}
\newcommand{\fA}{\mathfrak{A}}

\newcommand{\Gcal}{\mathcal{G}}  
\newcommand{\G}{\mathcal{G}}  

\newcommand{\Kcal}{\mathcal{K}}  
\newcommand{\Ccal}{\mathcal{C}}
\newcommand{\Dcal}{\mathcal{D}}
\newcommand{\Ecal}{\mathcal{E}} 

\newcommand{\Fcal}{\mathcal{F}}

\newcommand{\Mcal}{\mathcal{M}}
\newcommand{\Ocal}{\mathcal{O}}
\newcommand{\Scal}{\mathcal{S}}
\newcommand{\Pcal}{\mathcal{P}}
\newcommand{\Qcal}{\mathcal{Q}}
\newcommand{\Rcal}{\mathcal{R}}
\newcommand{\Tcal}{\mathcal{T}}

\newcommand{\Xcal}{\mathcal{X}}

\newcommand{\Ycal}{\mathcal{Y}}

\newcommand{\Ci}{\mathcal{C}^\infty} 
\newcommand{\WF}{\mathrm{WF}}         
\newcommand{\id}{\mathrm{id}}               

\newcommand{\loc}{\mathrm{loc}}

\newcommand{\reg}{\mathrm{reg}}

\newcommand{\mc}{{\mu\mathrm{c}}}

\newcommand{\NN}{\mathbb{N}}          
\newcommand{\ZZ}{\mathbbmss{Z}}     
\newcommand{\RR}{\mathbb{R}}           
\newcommand{\CC}{\mathbb{C}}           

\newcommand{\bet}{\beta}

\newcommand{\la}{\lambda}

\newcommand{\ph}{\varphi}

\newcommand{\T}{\cdot_{{}^\Tcal}}

\newcommand{\TT}{\Tcal}

\newcommand{\DFp}{\Dcal_{\mathrm{F}}}

\newcommand{\eom}{{\textsc{eom}}}

\newcommand{\sst}[1]{\scriptscriptstyle{#1}}  
\newcommand{\minus}{\sst{-1}}   
\newcommand{\be}{\begin{equation}}
\newcommand{\ee}{\end{equation}}

\DeclareMathOperator{\supp}{supp}      

\newcommand{\Pei}[2]{\{#1,#2\}}

\newcommand{\mT}{m_{\sst{\TT}}}

\newcommand{\starb}{\mathbin{\star_{\sst \TT}}}
\DeclareMathOperator{\Aut}{Aut}
\newcommand{\abs}[1]{\lvert#1\rvert}
\newcommand{\acts}[1]{\overset{\twoheadrightarrow}{#1}}

\newcommand{\starz}{\star_{\sst 0}}
\newcommand{\bidis}{K}
\newcommand{\dTH}{\cdot_{{}^{\TT_H}}}
\newcommand{\starbint}{\star_{\sst{\TT,\mathrm{int}}}}
\DeclareMathOperator{\Iso}{Iso}
\newcommand{\starH}{\star_{\sst H}}
\newcommand{\starHint}{\star_{\sst{H,\mathrm{int}}}}

\newcommand{\bint}{\bullet_{\sst{\mathrm{int}}}}
\newcommand{\startr}{\mathbin{\star_{\sst \TT,r}}}
\newcommand{\Weyl}{{\mathrm{Weyl}}}
\DeclareMathOperator{\Had}{Had}

\allowdisplaybreaks[1]

\newcommand{\SolidLine}{\begin{tikzpicture}[baseline=90mm,x=1.00mm, y=1.00mm, inner xsep=0pt, inner ysep=0pt, outer xsep=0pt, outer ysep=0pt]
	\useasboundingbox  (83.00,89.92) rectangle +(20,5.08);
	\path[line width=0mm] (82.92,88.38) rectangle +(20.60,8.20);
	\definecolor{L}{rgb}{0,0,0}
	\path[line width=0.30mm, draw=L] (93.00,92.90);
	\path[line width=0.30mm, draw=L, dash pattern=on 3pt off 2pt] (93.34,94.58);
	\path[line width=0.21mm, draw=L] (92.81,91.74) -- (91.72,90.99) -- (92.81,90.38);
	\path[line width=0.21mm, draw=L] (84.92,90.99) -- (101.52,90.99);
	\end{tikzpicture}}
\newcommand{\SolidLineOneTwo}{\begin{tikzpicture}[baseline=90mm,x=1.00mm, y=1.00mm, inner xsep=0pt, inner ysep=0pt, outer xsep=0pt, outer ysep=0pt]
	\useasboundingbox  (83.00,89.92) rectangle +(20,3.5);
	\path[line width=0mm] (82.89,87.97) rectangle +(19.63,8.61);
	\definecolor{L}{rgb}{0,0,0}
	\path[line width=0.30mm, draw=L] (93.00,92.90);
	\path[line width=0.30mm, draw=L, dash pattern=on 3pt off 2pt] (93.34,94.58);
	\path[line width=0.21mm, draw=L] (92.77,91.75) -- (91.68,91.01) -- (92.77,90.39);
	\path[line width=0.21mm, draw=L] (85.98,91.03) -- (99.20,91.05);
	\draw(84.89,90.40) node[anchor=base west]{\fontsize{5.69}{6.83}\selectfont 1};
	\draw(99.52,90.43) node[anchor=base west]{\fontsize{5.69}{6.83}\selectfont 2};
	\end{tikzpicture}}
\newcommand{\GraphTwentyFour}{\begin{tikzpicture}[baseline=90mm,x=1.00mm, y=1.00mm, inner xsep=0pt, inner ysep=0pt, outer xsep=0pt, outer ysep=0pt]
	\useasboundingbox  (83.00,89.92) rectangle +(20,3.5);
	\path[line width=0mm] (82.48,88.00) rectangle +(20.22,8.58);
	\definecolor{L}{rgb}{0,0,0}
	\path[line width=0.30mm, draw=L] (93.00,92.90);
	\path[line width=0.30mm, draw=L, dash pattern=on 3pt off 2pt] (93.34,94.58);
	\path[line width=0.21mm, draw=L] (85.89,91.01) -- (99.11,91.03);
	\definecolor{F}{rgb}{0,0,0}
	\path[line width=0.21mm, draw=L, fill=F] (85.89,91.01) -- (87.29,90.31) -- (85.89,91.01) -- (87.29,91.71) -- (85.89,91.01) -- cycle;
	\path[line width=0.21mm, draw=L, fill=F] (99.11,91.03) -- (97.89,91.73) -- (97.90,90.33) -- (99.11,91.03) -- cycle;
	\draw(84.48,90.43) node[anchor=base west]{\fontsize{5.69}{6.83}\selectfont 1};
	\draw(99.70,90.43) node[anchor=base west]{\fontsize{5.69}{6.83}\selectfont 2};
	\path[line width=0.30mm, draw=L, fill=F] (92.59,91.04) circle (0.38mm);
	\end{tikzpicture}}
\newcommand{\SolidLineWithVertex}{\begin{tikzpicture}[baseline=90mm,x=1.00mm, y=1.00mm, inner xsep=0pt, inner ysep=0pt, outer xsep=0pt, outer ysep=0pt]
	\useasboundingbox  (83.00,89.92) rectangle +(20,5.08);
	\path[line width=0mm] (83.01,88.42) rectangle +(20.60,8.16);
	\definecolor{L}{rgb}{0,0,0}
	\path[line width=0.30mm, draw=L] (93.00,92.90);
	\path[line width=0.30mm, draw=L, dash pattern=on 3pt off 2pt] (93.34,94.58);
	\path[line width=0.21mm, draw=L] (89.78,91.78) -- (88.69,91.03) -- (89.78,90.42);
	\path[line width=0.21mm, draw=L] (85.01,91.05) -- (101.60,91.05);
	\definecolor{F}{rgb}{0,0,0}
	\path[line width=0.30mm, draw=L, fill=F] (93.35,91.08) circle (0.38mm);
	\path[line width=0.21mm, draw=L] (98.52,91.78) -- (97.43,91.03) -- (98.52,90.42);
	\end{tikzpicture}}
\newcommand{\GraphOneOpposite}{\begin{tikzpicture}[baseline=90mm,x=1.00mm, y=1.00mm, inner xsep=0pt, inner ysep=0pt, outer xsep=0pt, outer ysep=0pt]
	\useasboundingbox  (83.00,89.92) rectangle +(21.67,5.08);
	\path[line width=0mm] (81.49,87.96) rectangle +(22.38,8.62);
	\definecolor{L}{rgb}{0,0,0}
	\path[line width=0.30mm, draw=L] (93.00,92.90);
	\draw(83.49,90.39) node[anchor=base west]{\fontsize{5.69}{6.83}\selectfont 1};
	\path[line width=0.30mm, draw=L, dash pattern=on 3pt off 2pt] (93.34,94.58);
	\path[line width=0.21mm, draw=L] (92.81,91.74) -- (93.78,91.01) -- (92.81,90.38);
	\path[line width=0.21mm, draw=L, dash pattern=on 3pt off 2pt] (84.92,90.99) -- (101.57,91.04);
	\definecolor{F}{rgb}{0,0,0}
	\path[line width=0.30mm, draw=L, fill=F] (101.49,91.02) circle (0.38mm);
	\end{tikzpicture}}

\newcommand{\GraphTwentyTwo}{\begin{tikzpicture}[baseline=89.2mm,x=1.00mm, y=1.00mm, inner xsep=0pt, inner ysep=0pt, outer xsep=0pt, outer ysep=0pt]
	\useasboundingbox (83.50,80.5) rectangle +(18.5,19);
	\path[line width=0mm] (82.30,80.03) rectangle +(21.26,19.91);
	\definecolor{L}{rgb}{0,0,0}
	\path[line width=0.30mm, draw=L] (93.00,92.90);
	\draw(84.30,89.34) node[anchor=base west]{\fontsize{5.69}{6.83}\selectfont 1};
	\draw(100.56,89.18) node[anchor=base west]{\fontsize{5.69}{6.83}\selectfont 2};
	\path[line width=0.30mm, draw=L, dash pattern=on 3pt off 2pt] (93.34,94.58);
	\path[line width=0.21mm, draw=L] (100.94,90.94) .. controls (100.41,94.93) and (97.00,97.92) .. (92.97,97.91) .. controls (89.03,97.91) and (85.68,95.04) .. (85.07,91.15);
	\path[line width=0.21mm, draw=L] (85.03,88.68) .. controls (85.69,84.82) and (89.05,82.01) .. (92.97,82.03) .. controls (96.83,82.05) and (100.13,84.84) .. (100.79,88.65);
	\path[line width=0.21mm, draw=L] (100.04,90.01) .. controls (96.04,90.02) and (92.93,92.94) .. (92.89,97.87) .. controls (92.89,98.03) and (93.01,97.87) .. (92.86,97.83);
	\path[line width=0.21mm, draw=L] (99.30,95.93) -- (98.00,96.19) -- (98.34,94.77);
	\path[line width=0.21mm, draw=L] (95.50,92.66) -- (94.20,92.93) -- (94.52,91.56);
	\path[line width=0.21mm, draw=L] (87.26,96.64) -- (87.00,95.29) -- (88.39,95.65);
	\path[line width=0.21mm, draw=L] (93.79,82.73) -- (92.68,82.01) -- (93.87,81.25);
	\end{tikzpicture}}
\newcommand{\GraphTwentyTwoLine}{\begin{tikzpicture}[baseline=89.2mm,x=1.00mm, y=1.00mm, inner xsep=0pt, inner ysep=0pt, outer xsep=0pt, outer ysep=0pt]
	\useasboundingbox (89,83) rectangle +(21,15);
	\path[line width=0mm] (87.43,82.41) rectangle +(23.65,14.50);
	\definecolor{L}{rgb}{0,0,0}
	\path[line width=0.30mm, draw=L] (95.32,91.80);
	\draw(89.43,89.36) node[anchor=base west]{\fontsize{5.69}{6.83}\selectfont 1};
	\draw(108.08,89.32) node[anchor=base west]{\fontsize{5.69}{6.83}\selectfont 2};
	\path[line width=0.30mm, draw=L, dash pattern=on 3pt off 2pt] (95.53,92.84);
	\path[line width=0.21mm, draw=L] (100.33,89.97) .. controls (100.27,92.71) and (98.04,94.90) .. (95.30,94.91) .. controls (92.85,94.91) and (90.77,93.13) .. (90.39,90.71);
	\path[line width=0.21mm, draw=L] (90.37,89.18) .. controls (90.78,86.78) and (92.87,85.03) .. (95.30,85.05) .. controls (98.02,85.06) and (100.24,87.22) .. (100.33,89.94);
	\path[line width=0.21mm, draw=L] (100.36,89.94) .. controls (97.02,89.92) and (95.28,91.82) .. (95.26,94.88);
	\path[line width=0.21mm, draw=L] (99.58,93.65) -- (98.28,93.92) -- (98.62,92.49);
	\path[line width=0.21mm, draw=L] (97.39,91.51) -- (96.09,91.78) -- (96.40,90.41);
	\path[line width=0.21mm, draw=L] (100.31,89.94) -- (107.87,89.94);
	\path[line width=0.21mm, draw=L] (91.66,94.26) -- (91.31,92.97) -- (92.75,93.22);
	\path[line width=0.21mm, draw=L] (104.95,90.74) -- (103.89,89.94) -- (105.05,89.33);
	\path[line width=0.21mm, draw=L] (96.20,85.82) -- (95.14,85.02) -- (96.29,84.41);
	\end{tikzpicture}}
\newcommand{\GraphTwentyTwoLineLeft}{\begin{tikzpicture}[baseline=89.2mm,x=1.00mm, y=1.00mm, inner xsep=0pt, inner ysep=0pt, outer xsep=0pt, outer ysep=0pt]
	\useasboundingbox (81,83) rectangle +(21,15);
	\path[line width=0mm] (79.54,82.41) rectangle +(23.47,14.50);
	\definecolor{L}{rgb}{0,0,0}
	\path[line width=0.30mm, draw=L] (95.32,91.80);
	\draw(81.54,89.38) node[anchor=base west]{\fontsize{5.69}{6.83}\selectfont 1};
	\draw(100.00,89.28) node[anchor=base west]{\fontsize{5.69}{6.83}\selectfont 2};
	\path[line width=0.30mm, draw=L, dash pattern=on 3pt off 2pt] (95.53,92.84);
	\path[line width=0.21mm, draw=L] (100.35,90.99) .. controls (99.81,93.33) and (97.70,94.97) .. (95.30,94.91) .. controls (92.51,94.83) and (90.29,92.55) .. (90.29,89.76);
	\path[line width=0.21mm, draw=L] (90.27,89.96) .. controls (90.36,87.24) and (92.58,85.08) .. (95.30,85.05) .. controls (97.66,85.02) and (99.71,86.64) .. (100.23,88.94);
	\path[line width=0.21mm, draw=L] (99.74,90.01) .. controls (96.86,90.34) and (95.28,91.82) .. (95.26,94.88);
	\path[line width=0.21mm, draw=L] (99.58,93.65) -- (98.28,93.92) -- (98.62,92.49);
	\path[line width=0.21mm, draw=L] (97.39,91.51) -- (96.09,91.78) -- (96.40,90.41);
	\path[line width=0.21mm, draw=L] (82.71,89.94) -- (90.28,89.94);
	\path[line width=0.21mm, draw=L] (91.66,94.26) -- (91.31,92.97) -- (92.75,93.22);
	\path[line width=0.21mm, draw=L] (87.51,90.71) -- (86.45,89.91) -- (87.61,89.30);
	\path[line width=0.21mm, draw=L] (96.20,85.82) -- (95.14,85.02) -- (96.29,84.41);
	\end{tikzpicture}}
\newcommand{\GraphTwentyTwoLineSymm}{\begin{tikzpicture}[baseline=89.2mm,x=1.00mm, y=1.00mm, inner xsep=0pt, inner ysep=0pt, outer xsep=0pt, outer ysep=0pt]
	\useasboundingbox (81,83) rectangle +(29,15);
	\path[line width=0mm] (79.54,82.32) rectangle +(31.54,14.67);
	\definecolor{L}{rgb}{0,0,0}
	\path[line width=0.30mm, draw=L] (95.32,91.80);
	\draw(81.54,89.38) node[anchor=base west]{\fontsize{5.69}{6.83}\selectfont 1};
	\draw(108.08,89.39) node[anchor=base west]{\fontsize{5.69}{6.83}\selectfont 2};
	\path[line width=0.30mm, draw=L, dash pattern=on 3pt off 2pt] (95.53,92.84);
	\path[line width=0.21mm, draw=L] (100.36,89.90) .. controls (96.94,90.10) and (95.28,91.82) .. (95.26,94.88);
	\path[line width=0.21mm, draw=L] (99.61,93.83) -- (98.30,94.09) -- (98.65,92.67);
	\path[line width=0.21mm, draw=L] (97.39,91.51) -- (96.09,91.78) -- (96.40,90.41);
	\path[line width=0.21mm, draw=L] (82.71,89.94) -- (90.28,89.94);
	\path[line width=0.21mm, draw=L] (91.57,94.20) -- (91.23,92.92) -- (92.67,93.17);
	\path[line width=0.21mm, draw=L] (87.51,90.71) -- (86.45,89.91) -- (87.61,89.30);
	\path[line width=0.21mm, draw=L] (96.05,85.73) -- (94.99,84.93) -- (96.15,84.32);
	\path[line width=0.21mm, draw=L] (95.32,89.97) circle (5.02mm);
	\path[line width=0.21mm, draw=L] (100.24,89.92) -- (107.80,89.92);
	\path[line width=0.21mm, draw=L] (105.11,90.68) -- (104.05,89.88) -- (105.20,89.27);
	\end{tikzpicture}}
\newcommand{\GraphTwentyTwoLineOpposite}{\begin{tikzpicture}[baseline=89.2mm,x=1.00mm, y=1.00mm, inner xsep=0pt, inner ysep=0pt, outer xsep=0pt, outer ysep=0pt]
	\useasboundingbox (89,83) rectangle +(21,15);
	\path[line width=0mm] (87.43,82.41) rectangle +(23.65,14.50);
	\definecolor{L}{rgb}{0,0,0}
	\path[line width=0.30mm, draw=L] (93.32,90.00);
	\draw(89.43,89.36) node[anchor=base west]{\fontsize{5.69}{6.83}\selectfont 1};
	\draw(108.08,89.32) node[anchor=base west]{\fontsize{5.69}{6.83}\selectfont 2};
	\path[line width=0.30mm, draw=L, dash pattern=on 3pt off 2pt] (92.31,90.32);
	\path[line width=0.21mm, draw=L] (100.32,89.97) .. controls (100.27,92.71) and (98.04,94.90) .. (95.30,94.91) .. controls (92.85,94.91) and (90.77,93.13) .. (90.39,90.71);
	\path[line width=0.21mm, draw=L] (90.37,89.18) .. controls (90.78,86.78) and (92.87,85.03) .. (95.30,85.05) .. controls (98.02,85.06) and (100.24,87.22) .. (100.32,89.94);
	\path[line width=0.21mm, draw=L] (95.44,94.87) .. controls (95.11,91.55) and (93.48,89.96) .. (90.61,89.97);
	\path[line width=0.21mm, draw=L] (99.26,93.94) -- (97.96,94.21) -- (98.30,92.79);
	\path[line width=0.21mm, draw=L] (100.22,89.90) -- (107.87,89.94);
	\path[line width=0.21mm, draw=L] (91.66,94.26) -- (91.31,92.97) -- (92.75,93.22);
	\path[line width=0.21mm, draw=L] (104.95,90.74) -- (103.89,89.94) -- (105.05,89.33);
	\path[line width=0.21mm, draw=L] (96.20,85.82) -- (95.14,85.02) -- (96.29,84.41);
	\path[line width=0.21mm, draw=L] (93.92,92.10) -- (93.57,90.82) -- (95.02,91.06);
	\end{tikzpicture}}
\newcommand{\GraphTwentyTwoLineOppositeSymm}{\begin{tikzpicture}[baseline=89.2mm,x=1.00mm, y=1.00mm, inner xsep=0pt, inner ysep=0pt, outer xsep=0pt, outer ysep=0pt]
	\useasboundingbox (81,83) rectangle +(30,15);
	\path[line width=0mm] (79.57,82.41) rectangle +(32.39,14.58);
	\definecolor{L}{rgb}{0,0,0}
	\path[line width=0.30mm, draw=L] (93.32,90.00);
	\draw(81.57,89.40) node[anchor=base west]{\fontsize{5.69}{6.83}\selectfont 1};
	\draw(108.96,89.43) node[anchor=base west]{\fontsize{5.69}{6.83}\selectfont 2};
	\path[line width=0.30mm, draw=L, dash pattern=on 3pt off 2pt] (92.31,90.32);
	\path[line width=0.21mm, draw=L] (95.29,94.98) .. controls (95.23,91.70) and (93.23,89.99) .. (90.36,90.00);
	\path[line width=0.21mm, draw=L] (99.43,93.91) -- (98.13,94.17) -- (98.47,92.75);
	\path[line width=0.21mm, draw=L] (91.55,94.22) -- (91.20,92.94) -- (92.64,93.19);
	\path[line width=0.21mm, draw=L] (87.45,90.79) -- (86.39,89.99) -- (87.55,89.39);
	\path[line width=0.21mm, draw=L] (96.20,85.82) -- (95.14,85.02) -- (96.29,84.41);
	\path[line width=0.21mm, draw=L] (93.76,92.16) -- (93.41,90.88) -- (94.85,91.13);
	\path[line width=0.21mm, draw=L] (95.32,89.97) circle (5.02mm);
	\path[line width=0.21mm, draw=L] (82.71,89.99) -- (90.37,89.99);
	\path[line width=0.21mm, draw=L] (105.32,90.71) -- (104.26,89.91) -- (105.42,89.30);
	\path[line width=0.21mm, draw=L] (100.37,89.95) -- (108.57,89.92);
	\end{tikzpicture}}
\newcommand{\GraphTwentyTwoLineMiddle}{\begin{tikzpicture}[baseline=89.2mm,x=1.00mm, y=1.00mm, inner xsep=0pt, inner ysep=0pt, outer xsep=0pt, outer ysep=0pt]
	\useasboundingbox (89,83) rectangle +(21,15);
	\path[line width=0mm] (87.32,82.41) rectangle +(23.78,15.19);
	\definecolor{L}{rgb}{0,0,0}
	\path[line width=0.30mm, draw=L] (93.88,91.19);
	\draw(89.32,89.30) node[anchor=base west]{\fontsize{5.69}{6.83}\selectfont 1};
	\draw(108.10,89.27) node[anchor=base west]{\fontsize{5.69}{6.83}\selectfont 2};
	\path[line width=0.30mm, draw=L, dash pattern=on 3pt off 2pt] (93.31,92.09);
	\path[line width=0.21mm, draw=L] (100.32,89.97) .. controls (100.25,92.70) and (98.03,94.88) .. (95.30,94.91) .. controls (92.78,94.93) and (90.64,93.08) .. (90.29,90.59);
	\path[line width=0.21mm, draw=L] (90.37,89.18) .. controls (90.78,86.78) and (92.87,85.03) .. (95.30,85.05) .. controls (98.02,85.06) and (100.24,87.22) .. (100.32,89.94);
	\path[line width=0.21mm, draw=L] (98.79,93.38) .. controls (96.36,91.41) and (93.84,91.59) .. (91.75,93.46);
	\path[line width=0.21mm, draw=L] (100.98,91.34) -- (99.86,92.05) -- (99.68,90.60);
	\path[line width=0.21mm, draw=L] (95.97,92.72) -- (94.85,92.00) -- (96.03,91.24);
	\path[line width=0.21mm, draw=L] (100.25,89.89) -- (108.01,89.91);
	\path[line width=0.21mm, draw=L] (90.51,93.00) -- (90.54,91.67) -- (91.85,92.32);
	\path[line width=0.21mm, draw=L] (104.99,90.70) -- (103.93,89.90) -- (105.09,89.29);
	\path[line width=0.21mm, draw=L] (96.20,85.82) -- (95.14,85.02) -- (96.29,84.41);
	\path[line width=0.21mm, draw=L] (96.05,95.61) -- (94.93,94.89) -- (96.12,94.13);
	\end{tikzpicture}}
\newcommand{\GraphTwentyTwoLineMiddleSymm}{\begin{tikzpicture}[baseline=89.2mm,x=1.00mm, y=1.00mm, inner xsep=0pt, inner ysep=0pt, outer xsep=0pt, outer ysep=0pt]
	\useasboundingbox (81,83) rectangle +(30,15);
	\path[line width=0mm] (79.33,82.35) rectangle +(32.05,15.39);
	\definecolor{L}{rgb}{0,0,0}
	\path[line width=0.30mm, draw=L] (93.88,91.19);
	\draw(81.33,89.41) node[anchor=base west]{\fontsize{5.69}{6.83}\selectfont 1};
	\draw(108.38,89.44) node[anchor=base west]{\fontsize{5.69}{6.83}\selectfont 2};
	\path[line width=0.30mm, draw=L, dash pattern=on 3pt off 2pt] (93.31,92.09);
	\path[line width=0.21mm, draw=L] (98.86,93.38) .. controls (96.62,91.40) and (93.84,91.59) .. (91.75,93.46);
	\path[line width=0.21mm, draw=L] (100.67,91.90) -- (99.54,92.61) -- (99.36,91.16);
	\path[line width=0.21mm, draw=L] (95.97,92.72) -- (94.85,92.00) -- (96.03,91.24);
	\path[line width=0.21mm, draw=L] (82.54,89.97) -- (90.30,89.98);
	\path[line width=0.21mm, draw=L] (90.43,93.00) -- (90.46,91.67) -- (91.77,92.32);
	\path[line width=0.21mm, draw=L] (87.29,90.78) -- (86.23,89.98) -- (87.38,89.37);
	\path[line width=0.21mm, draw=L] (96.15,85.76) -- (95.09,84.96) -- (96.25,84.35);
	\path[line width=0.21mm, draw=L] (96.03,95.74) -- (94.91,95.02) -- (96.09,94.26);
	\path[line width=0.21mm, draw=L] (95.26,89.97) circle (5.02mm);
	\path[line width=0.21mm, draw=L] (100.37,90.04) -- (108.13,90.06);
	\path[line width=0.21mm, draw=L] (105.12,90.85) -- (104.06,90.05) -- (105.21,89.45);
	\end{tikzpicture}}
\newcommand{\GraphTwentyTwoOpposite}{\begin{tikzpicture}[baseline=89.2mm,x=1.00mm, y=1.00mm, inner xsep=0pt, inner ysep=0pt, outer xsep=0pt, outer ysep=0pt]
	\useasboundingbox (83.50,80.5) rectangle +(18.5,19);
	\path[line width=0mm] (82.30,79.25) rectangle +(21.26,20.67);
	\definecolor{L}{rgb}{0,0,0}
	\path[line width=0.30mm, draw=L] (91.03,90.64);
	\draw(84.30,89.34) node[anchor=base west]{\fontsize{5.69}{6.83}\selectfont 1};
	\draw(100.56,89.18) node[anchor=base west]{\fontsize{5.69}{6.83}\selectfont 2};
	\path[line width=0.30mm, draw=L, dash pattern=on 3pt off 2pt] (89.33,90.82);
	\path[line width=0.21mm, draw=L] (100.94,90.94) .. controls (100.41,94.93) and (97.00,97.92) .. (92.97,97.91) .. controls (89.03,97.91) and (85.68,95.04) .. (85.07,91.15);
	\path[line width=0.21mm, draw=L] (85.03,88.68) .. controls (85.69,84.82) and (89.05,82.01) .. (92.97,82.03) .. controls (96.83,82.05) and (100.13,84.84) .. (100.79,88.65);
	\path[line width=0.21mm, draw=L] (93.24,97.93) .. controls (93.23,93.71) and (91.00,90.57) .. (86.09,90.06) .. controls (85.94,90.05) and (85.52,90.12) .. (85.58,89.97);
	\path[line width=0.21mm, draw=L] (99.97,95.29) -- (98.67,95.55) -- (99.01,94.13);
	\path[line width=0.21mm, draw=L] (91.06,93.33) -- (90.92,92.01) -- (92.25,92.45);
	\path[line width=0.21mm, draw=L] (87.26,96.64) -- (87.00,95.29) -- (88.39,95.65);
	\path[line width=0.21mm, draw=L] (93.79,82.73) -- (92.68,82.01) -- (93.87,81.25);
	\end{tikzpicture}}
\newcommand{\GraphTwentyTwoOppositeLeft}{\begin{tikzpicture}[baseline=89.2mm,x=1.00mm, y=1.00mm, inner xsep=0pt, inner ysep=0pt, outer xsep=0pt, outer ysep=0pt]
	\useasboundingbox (81,82.5) rectangle +(21,16);
	\path[line width=0mm] (79.57,82.41) rectangle +(23.33,14.50);
	\definecolor{L}{rgb}{0,0,0}
	\path[line width=0.30mm, draw=L] (93.32,90.00);
	\draw(81.57,89.40) node[anchor=base west]{\fontsize{5.69}{6.83}\selectfont 1};
	\draw(99.89,89.35) node[anchor=base west]{\fontsize{5.69}{6.83}\selectfont 2};
	\path[line width=0.30mm, draw=L, dash pattern=on 3pt off 2pt] (92.31,90.32);
	\path[line width=0.21mm, draw=L] (100.17,91.03) .. controls (99.67,93.31) and (97.64,94.93) .. (95.30,94.91) .. controls (92.58,94.88) and (90.38,92.70) .. (90.32,89.99);
	\path[line width=0.21mm, draw=L] (90.32,89.99) .. controls (90.37,87.27) and (92.58,85.08) .. (95.30,85.05) .. controls (97.67,85.02) and (99.73,86.67) .. (100.21,88.99);
	\path[line width=0.21mm, draw=L] (95.17,94.87) .. controls (94.84,91.55) and (93.21,89.96) .. (90.35,89.97);
	\path[line width=0.21mm, draw=L] (99.26,93.94) -- (97.96,94.21) -- (98.30,92.79);
	\path[line width=0.21mm, draw=L] (82.72,89.95) -- (90.37,89.99);
	\path[line width=0.21mm, draw=L] (91.66,94.26) -- (91.31,92.97) -- (92.75,93.22);
	\path[line width=0.21mm, draw=L] (87.45,90.79) -- (86.39,89.99) -- (87.55,89.39);
	\path[line width=0.21mm, draw=L] (96.20,85.82) -- (95.14,85.02) -- (96.29,84.41);
	\path[line width=0.21mm, draw=L] (93.76,92.16) -- (93.41,90.88) -- (94.85,91.13);
	\end{tikzpicture}}
\newcommand{\GraphTwentyTwoLineMiddleLeft}{\begin{tikzpicture}[baseline=89.2mm,x=1.00mm, y=1.00mm, inner xsep=0pt, inner ysep=0pt, outer xsep=0pt, outer ysep=0pt]
	\useasboundingbox (81,82.5) rectangle +(21,16);
	\path[line width=0mm] (79.33,82.41) rectangle +(23.58,15.19);
	\definecolor{L}{rgb}{0,0,0}
	\path[line width=0.30mm, draw=L] (93.88,91.19);
	\draw(81.33,89.41) node[anchor=base west]{\fontsize{5.69}{6.83}\selectfont 1};
	\draw(99.90,89.40) node[anchor=base west]{\fontsize{5.69}{6.83}\selectfont 2};
	\path[line width=0.30mm, draw=L, dash pattern=on 3pt off 2pt] (93.31,92.09);
	\path[line width=0.21mm, draw=L] (100.18,91.07) .. controls (99.66,93.33) and (97.63,94.93) .. (95.30,94.91) .. controls (92.50,94.88) and (90.25,92.59) .. (90.26,89.79);
	\path[line width=0.21mm, draw=L] (90.25,89.97) .. controls (90.35,87.24) and (92.57,85.07) .. (95.30,85.05) .. controls (97.68,85.02) and (99.74,86.67) .. (100.24,89.00);
	\path[line width=0.21mm, draw=L] (98.79,93.38) .. controls (96.36,91.41) and (93.84,91.59) .. (91.75,93.46);
	\path[line width=0.21mm, draw=L] (100.67,91.90) -- (99.54,92.61) -- (99.36,91.16);
	\path[line width=0.21mm, draw=L] (95.97,92.72) -- (94.85,92.00) -- (96.03,91.24);
	\path[line width=0.21mm, draw=L] (82.54,89.97) -- (90.30,89.98);
	\path[line width=0.21mm, draw=L] (90.51,93.00) -- (90.54,91.67) -- (91.85,92.32);
	\path[line width=0.21mm, draw=L] (87.29,90.78) -- (86.23,89.98) -- (87.38,89.37);
	\path[line width=0.21mm, draw=L] (96.20,85.82) -- (95.14,85.02) -- (96.29,84.41);
	\path[line width=0.21mm, draw=L] (96.05,95.61) -- (94.93,94.89) -- (96.12,94.13);
	\end{tikzpicture}}
\newcommand{\GraphTwentyTwoMiddle}{\begin{tikzpicture}[baseline=89.2mm,x=1.00mm, y=1.00mm, inner xsep=0pt, inner ysep=0pt, outer xsep=0pt, outer ysep=0pt]
	\useasboundingbox (83.50,80.5) rectangle +(18.5,19);
	\path[line width=0mm] (82.30,79.25) rectangle +(21.26,21.32);
	\draw(84.30,89.34) node[anchor=base west]{\fontsize{5.69}{6.83}\selectfont 1};
	\draw(100.56,89.18) node[anchor=base west]{\fontsize{5.69}{6.83}\selectfont 2};
	\definecolor{L}{rgb}{0,0,0}
	\path[line width=0.30mm, draw=L, dash pattern=on 3pt off 2pt] (89.33,90.82);
	\path[line width=0.21mm, draw=L] (100.94,90.94) .. controls (100.41,94.93) and (97.00,97.92) .. (92.97,97.91) .. controls (89.03,97.91) and (85.68,95.04) .. (85.07,91.15);
	\path[line width=0.21mm, draw=L] (85.03,88.68) .. controls (85.69,84.82) and (89.05,82.01) .. (92.97,82.03) .. controls (96.83,82.05) and (100.13,84.84) .. (100.79,88.65);
	\path[line width=0.21mm, draw=L] (98.59,95.65) .. controls (95.39,92.71) and (91.07,92.87) .. (87.60,95.74);
	\path[line width=0.21mm, draw=L] (101.17,93.14) -- (100.02,93.81) -- (99.88,92.36);
	\path[line width=0.21mm, draw=L] (85.55,94.65) -- (85.73,93.29) -- (86.93,94.07);
	\path[line width=0.21mm, draw=L] (93.79,82.73) -- (92.68,82.01) -- (93.87,81.25);
	\path[line width=0.21mm, draw=L] (93.73,98.57) -- (92.61,97.85) -- (93.80,97.10);
	\path[line width=0.21mm, draw=L] (93.79,94.33) -- (92.68,93.61) -- (93.87,92.85);
	\end{tikzpicture}}
\newcommand{\GraphTwentyTwoWithLine}{\begin{tikzpicture}[baseline=89.2mm,x=1.00mm, y=1.00mm, inner xsep=0pt, inner ysep=0pt, outer xsep=0pt, outer ysep=0pt]
	\useasboundingbox (83.50,80.5) rectangle +(18.5,19);
	\path[line width=0mm] (82.30,79.99) rectangle +(21.26,19.92);
	\draw(84.30,89.34) node[anchor=base west]{\fontsize{5.69}{6.83}\selectfont 1};
	\draw(100.56,89.18) node[anchor=base west]{\fontsize{5.69}{6.83}\selectfont 2};
	\definecolor{L}{rgb}{0,0,0}
	\path[line width=0.30mm, draw=L, dash pattern=on 3pt off 2pt] (89.33,90.82);
	\path[line width=0.21mm, draw=L] (100.94,90.94) .. controls (100.41,94.93) and (97.00,97.92) .. (92.97,97.91) .. controls (89.03,97.91) and (85.68,95.04) .. (85.07,91.15);
	\path[line width=0.21mm, draw=L] (85.03,88.68) .. controls (85.69,84.82) and (89.05,82.01) .. (92.97,82.03) .. controls (96.83,82.05) and (100.13,84.84) .. (100.79,88.65);
	\path[line width=0.21mm, draw=L] (99.67,95.50) -- (98.36,95.77) -- (98.71,94.34);
	\path[line width=0.21mm, draw=L] (87.26,96.64) -- (87.00,95.29) -- (88.39,95.65);
	\path[line width=0.21mm, draw=L] (92.89,97.91) -- (92.89,81.99);
	\path[line width=0.21mm, draw=L] (98.44,85.36) -- (98.25,84.05) -- (99.65,84.47);
	\path[line width=0.21mm, draw=L] (88.49,84.30) -- (87.13,84.51) -- (87.54,83.13);
	\path[line width=0.21mm, draw=L] (92.22,90.52) -- (92.86,89.34) -- (93.64,90.47);
	\end{tikzpicture}}
\newcommand{\SolidLineLoop}{\begin{tikzpicture}[baseline=89.9mm,x=1.00mm, y=1.00mm, inner xsep=0pt, inner ysep=0pt, outer xsep=0pt, outer ysep=0pt]
	\useasboundingbox (84.5,85.5) rectangle +(17,11);
	\path[line width=0mm] (82.83,84.82) rectangle +(19.69,12.47);
	\definecolor{L}{rgb}{0,0,0}
	\path[line width=0.30mm, draw=L] (93.00,92.90);
	\path[line width=0.30mm, draw=L, dash pattern=on 3pt off 2pt] (93.34,94.58);
	\path[line width=0.21mm, draw=L] (96.67,91.78) -- (95.58,91.03) -- (96.67,90.42);
	\path[line width=0.21mm, draw=L] (92.50,91.04) -- (99.19,91.05);
	\draw(84.83,90.42) node[anchor=base west]{\fontsize{5.69}{6.83}\selectfont 1};
	\draw(99.52,90.43) node[anchor=base west]{\fontsize{5.69}{6.83}\selectfont 2};
	\path[line width=0.21mm, draw=L] (85.40,90.04) .. controls (85.59,89.56) and (85.86,89.11) .. (86.21,88.73) .. controls (87.07,87.80) and (88.31,87.29) .. (89.56,87.47) .. controls (90.84,87.65) and (91.90,88.53) .. (92.34,89.74);
	\path[line width=0.21mm, draw=L] (92.34,89.74) .. controls (92.79,90.94) and (92.56,92.28) .. (91.76,93.28) .. controls (90.95,94.27) and (89.67,94.76) .. (88.40,94.54) .. controls (87.15,94.33) and (86.11,93.45) .. (85.63,92.27) .. controls (85.60,92.18) and (85.56,92.09) .. (85.53,92.00);
	\path[line width=0.21mm, draw=L] (89.65,95.29) -- (88.60,94.59) -- (89.69,93.97);
	\path[line width=0.21mm, draw=L] (89.91,88.18) -- (88.82,87.43) -- (89.91,86.82);
	\end{tikzpicture}}
\newcommand{\SolidLineLoopMiddle}{\begin{tikzpicture}[baseline=89.9mm,x=1.00mm, y=1.00mm, inner xsep=0pt, inner ysep=0pt, outer xsep=0pt, outer ysep=0pt]
	\useasboundingbox (84.5,85.5) rectangle +(17,11);
	\path[line width=0mm] (82.83,84.82) rectangle +(19.69,12.43);
	\definecolor{L}{rgb}{0,0,0}
	\path[line width=0.30mm, draw=L] (93.00,92.90);
	\path[line width=0.30mm, draw=L, dash pattern=on 3pt off 2pt] (93.34,94.58);
	\path[line width=0.21mm, draw=L] (96.67,91.78) -- (95.58,91.03) -- (96.67,90.42);
	\path[line width=0.21mm, draw=L] (86.04,91.04) -- (99.19,91.05);
	\draw(84.83,90.42) node[anchor=base west]{\fontsize{5.69}{6.83}\selectfont 1};
	\draw(99.52,90.43) node[anchor=base west]{\fontsize{5.69}{6.83}\selectfont 2};
	\path[line width=0.21mm, draw=L] (85.40,90.04) .. controls (85.59,89.56) and (85.86,89.11) .. (86.21,88.73) .. controls (87.07,87.80) and (88.31,87.29) .. (89.56,87.47) .. controls (90.84,87.65) and (91.90,88.53) .. (92.34,89.74);
	\path[line width=0.21mm, draw=L] (92.34,89.74) .. controls (92.79,90.94) and (92.56,92.28) .. (91.76,93.28) .. controls (90.95,94.27) and (89.67,94.76) .. (88.40,94.54) .. controls (87.15,94.33) and (86.11,93.45) .. (85.63,92.27) .. controls (85.60,92.18) and (85.56,92.09) .. (85.53,92.00);
	\path[line width=0.21mm, draw=L] (89.79,95.24) -- (88.74,94.54) -- (89.83,93.93);
	\path[line width=0.21mm, draw=L] (89.91,88.18) -- (88.82,87.43) -- (89.91,86.82);
	\path[line width=0.21mm, draw=L] (89.81,91.78) -- (88.72,91.03) -- (89.81,90.42);
	\end{tikzpicture}}
\newcommand{\SolidLineLoopMiddleLeft}{\begin{tikzpicture}[baseline=89.9mm,x=1.00mm, y=1.00mm, inner xsep=0pt, inner ysep=0pt, outer xsep=0pt, outer ysep=0pt]
	\useasboundingbox (78,85.5) rectangle +(17,11);
	\path[line width=0mm] (76.51,84.82) rectangle +(19.56,12.43);
	\definecolor{L}{rgb}{0,0,0}
	\path[line width=0.30mm, draw=L] (85.73,89.15);
	\path[line width=0.30mm, draw=L, dash pattern=on 3pt off 2pt] (85.38,87.47);
	\path[line width=0.21mm, draw=L] (83.94,91.70) -- (82.85,90.96) -- (83.94,90.34);
	\path[line width=0.21mm, draw=L] (79.58,90.91) -- (92.73,90.92);
	\draw(78.51,90.42) node[anchor=base west]{\fontsize{5.69}{6.83}\selectfont 1};
	\draw(93.08,90.33) node[anchor=base west]{\fontsize{5.69}{6.83}\selectfont 2};
	\path[line width=0.21mm, draw=L] (93.35,91.96) .. controls (93.16,92.44) and (92.89,92.89) .. (92.55,93.27) .. controls (91.70,94.21) and (90.46,94.73) .. (89.20,94.56) .. controls (87.93,94.38) and (86.86,93.51) .. (86.42,92.30);
	\path[line width=0.21mm, draw=L] (86.42,92.30) .. controls (85.96,91.11) and (86.17,89.76) .. (86.98,88.76) .. controls (87.78,87.76) and (89.06,87.27) .. (90.32,87.48) .. controls (91.58,87.68) and (92.62,88.55) .. (93.11,89.73) .. controls (93.14,89.82) and (93.18,89.91) .. (93.21,90.00);
	\path[line width=0.21mm, draw=L] (90.50,95.24) -- (89.45,94.54) -- (90.54,93.93);
	\path[line width=0.21mm, draw=L] (90.62,88.18) -- (89.53,87.43) -- (90.62,86.82);
	\path[line width=0.21mm, draw=L] (90.44,91.62) -- (89.35,90.88) -- (90.44,90.26);
	\end{tikzpicture}}
\newcommand{\SolidLineLoopMiddleSymm}{\begin{tikzpicture}[baseline=89.9mm,x=1.00mm, y=1.00mm, inner xsep=0pt, inner ysep=0pt, outer xsep=0pt, outer ysep=0pt]
	\useasboundingbox (78,85.5) rectangle +(24.5,11);
	\path[line width=0mm] (76.96,84.80) rectangle +(26.91,12.43);
	\definecolor{L}{rgb}{0,0,0}
	\path[line width=0.30mm, draw=L] (86.18,89.12);
	\path[line width=0.30mm, draw=L, dash pattern=on 3pt off 2pt] (85.83,87.45);
	\path[line width=0.21mm, draw=L] (84.36,91.64) -- (83.26,90.90) -- (84.36,90.28);
	\path[line width=0.21mm, draw=L] (80.03,90.89) -- (100.52,90.92);
	\draw(78.96,90.40) node[anchor=base west]{\fontsize{5.69}{6.83}\selectfont 1};
	\draw(100.87,90.39) node[anchor=base west]{\fontsize{5.69}{6.83}\selectfont 2};
	\path[line width=0.21mm, draw=L] (90.94,95.22) -- (89.90,94.52) -- (90.99,93.91);
	\path[line width=0.21mm, draw=L] (91.07,88.16) -- (89.98,87.41) -- (91.07,86.80);
	\path[line width=0.21mm, draw=L] (90.91,91.66) -- (89.82,90.91) -- (90.91,90.30);
	\path[line width=0.21mm, draw=L] (90.24,90.98) circle (3.56mm);
	\path[line width=0.21mm, draw=L] (98.05,91.60) -- (96.96,90.86) -- (98.05,90.24);
	\end{tikzpicture}}
\newcommand{\SolidLineLoopTwo}{\begin{tikzpicture}[baseline=89.9mm,x=1.00mm, y=1.00mm, inner xsep=0pt, inner ysep=0pt, outer xsep=0pt, outer ysep=0pt]
	\useasboundingbox (84.5,85.5) rectangle +(23.5,11);
	\path[line width=0mm] (82.83,84.82) rectangle +(26.57,12.47);
	\definecolor{L}{rgb}{0,0,0}
	\path[line width=0.30mm, draw=L] (93.00,92.90);
	\path[line width=0.30mm, draw=L, dash pattern=on 3pt off 2pt] (93.34,94.58);
	\path[line width=0.21mm, draw=L] (96.67,91.78) -- (95.58,91.03) -- (96.67,90.42);
	\path[line width=0.21mm, draw=L] (92.50,91.04) -- (99.19,91.05);
	\draw(84.83,90.42) node[anchor=base west]{\fontsize{5.69}{6.83}\selectfont 1};
	\path[line width=0.21mm, draw=L] (85.40,90.04) .. controls (85.59,89.56) and (85.86,89.11) .. (86.21,88.73) .. controls (87.07,87.80) and (88.31,87.29) .. (89.56,87.47) .. controls (90.84,87.65) and (91.90,88.53) .. (92.34,89.74);
	\path[line width=0.21mm, draw=L] (92.34,89.74) .. controls (92.79,90.94) and (92.56,92.28) .. (91.76,93.28) .. controls (90.95,94.27) and (89.67,94.76) .. (88.40,94.54) .. controls (87.15,94.33) and (86.11,93.45) .. (85.63,92.27) .. controls (85.60,92.18) and (85.56,92.09) .. (85.53,92.00);
	\path[line width=0.21mm, draw=L] (89.65,95.29) -- (88.60,94.59) -- (89.69,93.97);
	\path[line width=0.21mm, draw=L] (89.91,88.18) -- (88.82,87.43) -- (89.91,86.82);
	\path[line width=0.30mm, draw=L] (102.87,94.82);
	\draw(106.40,90.46) node[anchor=base west]{\fontsize{5.69}{6.83}\selectfont 2};
	\path[line width=0.21mm, draw=L] (106.32,92.17) .. controls (106.11,92.64) and (105.81,93.07) .. (105.45,93.43) .. controls (104.54,94.32) and (103.28,94.77) .. (102.03,94.53) .. controls (100.77,94.28) and (99.75,93.35) .. (99.38,92.12);
	\path[line width=0.21mm, draw=L] (99.38,92.12) .. controls (98.99,90.90) and (99.28,89.57) .. (100.13,88.62) .. controls (100.99,87.67) and (102.29,87.24) .. (103.55,87.52) .. controls (104.79,87.80) and (105.78,88.72) .. (106.20,89.93) .. controls (106.23,90.02) and (106.26,90.11) .. (106.29,90.20);
	\path[line width=0.21mm, draw=L] (103.47,95.28) -- (102.38,94.53) -- (103.47,93.92);
	\path[line width=0.21mm, draw=L] (103.60,88.19) -- (102.51,87.44) -- (103.60,86.82);
	\end{tikzpicture}}
\newcommand{\SolidLineLoopTwoLine}{\begin{tikzpicture}[baseline=89.9mm,x=1.00mm, y=1.00mm, inner xsep=0pt, inner ysep=0pt, outer xsep=0pt, outer ysep=0pt]
	\useasboundingbox (84.5,85.5) rectangle +(31,11);
	\path[line width=0mm] (82.83,84.82) rectangle +(33.75,12.55);
	\definecolor{L}{rgb}{0,0,0}
	\path[line width=0.30mm, draw=L] (93.00,92.90);
	\path[line width=0.30mm, draw=L, dash pattern=on 3pt off 2pt] (93.34,94.58);
	\path[line width=0.21mm, draw=L] (96.67,91.78) -- (95.58,91.03) -- (96.67,90.42);
	\path[line width=0.21mm, draw=L] (92.50,91.04) -- (99.19,91.05);
	\draw(84.83,90.42) node[anchor=base west]{\fontsize{5.69}{6.83}\selectfont 1};
	\path[line width=0.21mm, draw=L] (85.40,90.04) .. controls (85.59,89.56) and (85.86,89.11) .. (86.21,88.73) .. controls (87.07,87.80) and (88.31,87.29) .. (89.56,87.47) .. controls (90.84,87.65) and (91.90,88.53) .. (92.34,89.74);
	\path[line width=0.21mm, draw=L] (92.34,89.74) .. controls (92.79,90.94) and (92.56,92.28) .. (91.76,93.28) .. controls (90.95,94.27) and (89.67,94.76) .. (88.40,94.54) .. controls (87.15,94.33) and (86.11,93.45) .. (85.63,92.27) .. controls (85.60,92.18) and (85.56,92.09) .. (85.53,92.00);
	\path[line width=0.21mm, draw=L] (89.65,95.29) -- (88.60,94.59) -- (89.69,93.97);
	\path[line width=0.21mm, draw=L] (89.91,88.18) -- (88.82,87.43) -- (89.91,86.82);
	\path[line width=0.30mm, draw=L] (102.87,94.82);
	\path[line width=0.21mm, draw=L] (103.52,95.37) -- (102.42,94.62) -- (103.52,94.01);
	\path[line width=0.21mm, draw=L] (103.60,88.19) -- (102.51,87.44) -- (103.60,86.82);
	\draw(113.57,90.43) node[anchor=base west]{\fontsize{5.69}{6.83}\selectfont 2};
	\path[line width=0.21mm, draw=L] (102.78,91.04) circle (3.62mm);
	\path[line width=0.21mm, draw=L] (110.63,91.79) -- (109.54,91.04) -- (110.63,90.43);
	\path[line width=0.21mm, draw=L] (106.46,91.05) -- (113.15,91.06);
	\end{tikzpicture}}
\newcommand{\SolidLineLoopTwoLineLeft}{\begin{tikzpicture}[baseline=89.9mm,x=1.00mm, y=1.00mm, inner xsep=0pt, inner ysep=0pt, outer xsep=0pt, outer ysep=0pt]
	\useasboundingbox (77,85.5) rectangle +(31,11);
	\path[line width=0mm] (75.76,84.80) rectangle +(33.29,12.57);
	\definecolor{L}{rgb}{0,0,0}
	\path[line width=0.30mm, draw=L] (93.00,92.90);
	\path[line width=0.30mm, draw=L, dash pattern=on 3pt off 2pt] (93.34,94.58);
	\path[line width=0.21mm, draw=L] (96.67,91.78) -- (95.58,91.03) -- (96.67,90.42);
	\path[line width=0.21mm, draw=L] (92.50,91.04) -- (99.19,91.05);
	\draw(77.76,90.42) node[anchor=base west]{\fontsize{5.69}{6.83}\selectfont 1};
	\path[line width=0.21mm, draw=L] (89.66,95.23) -- (88.62,94.53) -- (89.71,93.92);
	\path[line width=0.21mm, draw=L] (89.87,88.16) -- (88.78,87.42) -- (89.87,86.80);
	\path[line width=0.30mm, draw=L] (102.87,94.82);
	\path[line width=0.21mm, draw=L] (103.52,95.37) -- (102.43,94.62) -- (103.52,94.01);
	\path[line width=0.21mm, draw=L] (103.60,88.19) -- (102.51,87.44) -- (103.60,86.82);
	\draw(106.05,90.34) node[anchor=base west]{\fontsize{5.69}{6.83}\selectfont 2};
	\path[line width=0.21mm, draw=L] (106.29,91.98) .. controls (105.87,93.57) and (104.42,94.67) .. (102.78,94.66) .. controls (100.78,94.64) and (99.16,93.03) .. (99.15,91.04) .. controls (99.15,89.04) and (100.78,87.43) .. (102.78,87.41) .. controls (104.39,87.40) and (105.82,88.46) .. (106.26,90.02);
	\path[line width=0.21mm, draw=L] (82.94,91.71) -- (81.85,90.96) -- (82.94,90.35);
	\path[line width=0.21mm, draw=L] (78.77,90.97) -- (85.46,90.98);
	\path[line width=0.21mm, draw=L] (88.94,90.97) circle (3.57mm);
	\end{tikzpicture}}
\newcommand{\SolidLineLoopTwoLineSymm}{\begin{tikzpicture}[baseline=89.9mm,x=1.00mm, y=1.00mm, inner xsep=0pt, inner ysep=0pt, outer xsep=0pt, outer ysep=0pt]
	\useasboundingbox (77,85.5) rectangle +(38.5,11);
	\path[line width=0mm] (75.76,84.80) rectangle +(40.81,12.57);
	\definecolor{L}{rgb}{0,0,0}
	\path[line width=0.30mm, draw=L] (93.00,92.90);
	\path[line width=0.30mm, draw=L, dash pattern=on 3pt off 2pt] (93.34,94.58);
	\path[line width=0.21mm, draw=L] (96.67,91.78) -- (95.58,91.03) -- (96.67,90.42);
	\path[line width=0.21mm, draw=L] (92.50,91.04) -- (99.19,91.05);
	\path[line width=0.30mm, draw=L] (102.87,94.82);
	\path[line width=0.21mm, draw=L] (103.52,95.37) -- (102.43,94.62) -- (103.52,94.01);
	\path[line width=0.21mm, draw=L] (103.60,88.19) -- (102.51,87.44) -- (103.60,86.82);
	\draw(113.57,90.43) node[anchor=base west]{\fontsize{5.69}{6.83}\selectfont 2};
	\path[line width=0.21mm, draw=L] (102.78,91.04) circle (3.62mm);
	\path[line width=0.21mm, draw=L] (110.63,91.79) -- (109.54,91.04) -- (110.63,90.43);
	\path[line width=0.21mm, draw=L] (106.46,91.05) -- (113.15,91.06);
	\path[line width=0.30mm, draw=L] (93.00,92.90);
	\path[line width=0.30mm, draw=L, dash pattern=on 3pt off 2pt] (93.34,94.58);
	\draw(77.76,90.42) node[anchor=base west]{\fontsize{5.69}{6.83}\selectfont 1};
	\path[line width=0.21mm, draw=L] (89.66,95.23) -- (88.62,94.53) -- (89.71,93.92);
	\path[line width=0.21mm, draw=L] (89.87,88.16) -- (88.78,87.42) -- (89.87,86.80);
	\path[line width=0.21mm, draw=L] (82.94,91.71) -- (81.85,90.96) -- (82.94,90.35);
	\path[line width=0.21mm, draw=L] (78.77,90.97) -- (85.46,90.98);
	\path[line width=0.21mm, draw=L] (88.94,90.97) circle (3.57mm);
	\end{tikzpicture}}
\newcommand{\SolidLineLoopReverse}{\begin{tikzpicture}[baseline=89.9mm,x=1.00mm, y=1.00mm, inner xsep=0pt, inner ysep=0pt, outer xsep=0pt, outer ysep=0pt]
	\useasboundingbox (78,85.5) rectangle +(17,11);
	\path[line width=0mm] (76.35,84.90) rectangle +(20.03,12.45);
	\definecolor{L}{rgb}{0,0,0}
	\path[line width=0.30mm, draw=L] (89.85,94.90);
	\path[line width=0.30mm, draw=L, dash pattern=on 3pt off 2pt] (85.59,87.32);
	\path[line width=0.21mm, draw=L] (83.63,91.71) -- (82.54,90.97) -- (83.63,90.35);
	\path[line width=0.21mm, draw=L] (79.45,90.98) -- (86.15,90.98);
	\draw(78.35,90.42) node[anchor=base west]{\fontsize{5.69}{6.83}\selectfont 1};
	\draw(93.38,90.53) node[anchor=base west]{\fontsize{5.69}{6.83}\selectfont 2};
	\path[line width=0.21mm, draw=L] (93.29,92.25) .. controls (93.08,92.72) and (92.79,93.15) .. (92.42,93.51) .. controls (91.52,94.40) and (90.25,94.85) .. (89.01,94.61) .. controls (87.75,94.36) and (86.73,93.43) .. (86.36,92.20);
	\path[line width=0.21mm, draw=L] (86.36,92.20) .. controls (85.97,90.98) and (86.25,89.65) .. (87.11,88.70) .. controls (87.97,87.75) and (89.27,87.32) .. (90.52,87.60) .. controls (91.77,87.88) and (92.76,88.80) .. (93.18,90.01) .. controls (93.21,90.10) and (93.24,90.19) .. (93.26,90.28);
	\path[line width=0.21mm, draw=L] (90.45,95.36) -- (89.36,94.61) -- (90.45,94.00);
	\path[line width=0.21mm, draw=L] (90.58,88.26) -- (89.49,87.52) -- (90.58,86.90);
	\end{tikzpicture}}
\newcommand{\SolidLineLoopDouble}{\begin{tikzpicture}[baseline=89.9mm,x=1.00mm, y=1.00mm, inner xsep=0pt, inner ysep=0pt, outer xsep=0pt, outer ysep=0pt]
	\useasboundingbox (78,85.5) rectangle +(24,11);
	\path[line width=0mm] (76.35,84.97) rectangle +(26.69,12.20);
	\definecolor{L}{rgb}{0,0,0}
	\path[line width=0.30mm, draw=L] (89.85,94.90);
	\path[line width=0.30mm, draw=L, dash pattern=on 3pt off 2pt] (85.59,87.32);
	\path[line width=0.21mm, draw=L] (83.63,91.71) -- (82.54,90.97) -- (83.63,90.35);
	\path[line width=0.21mm, draw=L] (79.45,90.98) -- (86.15,90.98);
	\draw(78.35,90.42) node[anchor=base west]{\fontsize{5.69}{6.83}\selectfont 1};
	\draw(100.04,90.43) node[anchor=base west]{\fontsize{5.69}{6.83}\selectfont 2};
	\path[line width=0.21mm, draw=L] (90.34,95.16) -- (89.25,94.42) -- (90.34,93.80);
	\path[line width=0.21mm, draw=L] (90.52,88.33) -- (89.43,87.58) -- (90.52,86.97);
	\path[line width=0.21mm, draw=L] (89.63,91.04) circle (3.44mm);
	\path[line width=0.21mm, draw=L] (97.20,91.79) -- (96.11,91.04) -- (97.20,90.43);
	\path[line width=0.21mm, draw=L] (93.03,91.05) -- (99.72,91.06);
	\end{tikzpicture}}
\newcommand{\GraphTwoSolid}{\begin{tikzpicture}[baseline=89.9mm,x=1.00mm, y=1.00mm, inner xsep=0pt, inner ysep=0pt, outer xsep=0pt, outer ysep=0pt]
	\useasboundingbox (83.50,86.8) rectangle +(18.5,7.9);
	\path[line width=0mm] (81.55,85.04) rectangle +(22.11,11.54);
	\definecolor{L}{rgb}{0,0,0}
	\path[line width=0.21mm, draw=L] (85.00,90.00) .. controls (89.50,87.00) and (95.50,87.00) .. (100.00,90.00);
	\path[line width=0.21mm, draw=L] (85.04,91.54) .. controls (89.54,94.54) and (95.54,94.54) .. (100.04,91.54);
	\path[line width=0.30mm, draw=L] (93.00,92.90);
	\draw(83.55,90.16) node[anchor=base west]{\fontsize{5.69}{6.83}\selectfont 1};
	\draw(100.67,90.18) node[anchor=base west]{\fontsize{5.69}{6.83}\selectfont 2};
	\path[line width=0.30mm, draw=L, dash pattern=on 3pt off 2pt] (93.34,94.58);
	\path[line width=0.21mm, draw=L] (93.63,94.42) -- (92.46,93.79) -- (93.63,93.06);
	\path[line width=0.21mm, draw=L] (93.71,88.40) -- (92.54,87.78) -- (93.71,87.04);
	\end{tikzpicture}}
\newcommand{\GraphTwoSolidDouble}{\begin{tikzpicture}[baseline=89.9mm,x=1.00mm, y=1.00mm, inner xsep=0pt, inner ysep=0pt, outer xsep=0pt, outer ysep=0pt]
	\useasboundingbox (90,87) rectangle +(21.5,7.9);
	\path[line width=0mm] (88.74,86.60) rectangle +(24.04,10.12);
	\definecolor{L}{rgb}{0,0,0}
	\path[line width=0.21mm, draw=L] (91.80,90.52) .. controls (94.34,88.83) and (98.88,88.77) .. (100.73,90.98);
	\path[line width=0.21mm, draw=L] (91.82,91.39) .. controls (94.36,93.08) and (98.48,93.29) .. (100.69,90.98);
	\path[line width=0.30mm, draw=L] (96.31,92.16);
	\draw(90.74,90.25) node[anchor=base west]{\fontsize{5.69}{6.83}\selectfont 1};
	\path[line width=0.30mm, draw=L, dash pattern=on 3pt off 2pt] (96.51,93.10);
	\path[line width=0.21mm, draw=L] (97.20,93.35) -- (96.03,92.73) -- (97.20,91.99);
	\path[line width=0.21mm, draw=L] (97.10,89.97) -- (95.93,89.35) -- (97.10,88.61);
	\path[line width=0.21mm, draw=L] (100.70,90.97) .. controls (103.18,88.82) and (106.86,88.90) .. (109.40,90.60);
	\path[line width=0.21mm, draw=L] (100.70,90.94) .. controls (103.12,93.36) and (106.88,93.16) .. (109.42,91.46);
	\path[line width=0.30mm, draw=L] (105.45,92.23);
	\draw(109.78,90.27) node[anchor=base west]{\fontsize{5.69}{6.83}\selectfont 2};
	\path[line width=0.30mm, draw=L, dash pattern=on 3pt off 2pt] (109.53,94.72);
	\path[line width=0.21mm, draw=L] (106.18,93.41) -- (105.01,92.78) -- (106.18,92.05);
	\path[line width=0.21mm, draw=L] (106.32,89.96) -- (105.15,89.33) -- (106.32,88.60);
	\end{tikzpicture}}
\newcommand{\GraphTwoSolidDoubleLine}{\begin{tikzpicture}[baseline=89.9mm,x=1.00mm, y=1.00mm, inner xsep=0pt, inner ysep=0pt, outer xsep=0pt, outer ysep=0pt]
	\useasboundingbox (90,87) rectangle +(28,7.9);
	\path[line width=0mm] (88.74,86.60) rectangle +(30.72,10.12);
	\definecolor{L}{rgb}{0,0,0}
	\path[line width=0.21mm, draw=L] (91.80,90.52) .. controls (94.34,88.83) and (98.88,88.77) .. (100.73,90.98);
	\path[line width=0.21mm, draw=L] (91.82,91.39) .. controls (94.36,93.08) and (98.48,93.29) .. (100.69,90.98);
	\path[line width=0.30mm, draw=L] (96.31,92.16);
	\draw(90.74,90.25) node[anchor=base west]{\fontsize{5.69}{6.83}\selectfont 1};
	\path[line width=0.30mm, draw=L, dash pattern=on 3pt off 2pt] (96.51,93.10);
	\path[line width=0.21mm, draw=L] (97.20,93.35) -- (96.03,92.73) -- (97.20,91.99);
	\path[line width=0.21mm, draw=L] (97.10,89.97) -- (95.93,89.35) -- (97.10,88.61);
	\path[line width=0.21mm, draw=L] (100.70,90.97) .. controls (103.18,88.82) and (107.38,88.77) .. (109.92,91.01);
	\path[line width=0.21mm, draw=L] (100.70,90.94) .. controls (103.12,93.36) and (107.83,93.36) .. (109.95,91.03);
	\path[line width=0.30mm, draw=L] (105.45,92.23);
	\draw(116.46,90.42) node[anchor=base west]{\fontsize{5.69}{6.83}\selectfont 2};
	\path[line width=0.30mm, draw=L, dash pattern=on 3pt off 2pt] (109.53,94.72);
	\path[line width=0.21mm, draw=L] (106.18,93.41) -- (105.01,92.78) -- (106.18,92.05);
	\path[line width=0.21mm, draw=L] (106.32,89.96) -- (105.15,89.33) -- (106.32,88.60);
	\path[line width=0.21mm, draw=L] (109.91,91.00) -- (116.20,91.00);
	\path[line width=0.21mm, draw=L] (114.01,91.64) -- (112.84,91.02) -- (114.01,90.28);
	\end{tikzpicture}}
\newcommand{\GraphTwoSolidDoubleLineLeft}{\begin{tikzpicture}[baseline=89.9mm,x=1.00mm, y=1.00mm, inner xsep=0pt, inner ysep=0pt, outer xsep=0pt, outer ysep=0pt]
	\useasboundingbox (83,87) rectangle +(29,7.9);
	\path[line width=0mm] (82.03,86.52) rectangle +(31.41,10.19);
	\definecolor{L}{rgb}{0,0,0}
	\path[line width=0.21mm, draw=L] (91.30,90.99) .. controls (93.33,88.45) and (98.72,88.85) .. (100.73,90.98);
	\path[line width=0.21mm, draw=L] (91.29,90.96) .. controls (93.11,93.27) and (98.48,93.29) .. (100.69,90.98);
	\path[line width=0.30mm, draw=L] (96.31,92.16);
	\draw(84.03,90.45) node[anchor=base west]{\fontsize{5.69}{6.83}\selectfont 1};
	\path[line width=0.30mm, draw=L, dash pattern=on 3pt off 2pt] (96.51,93.10);
	\path[line width=0.21mm, draw=L] (96.85,93.35) -- (95.68,92.73) -- (96.85,91.99);
	\path[line width=0.21mm, draw=L] (96.93,89.88) -- (95.76,89.26) -- (96.93,88.52);
	\path[line width=0.21mm, draw=L] (100.70,90.97) .. controls (103.18,88.82) and (107.95,88.78) .. (110.05,90.51);
	\path[line width=0.21mm, draw=L] (100.70,90.94) .. controls (103.12,93.36) and (107.20,93.35) .. (110.05,91.40);
	\path[line width=0.30mm, draw=L] (105.45,92.23);
	\draw(110.44,90.33) node[anchor=base west]{\fontsize{5.69}{6.83}\selectfont 2};
	\path[line width=0.30mm, draw=L, dash pattern=on 3pt off 2pt] (109.53,94.72);
	\path[line width=0.21mm, draw=L] (106.18,93.41) -- (105.01,92.78) -- (106.18,92.05);
	\path[line width=0.21mm, draw=L] (106.32,89.96) -- (105.15,89.33) -- (106.32,88.60);
	\path[line width=0.21mm, draw=L] (85.19,90.97) -- (91.49,90.97);
	\path[line width=0.21mm, draw=L] (89.30,91.61) -- (88.13,90.99) -- (89.30,90.25);
	\end{tikzpicture}}
\newcommand{\GraphTwoSolidDoubleLineSymm}{\begin{tikzpicture}[baseline=89.9mm,x=1.00mm, y=1.00mm, inner xsep=0pt, inner ysep=0pt, outer xsep=0pt, outer ysep=0pt]
	\useasboundingbox (83,87) rectangle +(35,7.9);
	\path[line width=0mm] (82.03,86.52) rectangle +(37.43,10.19);
	\definecolor{L}{rgb}{0,0,0}
	\path[line width=0.21mm, draw=L] (91.30,90.99) .. controls (93.33,88.45) and (98.72,88.85) .. (100.73,90.98);
	\path[line width=0.21mm, draw=L] (91.29,90.96) .. controls (93.11,93.27) and (98.48,93.29) .. (100.69,90.98);
	\path[line width=0.30mm, draw=L] (96.31,92.16);
	\draw(84.03,90.45) node[anchor=base west]{\fontsize{5.69}{6.83}\selectfont 1};
	\path[line width=0.30mm, draw=L, dash pattern=on 3pt off 2pt] (96.51,93.10);
	\path[line width=0.21mm, draw=L] (96.85,93.35) -- (95.68,92.73) -- (96.85,91.99);
	\path[line width=0.21mm, draw=L] (96.93,89.88) -- (95.76,89.26) -- (96.93,88.52);
	\path[line width=0.30mm, draw=L] (105.45,92.23);
	\path[line width=0.30mm, draw=L, dash pattern=on 3pt off 2pt] (109.53,94.72);
	\path[line width=0.21mm, draw=L] (106.18,93.41) -- (105.01,92.78) -- (106.18,92.05);
	\path[line width=0.21mm, draw=L] (106.32,89.96) -- (105.15,89.33) -- (106.32,88.60);
	\path[line width=0.21mm, draw=L] (85.19,90.97) -- (91.49,90.97);
	\path[line width=0.21mm, draw=L] (89.30,91.61) -- (88.13,90.99) -- (89.30,90.25);
	\path[line width=0.21mm, draw=L] (100.70,90.97) .. controls (103.18,88.82) and (107.38,88.77) .. (109.93,91.01);
	\path[line width=0.21mm, draw=L] (100.70,90.94) .. controls (103.12,93.36) and (107.83,93.36) .. (109.95,91.03);
	\path[line width=0.30mm, draw=L] (105.45,92.23);
	\draw(116.46,90.42) node[anchor=base west]{\fontsize{5.69}{6.83}\selectfont 2};
	\path[line width=0.30mm, draw=L, dash pattern=on 3pt off 2pt] (109.53,94.72);
	\path[line width=0.21mm, draw=L] (106.18,93.41) -- (105.01,92.78) -- (106.18,92.05);
	\path[line width=0.21mm, draw=L] (106.32,89.96) -- (105.15,89.33) -- (106.32,88.60);
	\path[line width=0.21mm, draw=L] (109.91,91.00) -- (116.20,91.00);
	\path[line width=0.21mm, draw=L] (114.01,91.64) -- (112.84,91.02) -- (114.01,90.28);
	\end{tikzpicture}}
\newcommand{\GraphTwoSolidTripple}{\begin{tikzpicture}[baseline=89.9mm,x=1.00mm, y=1.00mm, inner xsep=0pt, inner ysep=0pt, outer xsep=0pt, outer ysep=0pt]
	\useasboundingbox (83.50,86.8) rectangle +(18.5,7.9);
	\path[line width=0mm] (81.55,85.04) rectangle +(22.11,11.54);
	\definecolor{L}{rgb}{0,0,0}
	\path[line width=0.21mm, draw=L] (85.00,90.00) .. controls (89.50,87.00) and (95.50,87.00) .. (100.00,90.00);
	\path[line width=0.21mm, draw=L] (85.04,91.54) .. controls (89.54,94.54) and (95.54,94.54) .. (100.04,91.54);
	\path[line width=0.30mm, draw=L] (93.00,92.90);
	\draw(83.55,90.16) node[anchor=base west]{\fontsize{5.69}{6.83}\selectfont 1};
	\draw(100.67,90.18) node[anchor=base west]{\fontsize{5.69}{6.83}\selectfont 2};
	\path[line width=0.30mm, draw=L, dash pattern=on 3pt off 2pt] (93.34,94.58);
	\path[line width=0.21mm, draw=L] (93.63,94.42) -- (92.46,93.79) -- (93.63,93.06);
	\path[line width=0.21mm, draw=L] (93.71,88.40) -- (92.54,87.78) -- (93.71,87.04);
	\path[line width=0.21mm, draw=L] (84.97,90.82) -- (100.03,90.83);
	\path[line width=0.21mm, draw=L] (93.40,91.44) -- (92.23,90.82) -- (93.40,90.08);
	\end{tikzpicture}}

\newcommand{\GraphFive}{\begin{tikzpicture}[baseline=89.9mm,x=1.00mm, y=1.00mm, inner xsep=0pt, inner ysep=0pt, outer xsep=0pt, outer ysep=0pt]
	\useasboundingbox  (83.50,86.7) rectangle +(18.6,8);
	\path[line width=0mm] (81.55,85.06) rectangle +(22.11,11.51);
	\definecolor{L}{rgb}{0,0,0}
	\path[line width=0.21mm, draw=L, dash pattern=on 3pt off 2pt] (85.00,90.00) .. controls (89.50,87.00) and (95.50,87.00) .. (100.00,90.00);
	\path[line width=0.21mm, draw=L, dash pattern=on 3pt off 2pt] (85.04,91.54) .. controls (89.54,94.54) and (95.54,94.54) .. (100.04,91.54);
	\path[line width=0.30mm, draw=L] (93.00,92.90);
	\draw(83.55,90.16) node[anchor=base west]{\fontsize{5.69}{6.83}\selectfont 1};
	\draw(100.67,90.18) node[anchor=base west]{\fontsize{5.69}{6.83}\selectfont 2};
	\path[line width=0.30mm, draw=L, dash pattern=on 3pt off 2pt] (93.34,94.58);
	\path[line width=0.21mm, draw=L] (93.63,94.42) -- (92.46,93.79) -- (93.63,93.06);
	\path[line width=0.21mm, draw=L] (93.73,88.42) -- (92.56,87.80) -- (93.73,87.06);
	\path[line width=0.21mm, draw=L, dash pattern=on 3pt off 2pt] (99.99,90.74) -- (84.98,90.76);
	\path[line width=0.21mm, draw=L] (93.63,91.38) -- (92.46,90.76) -- (93.63,90.02);
	\end{tikzpicture}}

\newcommand{\GraphSeven}{\begin{tikzpicture}[baseline=95mm,x=1.00mm, y=1.00mm, inner xsep=0pt, inner ysep=0pt, outer xsep=0pt, outer ysep=0pt]
	\useasboundingbox (90.2,86) rectangle +(19.3,14.7);
	\path[line width=0mm] (88.47,84.47) rectangle +(22.66,17.94);
	\definecolor{L}{rgb}{0,0,0}
	\path[line width=0.30mm, draw=L] (94.34,96.87);
	\draw(99.57,95.43) node[anchor=base west]{\fontsize{5.69}{6.83}\selectfont 1};
	\path[line width=0.30mm, draw=L, dash pattern=on 3pt off 2pt] (93.75,96.11);
	\path[line width=0.21mm, draw=L, dash pattern=on 3pt off 2pt] (100.02,94.57) -- (99.97,86.91);
	\path[line width=0.21mm, draw=L] (99.38,91.20) -- (100.01,90.03) -- (100.74,91.20);
	\definecolor{F}{rgb}{0,0,0}
	\path[line width=0.30mm, draw=L, fill=F] (100.02,86.85) circle (0.38mm);
	\path[line width=0.21mm, draw=L, dash pattern=on 3pt off 2pt] (99.49,94.97) -- (93.02,89.98);
	\path[line width=0.21mm, draw=L] (96.54,93.52) -- (95.97,92.32) -- (97.35,92.42);
	\path[line width=0.30mm, draw=L, fill=F] (93.01,89.97) circle (0.38mm);
	\path[line width=0.30mm, draw=L, fill=F] (107.01,90.03) circle (0.38mm);
	\path[line width=0.21mm, draw=L, dash pattern=on 3pt off 2pt] (99.02,96.52) .. controls (98.02,99.00) and (94.04,100.95) .. (90.97,99.99);
	\path[line width=0.21mm, draw=L, dash pattern=on 3pt off 2pt] (99.00,95.54) .. controls (96.02,94.94) and (91.99,97.05) .. (90.98,100.05);
	\path[line width=0.30mm, draw=L] (94.34,96.87);
	\path[line width=0.30mm, draw=L, dash pattern=on 3pt off 2pt] (93.75,96.11);
	\path[line width=0.21mm, draw=L, dash pattern=on 3pt off 2pt] (101.02,95.49) .. controls (104.05,95.02) and (108.01,96.99) .. (108.95,99.93);
	\path[line width=0.21mm, draw=L, dash pattern=on 3pt off 2pt] (101.00,96.51) .. controls (101.98,98.97) and (106.05,100.95) .. (108.91,99.99);
	\path[line width=0.30mm, draw=L] (104.46,99.08);
	\path[line width=0.30mm, draw=L, dash pattern=on 3pt off 2pt] (104.18,100.00);
	\path[line width=0.30mm, draw=L, fill=F] (108.75,99.97) circle (0.38mm);
	\path[line width=0.30mm, draw=L, fill=F] (90.85,100.03) circle (0.38mm);
	\path[line width=0.21mm, draw=L, dash pattern=on 3pt off 2pt] (100.54,94.99) -- (107.01,90.04);
	\path[line width=0.21mm, draw=L] (102.76,92.43) -- (104.08,92.30) -- (103.52,93.56);
	\path[line width=0.21mm, draw=L] (105.63,95.56) -- (106.25,96.72) -- (104.97,96.72);
	\path[line width=0.21mm, draw=L] (104.07,98.57) -- (104.69,99.73) -- (103.40,99.73);
	\path[line width=0.21mm, draw=L] (95.15,96.69) -- (93.84,96.69) -- (94.44,95.55);
	\path[line width=0.21mm, draw=L] (96.61,99.72) -- (95.30,99.72) -- (95.90,98.59);
	\end{tikzpicture}}
\newcommand{\GraphEight}{\begin{tikzpicture}[baseline=90mm,x=1.00mm, y=1.00mm, inner xsep=0pt, inner ysep=0pt, outer xsep=0pt, outer ysep=0pt]
	\useasboundingbox  (83.00,89.92) rectangle +(21.67,5.08);
	\path[line width=0mm] (81.49,87.96) rectangle +(23.54,8.62);
	\definecolor{L}{rgb}{0,0,0}
	\path[line width=0.30mm, draw=L] (93.00,92.90);
	\draw(83.49,90.39) node[anchor=base west]{\fontsize{5.69}{6.83}\selectfont 1};
	\path[line width=0.30mm, draw=L, dash pattern=on 3pt off 2pt] (93.34,94.58);
	\path[line width=0.21mm, draw=L] (93.65,91.65) -- (92.48,91.02) -- (93.65,90.28);
	\path[line width=0.21mm, draw=L, dash pattern=on 3pt off 2pt] (84.92,90.99) -- (101.57,91.04);
	\draw(102.03,90.40) node[anchor=base west]{\fontsize{5.69}{6.83}\selectfont 2};
	\end{tikzpicture}}
\newcommand{\GraphNine}{\begin{tikzpicture}[baseline=89.9mm,x=1.00mm, y=1.00mm, inner xsep=0pt, inner ysep=0pt, outer xsep=0pt, outer ysep=0pt]
	\useasboundingbox (83.50,87) rectangle +(18.4,8);
	\path[line width=0mm] (81.63,85.24) rectangle +(21.79,11.56);
	\definecolor{L}{rgb}{0,0,0}
	\path[line width=0.21mm, draw=L, dash pattern=on 3pt off 2pt] (85.00,90.00) .. controls (89.50,87.00) and (97.01,86.92) .. (100.96,90.97);
	\path[line width=0.21mm, draw=L, dash pattern=on 3pt off 2pt] (85.02,91.92) .. controls (89.52,94.97) and (96.98,94.94) .. (100.88,91.35);
	\path[line width=0.30mm, draw=L] (93.00,92.90);
	\draw(83.63,90.40) node[anchor=base west]{\fontsize{5.69}{6.83}\selectfont 1};
	\path[line width=0.30mm, draw=L, dash pattern=on 3pt off 2pt] (93.34,94.58);
	\path[line width=0.21mm, draw=L] (92.10,93.44) -- (93.27,94.06) -- (92.10,94.80);
	\path[line width=0.21mm, draw=L] (92.05,87.24) -- (93.22,87.87) -- (92.05,88.60);
	\definecolor{F}{rgb}{0,0,0}
	\path[line width=0.30mm, draw=L, fill=F] (101.04,91.08) circle (0.38mm);
	\end{tikzpicture}}
\newcommand{\GraphTen}{\begin{tikzpicture}[baseline=90mm,x=1.00mm, y=1.00mm, inner xsep=0pt, inner ysep=0pt, outer xsep=0pt, outer ysep=0pt]
	\useasboundingbox  (84.6,89.92) rectangle +(17.3,2.5);
	\path[line width=0mm] (82.92,88.28) rectangle +(20.65,8.29);
	\definecolor{L}{rgb}{0,0,0}
	\path[line width=0.30mm, draw=L] (93.00,92.90);
	\path[line width=0.30mm, draw=L, dash pattern=on 3pt off 2pt] (93.34,94.58);
	\path[line width=0.21mm, draw=L] (93.65,91.65) -- (92.48,91.02) -- (93.65,90.28);
	\path[line width=0.21mm, draw=L, dash pattern=on 3pt off 2pt] (84.92,90.99) -- (101.57,91.04);
	\end{tikzpicture}}

\newcommand{\GraphEleven}{\begin{tikzpicture}[baseline=94mm,x=1.00mm, y=1.00mm, inner xsep=0pt, inner ysep=0pt, outer xsep=0pt, outer ysep=0pt]
	\useasboundingbox (94.8,93.8) rectangle +(20.6,2.2);
	\definecolor{L}{rgb}{0,0,0}
	\path[line width=0.30mm, draw=L] (96.52,96.93);
	\path[line width=0.30mm, draw=L, dash pattern=on 3pt off 2pt] (96.87,98.61);
	\path[line width=0.21mm, draw=L, dash pattern=on 3pt off 2pt] (95.10,94.89) -- (115.09,94.87);
	\definecolor{F}{rgb}{0,0,0}
	\path[line width=0.30mm, draw=L, fill=F] (104.96,94.96) circle (0.38mm);
	\path[line width=0.30mm, draw=L] (96.52,96.93);
	\path[line width=0.30mm, draw=L, dash pattern=on 3pt off 2pt] (96.87,98.61);
	\path[line width=0.21mm, draw=L] (101.04,95.53) -- (99.87,94.91) -- (101.04,94.17);
	\path[line width=0.30mm, draw=L] (113.47,96.91);
	\path[line width=0.30mm, draw=L, dash pattern=on 3pt off 2pt] (113.81,98.59);
	\path[line width=0.30mm, draw=L] (113.47,96.91);
	\path[line width=0.30mm, draw=L, dash pattern=on 3pt off 2pt] (113.81,98.59);
	\path[line width=0.21mm, draw=L] (111.13,95.52) -- (109.96,94.89) -- (111.13,94.16);	
	\end{tikzpicture}}
\newcommand{\GraphTwelve}{\begin{tikzpicture}[baseline=94mm,x=1.00mm, y=1.00mm, inner xsep=0pt, inner ysep=0pt, outer xsep=0pt, outer ysep=0pt]
	\useasboundingbox (94.8,93.8) rectangle +(29.6,2.2);
	\definecolor{L}{rgb}{0,0,0}
	\path[line width=0.30mm, draw=L] (96.52,96.93);
	\path[line width=0.30mm, draw=L, dash pattern=on 3pt off 2pt] (96.87,98.61);
	\path[line width=0.21mm, draw=L, dash pattern=on 3pt off 2pt] (95.10,94.89) -- (124.93,94.90);
	\definecolor{F}{rgb}{0,0,0}
	\path[line width=0.30mm, draw=L, fill=F] (104.96,94.96) circle (0.38mm);
	\path[line width=0.30mm, draw=L] (96.52,96.93);
	\path[line width=0.30mm, draw=L, dash pattern=on 3pt off 2pt] (96.87,98.61);
	\path[line width=0.21mm, draw=L] (101.04,95.53) -- (99.87,94.91) -- (101.04,94.17);
	\path[line width=0.30mm, draw=L] (113.47,96.91);
	\path[line width=0.30mm, draw=L, dash pattern=on 3pt off 2pt] (113.81,98.59);
	\path[line width=0.30mm, draw=L] (113.47,96.91);
	\path[line width=0.30mm, draw=L, dash pattern=on 3pt off 2pt] (113.81,98.59);
	\path[line width=0.21mm, draw=L] (111.13,95.52) -- (109.96,94.89) -- (111.13,94.16);
	\path[line width=0.30mm, draw=L, fill=F] (114.98,94.96) circle (0.38mm);
	\path[line width=0.21mm, draw=L] (120.96,95.55) -- (119.79,94.92) -- (120.96,94.19);
	\end{tikzpicture}}
\newcommand{\GraphThirteen}{\begin{tikzpicture}[baseline=89.5mm,x=1.00mm, y=1.00mm, inner xsep=0pt, inner ysep=0pt, outer xsep=0pt, outer ysep=0pt]
	\useasboundingbox   (85.3,89.5) rectangle +(10.4,7.8);
	\definecolor{L}{rgb}{0,0,0}
	\path[line width=0.30mm, draw=L] (94.51,93.38);
	\draw(90.02,89.61) node[anchor=base west]{\fontsize{5.69}{6.83}\selectfont 1};
	\path[line width=0.30mm, draw=L, dash pattern=on 3pt off 2pt] (94.85,95.06);
	\path[line width=0.21mm, draw=L, dash pattern=on 3pt off 2pt] (91.01,91.04) -- (95.05,96.54);
	\path[line width=0.21mm, draw=L] (93.34,93.08) -- (93.50,94.37) -- (92.26,93.92);
	\definecolor{F}{rgb}{0,0,0}
	\path[line width=0.30mm, draw=L, fill=F] (94.95,96.53) circle (0.38mm);
	\path[line width=0.21mm, draw=L, dash pattern=on 3pt off 2pt] (89.98,90.98) -- (86.06,96.55);
	\path[line width=0.30mm, draw=L, fill=F] (86.03,96.58) circle (0.38mm);
	\path[line width=0.21mm, draw=L] (88.86,93.82) -- (87.64,94.28) -- (87.91,92.94);
	\end{tikzpicture}}
\newcommand{\GraphFourteen}{\begin{tikzpicture}[baseline=94.2mm,x=1.00mm, y=1.00mm, inner xsep=0pt, inner ysep=0pt, outer xsep=0pt, outer ysep=0pt]
	\useasboundingbox  (79,94.2) rectangle +(14.8,7.9);
	\definecolor{L}{rgb}{0,0,0}
	\path[line width=0.30mm, draw=L] (93.50,93.26);
	\draw(88.27,94.37) node[anchor=base west]{\fontsize{5.69}{6.83}\selectfont 2};
	\path[line width=0.30mm, draw=L, dash pattern=on 3pt off 2pt] (93.84,94.94);
	\path[line width=0.30mm, draw=L] (79.77,91.24);
	\path[line width=0.21mm, draw=L] (84.59,95.78) -- (83.41,95.15) -- (84.59,94.42);
	\draw(79.00,94.38) node[anchor=base west]{\fontsize{5.69}{6.83}\selectfont 1};
	\path[line width=0.30mm, draw=L] (92.70,98.23);
	\path[line width=0.30mm, draw=L, dash pattern=on 3pt off 2pt] (93.04,99.91);
	\path[line width=0.21mm, draw=L, dash pattern=on 3pt off 2pt] (89.20,95.89) -- (93.23,101.39);
	\path[line width=0.21mm, draw=L] (91.52,97.93) -- (91.69,99.22) -- (90.45,98.77);
	\definecolor{F}{rgb}{0,0,0}
	\path[line width=0.30mm, draw=L, fill=F] (93.13,101.38) circle (0.38mm);
	\path[line width=0.21mm, draw=L, dash pattern=on 3pt off 2pt] (88.16,95.83) -- (84.25,101.40);
	\path[line width=0.30mm, draw=L, fill=F] (84.21,101.43) circle (0.38mm);
	\path[line width=0.21mm, draw=L] (87.04,98.67) -- (85.83,99.13) -- (86.09,97.80);
	\path[line width=0.21mm, draw=L, dash pattern=on 3pt off 2pt] (79.98,95.09) -- (88.02,95.09);
	\end{tikzpicture}}
\newcommand{\GraphFifteen}{\begin{tikzpicture}[baseline=89.9mm,x=1.00mm, y=1.00mm, inner xsep=0pt, inner ysep=0pt, outer xsep=0pt, outer ysep=0pt]
	\useasboundingbox (79.5,89.7) rectangle +(22.6,11.2);
	\definecolor{L}{rgb}{0,0,0}
	\path[line width=0.30mm, draw=L] (93.00,92.90);
	\draw(79.46,90.16) node[anchor=base west]{\fontsize{5.69}{6.83}\selectfont 1};
	\draw(100.67,90.18) node[anchor=base west]{\fontsize{5.69}{6.83}\selectfont 2};
	\path[line width=0.30mm, draw=L, dash pattern=on 3pt off 2pt] (93.34,94.58);
	\path[line width=0.21mm, draw=L, dash pattern=on 3pt off 2pt] (101.14,91.88) -- (101.15,99.80);
	\path[line width=0.21mm, draw=L] (101.80,94.96) -- (101.18,96.13) -- (100.44,94.96);
	\definecolor{F}{rgb}{0,0,0}
	\path[line width=0.30mm, draw=L, fill=F] (101.16,100.11) circle (0.38mm);
	\path[line width=0.30mm, draw=L] (81.90,92.87);
	\path[line width=0.21mm, draw=L, dash pattern=on 3pt off 2pt] (80.48,90.82) -- (100.46,90.81);
	\path[line width=0.30mm, draw=L, fill=F] (90.34,90.89) circle (0.38mm);
	\path[line width=0.30mm, draw=L] (81.90,92.87);
	\path[line width=0.21mm, draw=L] (86.42,91.47) -- (85.25,90.84) -- (86.42,90.11);
	\path[line width=0.30mm, draw=L] (98.85,92.85);
	\path[line width=0.30mm, draw=L] (98.85,92.85);
	\path[line width=0.21mm, draw=L] (96.50,91.45) -- (95.33,90.82) -- (96.50,90.09);
	\end{tikzpicture}}
\newcommand{\GraphSixteen}{\begin{tikzpicture}[baseline=94.3mm,x=1.00mm, y=1.00mm, inner xsep=0pt, inner ysep=0pt, outer xsep=0pt, outer ysep=0pt]
	\useasboundingbox    (79.1,94.2) rectangle +(10.4,8.05);
	\definecolor{L}{rgb}{0,0,0}
	\path[line width=0.30mm, draw=L] (93.50,93.26);
	\draw(88.13,94.52) node[anchor=base west]{\fontsize{5.69}{6.83}\selectfont 2};
	\path[line width=0.30mm, draw=L, dash pattern=on 3pt off 2pt] (93.84,94.94);
	\path[line width=0.30mm, draw=L] (79.77,91.24);
	\path[line width=0.21mm, draw=L] (84.59,95.78) -- (83.41,95.15) -- (84.59,94.42);
	\draw(79.07,94.65) node[anchor=base west]{\fontsize{5.69}{6.83}\selectfont 1};
	\path[line width=0.30mm, draw=L] (83.56,98.21);
	\path[line width=0.30mm, draw=L, dash pattern=on 3pt off 2pt] (83.90,99.89);
	\path[line width=0.21mm, draw=L, dash pattern=on 3pt off 2pt] (80.19,96.05) -- (84.09,101.37);
	\path[line width=0.21mm, draw=L] (82.39,97.91) -- (82.55,99.20) -- (81.31,98.75);
	\path[line width=0.21mm, draw=L, dash pattern=on 3pt off 2pt] (88.06,96.05) -- (84.18,101.58);
	\definecolor{F}{rgb}{0,0,0}
	\path[line width=0.30mm, draw=L, fill=F] (84.19,101.52) circle (0.38mm);
	\path[line width=0.21mm, draw=L] (87.04,98.67) -- (85.83,99.13) -- (86.09,97.80);
	\path[line width=0.21mm, draw=L, dash pattern=on 3pt off 2pt] (80.21,95.08) -- (88.02,95.09);
	\end{tikzpicture}}
\newcommand{\GraphSeventeen}{\begin{tikzpicture}[baseline=94.3mm,x=1.00mm, y=1.00mm, inner xsep=0pt, inner ysep=0pt, outer xsep=0pt, outer ysep=0pt]
	\useasboundingbox   (74.3,94.2) rectangle +(15.2,8.05);
	\definecolor{L}{rgb}{0,0,0}
	\path[line width=0.30mm, draw=L] (93.50,93.26);
	\draw(88.13,94.52) node[anchor=base west]{\fontsize{5.69}{6.83}\selectfont 2};
	\path[line width=0.30mm, draw=L, dash pattern=on 3pt off 2pt] (93.84,94.94);
	\path[line width=0.30mm, draw=L] (79.77,91.24);
	\path[line width=0.21mm, draw=L] (84.59,95.78) -- (83.41,95.15) -- (84.59,94.42);
	\draw(79.07,94.65) node[anchor=base west]{\fontsize{5.69}{6.83}\selectfont 1};
	\path[line width=0.30mm, draw=L] (83.56,98.21);
	\path[line width=0.30mm, draw=L, dash pattern=on 3pt off 2pt] (83.90,99.89);
	\path[line width=0.21mm, draw=L, dash pattern=on 3pt off 2pt] (80.19,96.05) -- (84.09,101.37);
	\path[line width=0.21mm, draw=L] (82.39,97.91) -- (82.55,99.20) -- (81.31,98.75);
	\path[line width=0.21mm, draw=L, dash pattern=on 3pt off 2pt] (88.06,96.05) -- (84.18,101.58);
	\definecolor{F}{rgb}{0,0,0}
	\path[line width=0.30mm, draw=L, fill=F] (84.19,101.52) circle (0.38mm);
	\path[line width=0.21mm, draw=L] (87.04,98.67) -- (85.83,99.13) -- (86.09,97.80);
	\path[line width=0.21mm, draw=L, dash pattern=on 3pt off 2pt] (80.21,95.08) -- (88.02,95.09);
	\path[line width=0.30mm, draw=L, dash pattern=on 3pt off 2pt] (74.74,99.86);
	\path[line width=0.21mm, draw=L, dash pattern=on 3pt off 2pt] (78.89,96.03) -- (75.02,101.56);
	\path[line width=0.30mm, draw=L, fill=F] (75.03,101.50) circle (0.38mm);
	\path[line width=0.21mm, draw=L] (77.88,98.64) -- (76.66,99.11) -- (76.93,97.77);
	\end{tikzpicture}}

\newcommand{\GraphEighteenOpposite}{\begin{tikzpicture}[baseline=94.3mm,x=1.00mm, y=1.00mm, inner xsep=0pt, inner ysep=0pt, outer xsep=0pt, outer ysep=0pt]
	\useasboundingbox   (71.2,94.2) rectangle +(18.3,8.05);
	\path[line width=0mm] (69.17,89.24) rectangle +(26.67,14.66);
	\definecolor{L}{rgb}{0,0,0}
	\path[line width=0.30mm, draw=L] (93.50,93.26);
	\draw(88.13,94.52) node[anchor=base west]{\fontsize{5.69}{6.83}\selectfont 2};
	\path[line width=0.30mm, draw=L, dash pattern=on 3pt off 2pt] (93.84,94.94);
	\path[line width=0.30mm, draw=L] (79.77,91.24);
	\path[line width=0.21mm, draw=L] (84.21,94.45) -- (83.11,95.09) -- (84.21,95.81);
	\draw(71.17,94.50) node[anchor=base west]{\fontsize{5.69}{6.83}\selectfont 1};
	\path[line width=0.30mm, draw=L] (83.56,98.21);
	\path[line width=0.30mm, draw=L, dash pattern=on 3pt off 2pt] (83.90,99.89);
	\path[line width=0.21mm, draw=L, dash pattern=on 3pt off 2pt] (79.43,95.08) -- (84.09,101.37);
	\path[line width=0.21mm, draw=L] (82.39,97.91) -- (82.55,99.20) -- (81.31,98.75);
	\path[line width=0.21mm, draw=L, dash pattern=on 3pt off 2pt] (88.06,96.05) -- (84.18,101.58);
	\definecolor{F}{rgb}{0,0,0}
	\path[line width=0.30mm, draw=L, fill=F] (84.19,101.52) circle (0.38mm);
	\path[line width=0.21mm, draw=L] (87.04,98.67) -- (85.83,99.13) -- (86.09,97.80);
	\path[line width=0.21mm, draw=L, dash pattern=on 3pt off 2pt] (72.08,95.10) -- (88.02,95.09);
	\path[line width=0.21mm, draw=L, dash pattern=on 3pt off 2pt] (88.06,96.05) -- (84.18,101.58);
	\path[line width=0.30mm, draw=L, fill=F] (79.46,95.15) circle (0.47mm);
	\path[line width=0.21mm, draw=L] (74.45,94.46) -- (75.62,95.09) -- (74.45,95.83);
	\end{tikzpicture}}
\newcommand{\GraphSquare}{\begin{tikzpicture}[baseline=94.3mm,x=1.00mm, y=1.00mm, inner xsep=0pt, inner ysep=0pt, outer xsep=0pt, outer ysep=0pt]
	\useasboundingbox   (78,94.2) rectangle +(12,12);
	\definecolor{L}{rgb}{0,0,0}
	\path[line width=0mm] (77.06,89.24) rectangle +(18.78,16.97);
	\definecolor{L}{rgb}{0,0,0}
	\path[line width=0.30mm, draw=L] (93.50,93.26);
	\draw(88.13,94.52) node[anchor=base west]{\fontsize{5.69}{6.83}\selectfont 2};
	\path[line width=0.30mm, draw=L, dash pattern=on 3pt off 2pt] (93.84,94.94);
	\path[line width=0.30mm, draw=L] (79.77,91.24);
	\path[line width=0.21mm, draw=L] (84.59,95.78) -- (83.41,95.15) -- (84.59,94.42);
	\draw(79.20,94.65) node[anchor=base west]{\fontsize{5.69}{6.83}\selectfont 1};
	\path[line width=0.30mm, draw=L] (83.56,98.21);
	\path[line width=0.30mm, draw=L, dash pattern=on 3pt off 2pt] (83.90,99.89);
	\path[line width=0.21mm, draw=L, dash pattern=on 3pt off 2pt] (88.61,96.28) -- (88.58,103.53);
	\definecolor{F}{rgb}{0,0,0}
	\path[line width=0.30mm, draw=L, fill=F] (79.61,103.56) circle (0.38mm);
	\path[line width=0.21mm, draw=L, dash pattern=on 3pt off 2pt] (80.21,95.08) -- (88.02,95.09);
	\path[line width=0.21mm, draw=L, dash pattern=on 3pt off 2pt] (79.64,96.19) -- (79.64,103.65);
	\path[line width=0.21mm, draw=L] (89.23,99.28) -- (88.58,100.41) -- (87.99,99.25);
	\path[line width=0.21mm, draw=L] (80.30,99.31) -- (79.65,100.44) -- (79.06,99.28);
	\path[line width=0.21mm, draw=L] (83.85,104.21) -- (84.99,103.49) -- (83.85,102.85);
	\path[line width=0.21mm, draw=L, dash pattern=on 3pt off 2pt] (79.89,103.46) -- (88.55,103.48);
	\path[line width=0.30mm, draw=L, fill=F] (88.69,103.46) circle (0.38mm);
	\end{tikzpicture}}
\newcommand{\GraphSquareRotated}{\begin{tikzpicture}[baseline=98mm,x=1.00mm, y=1.00mm, inner xsep=0pt, inner ysep=0pt, outer xsep=0pt, outer ysep=0pt]
	\useasboundingbox   (78,90.2) rectangle +(13,17);
	\path[line width=0mm] (75.58,89.24) rectangle +(20.26,18.02);
	\definecolor{L}{rgb}{0,0,0}
	\path[line width=0.30mm, draw=L] (93.50,93.26);
	\draw(89.65,98.31) node[anchor=base west]{\fontsize{5.69}{6.83}\selectfont 2};
	\path[line width=0.30mm, draw=L, dash pattern=on 3pt off 2pt] (93.84,94.94);
	\path[line width=0.30mm, draw=L] (79.77,91.24);
	\path[line width=0.21mm, draw=L] (86.69,96.70) -- (86.34,95.42) -- (87.60,95.74);
	\draw(77.58,98.41) node[anchor=base west]{\fontsize{5.69}{6.83}\selectfont 1};
	\path[line width=0.30mm, draw=L] (84.55,98.13);
	\path[line width=0.30mm, draw=L, dash pattern=on 3pt off 2pt] (83.56,99.53);
	\path[line width=0.21mm, draw=L, dash pattern=on 3pt off 2pt] (89.46,99.52) -- (83.98,105.02);
	\definecolor{F}{rgb}{0,0,0}
	\path[line width=0.30mm, draw=L, fill=F] (83.99,104.88) circle (0.38mm);
	\path[line width=0.21mm, draw=L, dash pattern=on 3pt off 2pt] (84.03,93.03) -- (89.54,98.51);
	\path[line width=0.21mm, draw=L, dash pattern=on 3pt off 2pt] (84.01,93.05) -- (78.55,98.45);
	\path[line width=0.21mm, draw=L] (87.43,102.54) -- (86.16,102.85) -- (86.60,101.62);
	\path[line width=0.21mm, draw=L] (81.29,96.73) -- (81.65,95.47) -- (80.46,95.82);
	\path[line width=0.21mm, draw=L] (80.49,102.46) -- (81.80,102.80) -- (81.48,101.53);
	\path[line width=0.21mm, draw=L, dash pattern=on 3pt off 2pt] (78.57,99.48) -- (84.00,104.96);
	\path[line width=0.30mm, draw=L, fill=F] (83.98,93.08) circle (0.38mm);
	\end{tikzpicture}}
\newcommand{\GraphTwenty}{\begin{tikzpicture}[baseline=94mm,x=1.00mm, y=1.00mm, inner xsep=0pt, inner ysep=0pt, outer xsep=0pt, outer ysep=0pt]
	\useasboundingbox (90.5,85) rectangle +(30,20);
	\path[line width=0mm] (88.86,84.78) rectangle +(32.79,20.30);
	\definecolor{L}{rgb}{0,0,0}
	\path[line width=0.30mm, draw=L] (99.65,92.37);
	\draw(109.33,94.06) node[anchor=base west]{\fontsize{5.69}{6.83}\selectfont 2};
	\path[line width=0.30mm, draw=L, dash pattern=on 3pt off 2pt] (100.67,93.35);
	\path[line width=0.21mm, draw=L, dash pattern=on 3pt off 2pt] (91.79,93.69) .. controls (93.16,89.77) and (98.57,86.72) .. (103.63,87.01);
	\path[line width=0.21mm, draw=L, dash pattern=on 3pt off 2pt] (92.32,94.89) .. controls (97.41,95.83) and (103.01,91.62) .. (103.86,87.07);
	\path[line width=0.30mm, draw=L] (99.65,92.37);
	\path[line width=0.30mm, draw=L, dash pattern=on 3pt off 2pt] (100.67,93.35);
	\path[line width=0.21mm, draw=L, dash pattern=on 3pt off 2pt] (110.52,95.03) -- (119.29,95.04);
	\draw(90.86,94.11) node[anchor=base west]{\fontsize{5.69}{6.83}\selectfont 1};
	\path[line width=0.21mm, draw=L, dash pattern=on 3pt off 2pt] (109.07,94.24) -- (104.33,87.42);
	\path[line width=0.21mm, draw=L, dash pattern=on 3pt off 2pt] (109.30,95.47) -- (103.95,102.32);
	\path[line width=0.30mm, draw=L] (106.72,95.74);
	\path[line width=0.30mm, draw=L, dash pattern=on 3pt off 2pt] (108.07,96.15);
	\path[line width=0.30mm, draw=L] (106.72,95.74);
	\path[line width=0.30mm, draw=L, dash pattern=on 3pt off 2pt] (108.07,96.15);
	\path[line width=0.21mm, draw=L, dash pattern=on 3pt off 2pt] (96.28,102.34) -- (103.91,102.36);
	\path[line width=0.21mm, draw=L, dash pattern=on 3pt off 2pt] (96.18,102.37) -- (91.82,95.74);
	\draw(114.44,96.30) node[anchor=base west]{\fontsize{5.69}{6.83}\selectfont b};
	\draw(107.98,99.01) node[anchor=base west]{\fontsize{5.69}{6.83}\selectfont c};
	\draw(91.54,98.91) node[anchor=base west]{\fontsize{5.69}{6.83}\selectfont a};
	\definecolor{F}{rgb}{0,0,0}
	\path[line width=0.30mm, draw=L, fill=F] (96.10,102.36) circle (0.38mm);
	\path[line width=0.30mm, draw=L, fill=F] (103.93,102.29) circle (0.38mm);
	\path[line width=0.30mm, draw=L, fill=F] (119.26,95.03) circle (0.38mm);
	\path[line width=0.30mm, draw=L, fill=F] (103.89,87.17) circle (0.38mm);
	\path[line width=0.21mm, draw=L] (94.01,97.97) -- (94.17,99.26) -- (92.93,98.81);
	\path[line width=0.21mm, draw=L] (100.68,103.08) -- (99.61,102.35) -- (100.77,101.72);
	\path[line width=0.21mm, draw=L] (98.51,93.06) -- (99.80,92.97) -- (99.28,94.18);
	\path[line width=0.21mm, draw=L] (114.36,95.76) -- (115.40,95.01) -- (114.44,94.40);
	\path[line width=0.21mm, draw=L] (95.09,88.91) -- (96.38,88.82) -- (95.86,90.03);
	\path[line width=0.21mm, draw=L] (106.50,91.59) -- (106.29,90.31) -- (107.55,90.78);
	\path[line width=0.21mm, draw=L] (107.59,98.84) -- (106.34,99.21) -- (106.59,97.91);
	\draw(94.75,95.54) node[anchor=base west]{\fontsize{5.69}{6.83}\selectfont a};
	\draw(93.46,91.94) node[anchor=base west]{\fontsize{5.69}{6.83}\selectfont a};
	\end{tikzpicture}}
\newcommand{\GraphTwentyOne}{\begin{tikzpicture}[baseline=96.5mm,x=1.00mm, y=1.00mm, inner xsep=0pt, inner ysep=0pt, outer xsep=0pt, outer ysep=0pt]
	\useasboundingbox (81,89) rectangle +(16.4,16);
	\path[line width=0mm] (72.74,89.24) rectangle +(25.95,15.78);
	\definecolor{L}{rgb}{0,0,0}
	\path[line width=0.30mm, draw=L] (89.99,95.62);
	\draw(81.52,92.08) node[anchor=base west]{\fontsize{5.69}{6.83}\selectfont 2};
	\path[line width=0.30mm, draw=L, dash pattern=on 3pt off 2pt] (94.94,95.49);
	\path[line width=0.30mm, draw=L] (79.77,91.24);
	\path[line width=0.21mm, draw=L] (85.27,96.04) -- (84.69,94.85) -- (86.00,94.85);
	\draw(81.50,101.45) node[anchor=base west]{\fontsize{5.69}{6.83}\selectfont 1};
	\path[line width=0.30mm, draw=L] (83.54,97.84);
	\path[line width=0.30mm, draw=L, dash pattern=on 3pt off 2pt] (82.83,99.40);
	\path[line width=0.21mm, draw=L] (82.00,93.92) -- (82.02,101.00);
	\path[line width=0.21mm, draw=L] (82.66,96.93) -- (82.04,98.07) -- (81.30,96.97);
	\path[line width=0.21mm, draw=L] (88.79,97.51) -- (82.52,101.46);
	\path[line width=0.21mm, draw=L] (86.01,100.01) -- (84.75,100.01) -- (85.39,98.84);
	\path[line width=0.21mm, draw=L] (82.61,93.48) -- (88.81,97.52);
	\path[line width=0.30mm, draw=L, dash pattern=on 3pt off 2pt] (74.74,99.86);
	\path[line width=0.21mm, draw=L] (95.44,97.51) -- (88.71,97.50);
	\path[line width=0.21mm, draw=L] (92.59,98.10) -- (91.47,97.49) -- (92.60,96.90);
	\draw(95.69,96.80) node[anchor=base west]{\fontsize{5.69}{6.83}\selectfont 3};
	\end{tikzpicture}}
\newcommand{\GraphTwentyThree}{\begin{tikzpicture}[baseline=89.2mm,x=1.00mm, y=1.00mm, inner xsep=0pt, inner ysep=0pt, outer xsep=0pt, outer ysep=0pt]
	\useasboundingbox (89,83) rectangle +(21,15);
	\path[line width=0mm] (87.43,82.97) rectangle +(23.65,13.94);
	\definecolor{L}{rgb}{0,0,0}
	\path[line width=0.30mm, draw=L] (95.32,91.80);
	\draw(89.43,89.36) node[anchor=base west]{\fontsize{5.69}{6.83}\selectfont 1};
	\draw(108.08,89.32) node[anchor=base west]{\fontsize{5.69}{6.83}\selectfont 2};
	\path[line width=0.30mm, draw=L, dash pattern=on 3pt off 2pt] (95.53,92.84);
	\path[line width=0.21mm, draw=L] (100.32,89.97) .. controls (100.27,92.71) and (98.04,94.90) .. (95.30,94.91) .. controls (92.85,94.91) and (90.77,93.13) .. (90.39,90.71);
	\path[line width=0.21mm, draw=L] (90.37,89.18) .. controls (90.78,86.78) and (92.87,85.03) .. (95.30,85.05) .. controls (98.03,85.06) and (100.26,87.24) .. (100.33,89.97);
	\path[line width=0.21mm, draw=L] (99.58,93.65) -- (98.28,93.92) -- (98.62,92.49);
	\path[line width=0.21mm, draw=L] (92.73,86.61) -- (91.43,86.88) -- (91.75,85.51);
	\path[line width=0.21mm, draw=L] (100.31,89.94) -- (107.87,89.94);
	\path[line width=0.21mm, draw=L] (91.66,94.26) -- (91.31,92.97) -- (92.75,93.22);
	\path[line width=0.21mm, draw=L] (104.95,90.74) -- (103.89,89.94) -- (105.05,89.33);
	\path[line width=0.21mm, draw=L] (98.59,87.27) -- (98.24,85.99) -- (99.68,86.24);
	\path[line width=0.21mm, draw=L] (95.36,94.91) -- (95.34,84.97);
	\path[line width=0.21mm, draw=L] (94.64,90.73) -- (95.34,89.45) -- (96.14,90.69);
	\end{tikzpicture}}
\newcommand{\GraphTwentyThreeLeft}{\begin{tikzpicture}[baseline=89.2mm,x=1.00mm, y=1.00mm, inner xsep=0pt, inner ysep=0pt, outer xsep=0pt, outer ysep=0pt]
	\useasboundingbox (81,83) rectangle +(21,15);
	\path[line width=0mm] (79.61,82.97) rectangle +(23.34,13.94);
	\definecolor{L}{rgb}{0,0,0}
	\path[line width=0.30mm, draw=L] (95.32,91.80);
	\draw(81.61,89.49) node[anchor=base west]{\fontsize{5.69}{6.83}\selectfont 1};
	\draw(99.94,89.36) node[anchor=base west]{\fontsize{5.69}{6.83}\selectfont 2};
	\path[line width=0.30mm, draw=L, dash pattern=on 3pt off 2pt] (95.53,92.84);
	\path[line width=0.21mm, draw=L] (100.17,90.96) .. controls (99.71,93.28) and (97.66,94.93) .. (95.30,94.91) .. controls (92.56,94.88) and (90.36,92.65) .. (90.34,89.92);
	\path[line width=0.21mm, draw=L] (90.33,89.99) .. controls (90.39,87.28) and (92.58,85.09) .. (95.30,85.05) .. controls (97.68,85.01) and (99.75,86.66) .. (100.24,88.99);
	\path[line width=0.21mm, draw=L] (99.58,93.65) -- (98.28,93.92) -- (98.62,92.49);
	\path[line width=0.21mm, draw=L] (92.73,86.61) -- (91.43,86.88) -- (91.75,85.51);
	\path[line width=0.21mm, draw=L] (82.77,89.96) -- (90.34,89.96);
	\path[line width=0.21mm, draw=L] (91.66,94.26) -- (91.31,92.97) -- (92.75,93.22);
	\path[line width=0.21mm, draw=L] (87.42,90.77) -- (86.36,89.97) -- (87.51,89.36);
	\path[line width=0.21mm, draw=L] (98.59,87.27) -- (98.24,85.99) -- (99.68,86.24);
	\path[line width=0.21mm, draw=L] (95.36,94.91) -- (95.34,84.97);
	\path[line width=0.21mm, draw=L] (94.64,90.72) -- (95.34,89.45) -- (96.14,90.69);
	\end{tikzpicture}}
\newcommand{\GraphTwentyThreeSymm}{\begin{tikzpicture}[baseline=89.2mm,x=1.00mm, y=1.00mm, inner xsep=0pt, inner ysep=0pt, outer xsep=0pt, outer ysep=0pt]
	\useasboundingbox (81,83) rectangle +(29,15);
	\path[line width=0mm] (79.69,82.98) rectangle +(31.75,14.10);
	\definecolor{L}{rgb}{0,0,0}
	\path[line width=0.30mm, draw=L] (95.32,91.80);
	\draw(81.69,89.47) node[anchor=base west]{\fontsize{5.69}{6.83}\selectfont 1};
	\draw(108.44,89.42) node[anchor=base west]{\fontsize{5.69}{6.83}\selectfont 2};
	\path[line width=0.30mm, draw=L, dash pattern=on 3pt off 2pt] (95.53,92.84);
	\path[line width=0.21mm, draw=L] (99.75,93.80) -- (98.45,94.07) -- (98.79,92.64);
	\path[line width=0.21mm, draw=L] (92.75,86.64) -- (91.45,86.90) -- (91.77,85.53);
	\path[line width=0.21mm, draw=L] (82.88,89.93) -- (90.44,89.93);
	\path[line width=0.21mm, draw=L] (91.66,94.26) -- (91.31,92.97) -- (92.75,93.22);
	\path[line width=0.21mm, draw=L] (87.52,90.74) -- (86.46,89.94) -- (87.62,89.33);
	\path[line width=0.21mm, draw=L] (98.69,87.20) -- (98.35,85.92) -- (99.79,86.17);
	\path[line width=0.21mm, draw=L] (95.45,95.06) -- (95.45,84.98);
	\path[line width=0.21mm, draw=L] (94.73,90.70) -- (95.42,89.43) -- (96.23,90.66);
	\path[line width=0.21mm, draw=L] (95.45,90.05) circle (5.02mm);
	\path[line width=0.21mm, draw=L] (100.53,89.88) -- (108.09,89.88);
	\path[line width=0.21mm, draw=L] (105.10,90.67) -- (104.04,89.86) -- (105.19,89.26);
	\end{tikzpicture}}
\newcommand{\GraphTwentyFive}{\begin{tikzpicture}[baseline=90.7mm,x=1.00mm, y=1.00mm, inner xsep=0pt, inner ysep=0pt, outer xsep=0pt, outer ysep=0pt]
	\useasboundingbox  (84,85) rectangle +(17,13.5);
	\path[line width=0mm] (82.27,84.64) rectangle +(20.01,13.91);
	\definecolor{L}{rgb}{0,0,0}
	\path[line width=0.30mm, draw=L] (93.00,92.90);
	\path[line width=0.30mm, draw=L, dash pattern=on 3pt off 2pt] (93.34,94.58);
	\path[line width=0.21mm, draw=L] (85.68,92.11) -- (92.63,96.20) -- (98.90,92.13);
	\definecolor{F}{rgb}{0,0,0}
	\path[line width=0.21mm, draw=L, fill=F] (85.68,92.11) -- (87.24,92.22) -- (85.68,92.11) -- (86.54,93.43) -- (85.68,92.11) -- cycle;
	\path[line width=0.21mm, draw=L, fill=F] (98.90,92.13) -- (98.26,93.38) -- (97.50,92.20) -- (98.90,92.13) -- cycle;
	\draw(84.27,90.94) node[anchor=base west]{\fontsize{5.69}{6.83}\selectfont 1};
	\draw(99.28,90.82) node[anchor=base west]{\fontsize{5.69}{6.83}\selectfont 2};
	\path[line width=0.30mm, draw=L, fill=F] (92.65,96.17) circle (0.38mm);
	\path[line width=0.21mm, draw=L] (85.62,91.01) -- (92.56,86.92) -- (98.83,90.99);
	\path[line width=0.21mm, draw=L, fill=F] (85.62,91.01) -- (86.47,89.69) -- (85.62,91.01) -- (87.18,90.90) -- (85.62,91.01) -- cycle;
	\path[line width=0.21mm, draw=L, fill=F] (98.83,90.99) -- (97.43,90.92) -- (98.19,89.74) -- (98.83,90.99) -- cycle;
	\path[line width=0.30mm, draw=L, fill=F] (92.51,87.03) circle (0.38mm);
	\end{tikzpicture}}
\newcommand{\GraphTwentySix}{\begin{tikzpicture}[baseline=89.2mm,x=1.00mm, y=1.00mm, inner xsep=0pt, inner ysep=0pt, outer xsep=0pt, outer ysep=0pt]
	\useasboundingbox  (83.5,88) rectangle +(28.5,11.5);
	\path[line width=0mm] (81.84,86.94) rectangle +(31.21,12.00);
	\definecolor{L}{rgb}{0,0,0}
	\path[line width=0.30mm, draw=L] (92.30,91.88);
	\path[line width=0.30mm, draw=L, dash pattern=on 3pt off 2pt] (92.65,93.56);
	\path[line width=0.21mm, draw=L] (85.24,90.46) -- (91.27,96.49) -- (97.19,90.54);
	\definecolor{F}{rgb}{0,0,0}
	\path[line width=0.21mm, draw=L, fill=F] (85.24,90.46) -- (86.73,90.95) -- (85.24,90.46) -- (85.74,91.94) -- (85.24,90.46) -- cycle;
	\path[line width=0.21mm, draw=L, fill=F] (97.19,90.54) -- (96.83,91.89) -- (95.84,90.90) -- (97.19,90.54) -- cycle;
	\draw(83.84,89.46) node[anchor=base west]{\fontsize{5.69}{6.83}\selectfont 1};
	\draw(110.04,89.37) node[anchor=base west]{\fontsize{5.69}{6.83}\selectfont 2};
	\path[line width=0.30mm, draw=L, fill=F] (91.24,96.56) circle (0.38mm);
	\path[line width=0.21mm, draw=L] (85.43,89.90) -- (109.55,89.92);
	\path[line width=0.21mm, draw=L, fill=F] (85.43,89.90) -- (86.83,89.20) -- (85.43,89.90) -- (86.83,90.60) -- (85.43,89.90) -- cycle;
	\path[line width=0.21mm, draw=L, fill=F] (109.55,89.92) -- (108.34,90.62) -- (108.34,89.22) -- (109.55,89.92) -- cycle;
	\path[line width=0.30mm, draw=L, fill=F] (97.52,90.00) circle (0.38mm);
	\end{tikzpicture}}
\newcommand{\GraphTwentySixOpposite}{\begin{tikzpicture}[baseline=89.2mm,x=1.00mm, y=1.00mm, inner xsep=0pt, inner ysep=0pt, outer xsep=0pt, outer ysep=0pt]
	\useasboundingbox  (83.5,88) rectangle +(28.5,11.5);
	\path[line width=0mm] (81.84,86.94) rectangle +(31.21,12.12);
	\definecolor{L}{rgb}{0,0,0}
	\path[line width=0.30mm, draw=L] (104.92,92.00);
	\path[line width=0.30mm, draw=L, dash pattern=on 3pt off 2pt] (105.26,93.68);
	\path[line width=0.21mm, draw=L] (97.86,90.58) -- (103.88,96.60) -- (109.80,90.66);
	\definecolor{F}{rgb}{0,0,0}
	\path[line width=0.21mm, draw=L, fill=F] (97.86,90.58) -- (99.34,91.07) -- (97.86,90.58) -- (98.35,92.06) -- (97.86,90.58) -- cycle;
	\path[line width=0.21mm, draw=L, fill=F] (109.80,90.66) -- (109.45,92.01) -- (108.45,91.02) -- (109.80,90.66) -- cycle;
	\draw(83.84,89.46) node[anchor=base west]{\fontsize{5.69}{6.83}\selectfont 1};
	\draw(110.04,89.37) node[anchor=base west]{\fontsize{5.69}{6.83}\selectfont 2};
	\path[line width=0.30mm, draw=L, fill=F] (103.85,96.67) circle (0.38mm);
	\path[line width=0.21mm, draw=L] (85.43,89.90) -- (109.55,89.92);
	\path[line width=0.21mm, draw=L, fill=F] (85.43,89.90) -- (86.83,89.20) -- (85.43,89.90) -- (86.83,90.60) -- (85.43,89.90) -- cycle;
	\path[line width=0.21mm, draw=L, fill=F] (109.55,89.92) -- (108.34,90.62) -- (108.34,89.22) -- (109.55,89.92) -- cycle;
	\path[line width=0.30mm, draw=L, fill=F] (97.52,90.00) circle (0.38mm);
	\end{tikzpicture}}
\newcommand{\GraphTwentySeven}{\begin{tikzpicture}[baseline=89.2mm,x=1.00mm, y=1.00mm, inner xsep=0pt, inner ysep=0pt, outer xsep=0pt, outer ysep=0pt]
	\useasboundingbox  (80.5,83) rectangle +(29.5,14.5);
	\path[line width=0mm] (79.05,83.11) rectangle +(32.11,14.03);
	\definecolor{L}{rgb}{0,0,0}
	\path[line width=0.30mm, draw=L] (92.30,91.88);
	\path[line width=0.30mm, draw=L, dash pattern=on 3pt off 2pt] (92.64,93.56);
	\path[line width=0.21mm, draw=L] (82.34,90.06) -- (90.00,90.08);
	\definecolor{F}{rgb}{0,0,0}
	\path[line width=0.21mm, draw=L, fill=F] (82.34,90.06) -- (83.74,89.36) -- (82.34,90.06) -- (83.74,90.76) -- (82.34,90.06) -- cycle;
	\draw(81.05,89.61) node[anchor=base west]{\fontsize{5.69}{6.83}\selectfont 1};
	\draw(108.15,89.32) node[anchor=base west]{\fontsize{5.69}{6.83}\selectfont 2};
	\path[line width=0.21mm, draw=L] (100.07,90.01) -- (107.63,90.03);
	\path[line width=0.21mm, draw=L, fill=F] (107.63,90.03) -- (106.42,90.73) -- (106.42,89.33) -- (107.63,90.03) -- cycle;
	\path[line width=0.21mm, draw=L] (90.02,90.06) .. controls (90.08,87.32) and (92.30,85.14) .. (95.03,85.11) .. controls (97.44,85.08) and (99.53,86.79) .. (99.98,89.15);
	\path[line width=0.21mm, draw=L, fill=F] (99.98,89.15) -- (99.06,88.10) -- (100.44,87.83) -- (99.98,89.15) -- cycle;
	\path[line width=0.21mm, draw=L] (100.05,90.09) .. controls (100.04,92.86) and (97.80,95.11) .. (95.03,95.13) .. controls (92.53,95.16) and (90.39,93.31) .. (90.05,90.83);
	\path[line width=0.21mm, draw=L, fill=F] (90.05,90.83) -- (90.94,92.12) -- (90.05,90.83) -- (89.55,92.31) -- (90.05,90.83) -- cycle;
	\path[line width=0.30mm, draw=L, fill=F] (90.04,90.02) circle (0.38mm);
	\path[line width=0.30mm, draw=L, fill=F] (100.02,89.97) circle (0.38mm);
	\end{tikzpicture}}
\newcommand{\GraphTwentyEight}{\begin{tikzpicture}[baseline=89.2mm,x=1.00mm, y=1.00mm, inner xsep=0pt, inner ysep=0pt, outer xsep=0pt, outer ysep=0pt]
	\useasboundingbox  (71.5,81.5) rectangle +(40,18);
	\path[line width=0mm] (70.31,81.77) rectangle +(42.21,17.30);
	\definecolor{L}{rgb}{0,0,0}
	\path[line width=0.30mm, draw=L] (92.30,91.88);
	\path[line width=0.30mm, draw=L, dash pattern=on 3pt off 2pt] (92.64,93.56);
	\path[line width=0.21mm, draw=L] (85.72,90.75) -- (91.58,96.69) -- (97.21,90.70);
	\definecolor{F}{rgb}{0,0,0}
	\path[line width=0.21mm, draw=L, fill=F] (85.72,90.75) -- (87.20,91.26) -- (85.72,90.75) -- (86.20,92.24) -- (85.72,90.75) -- cycle;
	\path[line width=0.21mm, draw=L, fill=F] (97.21,90.70) -- (96.89,92.06) -- (95.87,91.11) -- (97.21,90.70) -- cycle;
	\draw(72.31,89.37) node[anchor=base west]{\fontsize{5.69}{6.83}\selectfont 1};
	\draw(109.51,89.30) node[anchor=base west]{\fontsize{5.69}{6.83}\selectfont 2};
	\path[line width=0.30mm, draw=L, fill=F] (91.58,96.69) circle (0.38mm);
	\path[line width=0.21mm, draw=L] (86.21,89.94) -- (108.91,89.96);
	\path[line width=0.21mm, draw=L, fill=F] (86.21,89.94) -- (87.61,89.24) -- (86.21,89.94) -- (87.61,90.64) -- (86.21,89.94) -- cycle;
	\path[line width=0.21mm, draw=L, fill=F] (108.91,89.96) -- (107.70,90.66) -- (107.70,89.26) -- (108.91,89.96) -- cycle;
	\path[line width=0.30mm, draw=L, fill=F] (97.52,90.00) circle (0.38mm);
	\path[line width=0.30mm, draw=L] (87.43,91.47) ellipse (0.00mm and 0.00mm);
	\path[line width=0.21mm, draw=L] (85.36,90.02) .. controls (85.30,86.61) and (88.03,83.81) .. (91.44,83.77) .. controls (94.54,83.74) and (97.17,86.06) .. (97.52,89.15);
	\path[line width=0.21mm, draw=L, fill=F] (97.52,89.15) -- (96.68,88.02) -- (98.08,87.86) -- (97.52,89.15) -- cycle;
	\path[line width=0.30mm, draw=L, fill=F] (85.38,90.06) circle (0.38mm);
	\path[line width=0.21mm, draw=L] (85.29,89.96) -- (73.73,89.96);
	\path[line width=0.21mm, draw=L, fill=F] (73.73,89.96) -- (75.13,89.26) -- (73.73,89.96) -- (75.13,90.66) -- (73.73,89.96) -- cycle;
	\end{tikzpicture}}
\newcommand{\GraphTwentyNine}{\begin{tikzpicture}[baseline=89.9mm,x=1.00mm, y=1.00mm, inner xsep=0pt, inner ysep=0pt, outer xsep=0pt, outer ysep=0pt]
	\useasboundingbox (83,87) rectangle +(35,7.9);
	\path[line width=0mm] (82.03,86.71) rectangle +(37.43,10.00);
	\definecolor{L}{rgb}{0,0,0}
	\path[line width=0.21mm, draw=L] (91.47,90.97) .. controls (93.27,88.00) and (97.88,88.10) .. (100.20,90.47);
	\definecolor{F}{rgb}{0,0,0}
	\path[line width=0.21mm, draw=L, fill=F] (100.20,90.47) -- (98.81,90.26) -- (99.69,89.17) -- (100.20,90.47) -- cycle;
	\path[line width=0.21mm, draw=L] (91.85,91.58) .. controls (93.67,93.89) and (98.64,93.97) .. (100.69,90.98);
	\path[line width=0.21mm, draw=L, fill=F] (91.85,91.58) -- (93.38,91.93) -- (91.85,91.58) -- (92.49,93.01) -- (91.85,91.58) -- cycle;
	\path[line width=0.30mm, draw=L] (96.31,92.16);
	\draw(84.03,90.45) node[anchor=base west]{\fontsize{5.69}{6.83}\selectfont 1};
	\path[line width=0.30mm, draw=L, dash pattern=on 3pt off 2pt] (96.51,93.10);
	\path[line width=0.30mm, draw=L] (105.45,92.23);
	\path[line width=0.30mm, draw=L, dash pattern=on 3pt off 2pt] (109.53,94.72);
	\path[line width=0.21mm, draw=L] (85.19,90.97) -- (91.49,90.97);
	\path[line width=0.21mm, draw=L, fill=F] (85.19,90.97) -- (86.59,90.27) -- (85.19,90.97) -- (86.59,91.67) -- (85.19,90.97) -- cycle;
	\path[line width=0.21mm, draw=L] (100.70,90.97) .. controls (103.37,88.12) and (107.13,88.05) .. (109.63,90.44);
	\path[line width=0.21mm, draw=L, fill=F] (109.63,90.44) -- (108.24,90.24) -- (109.10,89.14) -- (109.63,90.44) -- cycle;
	\path[line width=0.21mm, draw=L] (101.26,91.51) .. controls (103.68,93.93) and (108.00,94.02) .. (109.95,91.03);
	\path[line width=0.21mm, draw=L, fill=F] (101.26,91.51) -- (102.80,91.80) -- (101.26,91.51) -- (101.95,92.91) -- (101.26,91.51) -- cycle;
	\path[line width=0.30mm, draw=L] (105.45,92.23);
	\draw(116.46,90.42) node[anchor=base west]{\fontsize{5.69}{6.83}\selectfont 2};
	\path[line width=0.30mm, draw=L, dash pattern=on 3pt off 2pt] (109.53,94.72);
	\path[line width=0.21mm, draw=L] (109.91,91.00) -- (116.20,91.00);
	\path[line width=0.21mm, draw=L, fill=F] (116.20,91.00) -- (114.99,91.70) -- (114.99,90.30) -- (116.20,91.00) -- cycle;
	\path[line width=0.30mm, draw=L, fill=F] (91.48,90.99) circle (0.38mm);
	\path[line width=0.30mm, draw=L, fill=F] (100.73,90.94) circle (0.38mm);
	\path[line width=0.30mm, draw=L, fill=F] (110.02,91.02) circle (0.38mm);
	\end{tikzpicture}}
\newcommand{\GraphThirty}{\begin{tikzpicture}[baseline=89.2mm,x=1.00mm, y=1.00mm, inner xsep=0pt, inner ysep=0pt, outer xsep=0pt, outer ysep=0pt]
	\useasboundingbox (81,83) rectangle +(29,15);
	\path[line width=0mm] (79.54,82.95) rectangle +(31.54,14.29);
	\definecolor{L}{rgb}{0,0,0}
	\path[line width=0.30mm, draw=L] (95.32,91.80);
	\draw(81.54,89.38) node[anchor=base west]{\fontsize{5.69}{6.83}\selectfont 1};
	\draw(108.08,89.39) node[anchor=base west]{\fontsize{5.69}{6.83}\selectfont 2};
	\path[line width=0.30mm, draw=L, dash pattern=on 3pt off 2pt] (95.53,92.84);
	\path[line width=0.21mm, draw=L] (100.36,89.90) .. controls (97.01,89.88) and (95.27,91.00) .. (95.25,94.06);
	\definecolor{F}{rgb}{0,0,0}
	\path[line width=0.21mm, draw=L, fill=F] (95.25,94.06) -- (94.73,92.59) -- (95.25,94.06) -- (96.12,92.76) -- (95.25,94.06) -- cycle;
	\path[line width=0.21mm, draw=L] (82.71,89.94) -- (90.28,89.94);
	\path[line width=0.21mm, draw=L, fill=F] (82.71,89.94) -- (84.11,89.24) -- (82.71,89.94) -- (84.11,90.64) -- (82.71,89.94) -- cycle;
	\path[line width=0.21mm, draw=L] (100.30,90.75) .. controls (99.82,93.20) and (97.80,94.99) .. (95.32,95.00) .. controls (92.83,95.00) and (90.71,93.19) .. (90.32,90.73);
	\path[line width=0.21mm, draw=L, fill=F] (100.30,90.75) -- (100.60,92.12) -- (99.26,91.69) -- (100.30,90.75) -- cycle;
	\path[line width=0.21mm, draw=L, fill=F] (90.32,90.73) -- (91.41,91.84) -- (90.32,90.73) -- (90.08,92.28) -- (90.32,90.73) -- cycle;
	\path[line width=0.21mm, draw=L] (90.30,90.01) .. controls (90.30,89.99) and (90.30,89.98) .. (90.30,89.97) .. controls (90.30,87.20) and (92.55,84.96) .. (95.32,84.95) .. controls (97.81,84.94) and (99.92,86.76) .. (100.28,89.22);
	\path[line width=0.21mm, draw=L, fill=F] (100.28,89.22) -- (99.28,88.24) -- (100.62,87.86) -- (100.28,89.22) -- cycle;
	\path[line width=0.21mm, draw=L] (100.24,89.92) -- (107.80,89.92);
	\path[line width=0.21mm, draw=L, fill=F] (107.80,89.92) -- (106.59,90.62) -- (106.59,89.22) -- (107.80,89.92) -- cycle;
	\path[line width=0.30mm, draw=L, fill=F] (90.30,89.95) circle (0.38mm);
	\path[line width=0.30mm, draw=L, fill=F] (100.33,89.97) circle (0.38mm);
	\path[line width=0.30mm, draw=L, fill=F] (95.28,94.86) circle (0.38mm);
	\end{tikzpicture}}
\newcommand{\GraphThirtyOpposite}{\begin{tikzpicture}[baseline=89.2mm,x=1.00mm, y=1.00mm, inner xsep=0pt, inner ysep=0pt, outer xsep=0pt, outer ysep=0pt]
	\useasboundingbox (81,83) rectangle +(30,15);
	\path[line width=0mm] (79.54,82.95) rectangle +(32.46,14.29);
	\definecolor{L}{rgb}{0,0,0}
	\path[line width=0.30mm, draw=L] (93.32,90.00);
	\draw(81.57,89.40) node[anchor=base west]{\fontsize{5.69}{6.83}\selectfont 1};
	\draw(108.99,89.36) node[anchor=base west]{\fontsize{5.69}{6.83}\selectfont 2};
	\path[line width=0.30mm, draw=L, dash pattern=on 3pt off 2pt] (92.31,90.32);
	\path[line width=0.21mm, draw=L] (95.20,94.11) .. controls (94.94,91.23) and (93.44,89.93) .. (90.26,89.96);
	\definecolor{F}{rgb}{0,0,0}
	\path[line width=0.21mm, draw=L, fill=F] (95.20,94.11) -- (94.30,93.03) -- (95.68,92.79) -- (95.20,94.11) -- cycle;
	\path[line width=0.30mm, draw=L] (95.32,91.80);
	\draw(81.54,89.38) node[anchor=base west]{\fontsize{5.69}{6.83}\selectfont 1};
	\path[line width=0.30mm, draw=L, dash pattern=on 3pt off 2pt] (95.53,92.84);
	\path[line width=0.21mm, draw=L] (82.79,89.88) -- (90.35,89.88);
	\path[line width=0.21mm, draw=L, fill=F] (82.79,89.88) -- (84.19,89.18) -- (82.79,89.88) -- (84.19,90.58) -- (82.79,89.88) -- cycle;
	\path[line width=0.21mm, draw=L] (100.37,90.73) .. controls (99.89,93.17) and (97.86,94.97) .. (95.39,94.97) .. controls (92.89,94.98) and (90.63,93.22) .. (90.38,90.71);
	\path[line width=0.21mm, draw=L, fill=F] (100.37,90.73) -- (100.66,92.10) -- (99.33,91.67) -- (100.37,90.73) -- cycle;
	\path[line width=0.21mm, draw=L, fill=F] (90.38,90.71) -- (91.44,91.86) -- (90.38,90.71) -- (90.09,92.24) -- (90.38,90.71) -- cycle;
	\path[line width=0.21mm, draw=L] (90.40,89.13) .. controls (90.78,86.71) and (92.87,84.93) .. (95.32,84.95) .. controls (97.97,84.97) and (100.14,87.03) .. (100.29,89.67);
	\path[line width=0.21mm, draw=L, fill=F] (90.40,89.13) -- (89.92,87.64) -- (90.40,89.13) -- (91.31,87.86) -- (90.40,89.13) -- cycle;
	\path[line width=0.21mm, draw=L] (100.24,89.92) -- (108.52,89.92);
	\path[line width=0.21mm, draw=L, fill=F] (108.52,89.92) -- (107.31,90.62) -- (107.31,89.22) -- (108.52,89.92) -- cycle;
	\path[line width=0.30mm, draw=L, fill=F] (90.30,89.95) circle (0.38mm);
	\path[line width=0.30mm, draw=L, fill=F] (100.33,89.97) circle (0.38mm);
	\path[line width=0.30mm, draw=L, fill=F] (95.28,94.86) circle (0.38mm);
	\end{tikzpicture}}
\newcommand{\GraphThirtyOne}{\begin{tikzpicture}[baseline=89.2mm,x=1.00mm, y=1.00mm, inner xsep=0pt, inner ysep=0pt, outer xsep=0pt, outer ysep=0pt]
	\useasboundingbox (89,83) rectangle +(13,15);
	\path[line width=0mm] (87.43,83.05) rectangle +(14.90,13.86);
	\draw(89.43,89.36) node[anchor=base west]{\fontsize{5.69}{6.83}\selectfont 1};
	\definecolor{L}{rgb}{0,0,0}
	\path[line width=0.21mm, draw=L, dash pattern=on 3pt off 2pt] (100.32,89.97) .. controls (100.27,92.71) and (98.04,94.90) .. (95.30,94.91) .. controls (92.85,94.91) and (90.77,93.13) .. (90.39,90.71);
	\path[line width=0.21mm, draw=L, dash pattern=on 3pt off 2pt] (90.37,89.18) .. controls (90.78,86.78) and (92.87,85.03) .. (95.30,85.05) .. controls (98.03,85.06) and (100.26,87.24) .. (100.33,89.97);
	\end{tikzpicture}}

\newcommand{\DashedLineOneTwo}{\begin{tikzpicture}[baseline=90mm,x=1.00mm, y=1.00mm, inner xsep=0pt, inner ysep=0pt, outer xsep=0pt, outer ysep=0pt]
	\useasboundingbox  (83.00,89.92) rectangle +(20,3.5);
	\path[line width=0mm] (82.89,87.97) rectangle +(19.63,8.61);
	\definecolor{L}{rgb}{0,0,0}
	\path[line width=0.30mm, draw=L] (93.00,92.90);
	\path[line width=0.30mm, draw=L, dash pattern=on 3pt off 2pt] (93.34,94.58);
	\path[line width=0.21mm, draw=L] (88.17,91.73) -- (89.26,90.98) -- (88.17,90.37);
	\path[line width=0.21mm, draw=L,dash pattern=on 3pt off 2pt] (85.98,91.03) -- (99.00,91.02);
	\draw(84.89,90.40) node[anchor=base west]{\fontsize{5.69}{6.83}\selectfont $\mathrm{1}$};
	\draw(99.52,90.43) node[anchor=base west]{\fontsize{5.69}{6.83}\selectfont $\mathrm{2}$};
	\definecolor{F}{rgb}{0,0,0}
	\path[line width=0.30mm, draw=L, fill=F] (92.53,91.10) circle (0.36mm);
	\path[line width=0.21mm, draw=L] (96.84,91.76) -- (95.75,91.01) -- (96.84,90.40);
	\end{tikzpicture}}

\newcommand{\GraphKOne}{\begin{tikzpicture}[baseline=89.2mm,x=1.00mm, y=1.00mm, inner xsep=0pt, inner ysep=0pt, outer xsep=0pt, outer ysep=0pt]
	\useasboundingbox  (83.5,88) rectangle +(16,11.5);
	\definecolor{L}{rgb}{0,0,0}
	\path[line width=0.30mm, draw=L] (92.30,91.88);
	\path[line width=0.30mm, draw=L, dash pattern=on 2.00mm off 1.00mm] (92.64,93.56);
	\path[line width=0.21mm, draw=L] (85.24,90.46) -- (91.27,96.49) -- (97.19,90.54);
	\definecolor{F}{rgb}{0,0,0}
	\path[line width=0.21mm, draw=L, fill=F] (85.24,90.46) -- (86.73,90.95) -- (85.24,90.46) -- (85.74,91.94) -- (85.24,90.46) -- cycle;
	\path[line width=0.21mm, draw=L, fill=F] (97.19,90.54) -- (96.83,91.89) -- (95.84,90.90) -- (97.19,90.54) -- cycle;
	\draw(83.84,89.46) node[anchor=base west]{\fontsize{5.69}{6.83}\selectfont 1};
	\draw(97.72,89.52) node[anchor=base west]{\fontsize{5.69}{6.83}\selectfont 2};
	\path[line width=0.30mm, draw=L, fill=F] (91.24,96.56) circle (0.38mm);
	\end{tikzpicture}}
\newcommand{\GraphKTwo}{\begin{tikzpicture}[baseline=89.2mm,x=1.00mm, y=1.00mm, inner xsep=0pt, inner ysep=0pt, outer xsep=0pt, outer ysep=0pt]
	\useasboundingbox  (83.5,88) rectangle +(16,16);
	\definecolor{L}{rgb}{0,0,0}
	\path[line width=0.30mm, draw=L] (92.30,91.88);
	\path[line width=0.30mm, draw=L, dash pattern=on 2.00mm off 1.00mm] (92.64,93.56);
	\path[line width=0.21mm, draw=L] (85.24,90.46) -- (91.27,96.49) -- (97.19,90.54);
	\definecolor{F}{rgb}{0,0,0}
	\path[line width=0.21mm, draw=L, fill=F] (85.24,90.46) -- (86.73,90.95) -- (85.24,90.46) -- (85.74,91.94) -- (85.24,90.46) -- cycle;
	\path[line width=0.21mm, draw=L, fill=F] (97.19,90.54) -- (96.83,91.89) -- (95.84,90.90) -- (97.19,90.54) -- cycle;
	\draw(83.84,89.46) node[anchor=base west]{\fontsize{5.69}{6.83}\selectfont 1};
	\draw(97.72,89.52) node[anchor=base west]{\fontsize{5.69}{6.83}\selectfont 2};
	\path[line width=0.30mm, draw=L, fill=F] (91.24,96.56) circle (0.38mm);
	\path[line width=0.21mm, draw=L] (91.58,96.78) -- (97.61,102.80);
	\path[line width=0.21mm, draw=L, fill=F] (91.58,96.78) -- (93.07,97.27) -- (91.58,96.78) -- (92.08,98.26) -- (91.58,96.78) -- cycle;
	\end{tikzpicture}}
\newcommand{\GraphKThree}{\begin{tikzpicture}[baseline=89.2mm,x=1.00mm, y=1.00mm, inner xsep=0pt, inner ysep=0pt, outer xsep=0pt, outer ysep=0pt]
	\useasboundingbox  (83.5,88) rectangle +(16,16);
	\definecolor{L}{rgb}{0,0,0}
	\path[line width=0.30mm, draw=L] (92.30,91.88);
	\path[line width=0.30mm, draw=L, dash pattern=on 2.00mm off 1.00mm] (92.64,93.56);
	\path[line width=0.21mm, draw=L] (85.24,90.46) -- (91.27,96.49) -- (97.19,90.54);
	\definecolor{F}{rgb}{0,0,0}
	\path[line width=0.21mm, draw=L, fill=F] (85.24,90.46) -- (86.73,90.95) -- (85.24,90.46) -- (85.74,91.94) -- (85.24,90.46) -- cycle;
	\path[line width=0.21mm, draw=L, fill=F] (97.19,90.54) -- (96.83,91.89) -- (95.84,90.90) -- (97.19,90.54) -- cycle;
	\draw(83.84,89.46) node[anchor=base west]{\fontsize{5.69}{6.83}\selectfont 1};
	\draw(97.72,89.52) node[anchor=base west]{\fontsize{5.69}{6.83}\selectfont 2};
	\path[line width=0.30mm, draw=L, fill=F] (91.24,96.56) circle (0.38mm);
	\path[line width=0.30mm, draw=L] (97.10,102.21);
	\path[line width=0.30mm, draw=L, dash pattern=on 2.00mm off 1.00mm] (97.44,103.89);
	\path[line width=0.21mm, draw=L] (91.58,96.89) -- (97.61,102.92);
	\path[line width=0.21mm, draw=L, fill=F] (91.58,96.89) -- (92.93,97.25) -- (91.94,98.24) -- (91.58,96.89) -- cycle;
	\path[line width=0.21mm, draw=L] (84.93,102.76) -- (90.85,96.82);
	\path[line width=0.21mm, draw=L, fill=F] (90.85,96.82) -- (90.36,98.30) -- (90.85,96.82) -- (89.37,97.31) -- (90.85,96.82) -- cycle;
	\end{tikzpicture}}
\newcommand{\GraphThirtyTwo}{\begin{tikzpicture}[baseline=90mm,x=1.00mm, y=1.00mm, inner xsep=0pt, inner ysep=0pt, outer xsep=0pt, outer ysep=0pt]
	\useasboundingbox  (83.00,89.92) rectangle +(20,10);
	\definecolor{L}{rgb}{0,0,0}
	\path[line width=0.30mm, draw=L] (93.00,92.90);
	\path[line width=0.30mm, draw=L, dash pattern=on 2.00mm off 1.00mm] (93.34,94.58);
	\path[line width=0.21mm, draw=L] (89.78,91.78) -- (88.69,91.03) -- (89.78,90.42);
	\path[line width=0.21mm, draw=L] (85.01,91.05) -- (101.61,91.05);
	\definecolor{F}{rgb}{0,0,0}
	\path[line width=0.30mm, draw=L, fill=F] (93.35,91.08) circle (0.38mm);
	\path[line width=0.21mm, draw=L] (98.52,91.78) -- (97.43,91.03) -- (98.52,90.42);
	\path[line width=0.30mm, draw=L] (93.28,96.98);
	\path[line width=0.30mm, draw=L, dash pattern=on 2.00mm off 1.00mm] (93.62,98.65);
	\path[line width=0.21mm, draw=L] (93.39,91.44) -- (93.40,98.70);
	\path[line width=0.21mm, draw=L] (92.67,95.61) -- (93.41,94.52) -- (94.03,95.61);
	\draw(83.93,90.42) node[anchor=base west]{\fontsize{5.69}{6.83}\selectfont 1};
	\draw(101.68,90.42) node[anchor=base west]{\fontsize{5.69}{6.83}\selectfont 2};
	\end{tikzpicture}}
\newcommand{\GraphThirtyThree}{\begin{tikzpicture}[baseline=90mm,x=1.00mm, y=1.00mm, inner xsep=0pt, inner ysep=0pt, outer xsep=0pt, outer ysep=0pt]
	\useasboundingbox  (83.00,89.92) rectangle +(20,8);
	\definecolor{L}{rgb}{0,0,0}
	\path[line width=0.30mm, draw=L] (92.06,92.02);
	\path[line width=0.30mm, draw=L, dash pattern=on 2.00mm off 1.00mm] (91.20,93.50);
	\path[line width=0.21mm, draw=L] (89.78,91.78) -- (88.69,91.03) -- (89.78,90.42);
	\path[line width=0.21mm, draw=L] (85.01,91.05) -- (101.61,91.05);
	\definecolor{F}{rgb}{0,0,0}
	\path[line width=0.30mm, draw=L, fill=F] (93.35,91.08) circle (0.38mm);
	\path[line width=0.21mm, draw=L] (98.52,91.78) -- (97.43,91.03) -- (98.52,90.42);
	\path[line width=0.30mm, draw=L] (89.56,95.25);
	\path[line width=0.30mm, draw=L, dash pattern=on 2.00mm off 1.00mm] (88.70,96.73);
	\path[line width=0.21mm, draw=L] (93.32,91.19) -- (88.51,96.62);
	\path[line width=0.21mm, draw=L] (90.01,93.82) -- (91.29,93.50) -- (91.03,94.73);
	\draw(83.93,90.42) node[anchor=base west]{\fontsize{5.69}{6.83}\selectfont 1};
	\draw(101.68,90.42) node[anchor=base west]{\fontsize{5.69}{6.83}\selectfont 2};
	\path[line width=0.30mm, draw=L] (94.74,91.94);
	\path[line width=0.30mm, draw=L, dash pattern=on 2.00mm off 1.00mm] (95.59,93.43);
	\path[line width=0.30mm, draw=L] (97.23,95.18);
	\path[line width=0.30mm, draw=L, dash pattern=on 2.00mm off 1.00mm] (98.09,96.66);
	\path[line width=0.21mm, draw=L] (93.48,91.12) -- (98.28,96.54);
	\path[line width=0.21mm, draw=L] (95.46,94.31) -- (96.70,94.78) -- (96.58,93.53);
	\end{tikzpicture}}
\newcommand{\GraphThirtyFour}{\begin{tikzpicture}[baseline=90mm,x=1.00mm, y=1.00mm, inner xsep=0pt, inner ysep=0pt, outer xsep=0pt, outer ysep=0pt]
	\useasboundingbox  (83.00,89.92) rectangle +(20,10);
	\definecolor{L}{rgb}{0,0,0}
	\path[line width=0.30mm, draw=L] (89.61,92.84);
	\path[line width=0.30mm, draw=L, dash pattern=on 2.00mm off 1.00mm] (89.96,94.52);
	\path[line width=0.21mm, draw=L] (93.77,91.78) -- (92.68,91.03) -- (93.77,90.42);
	\path[line width=0.21mm, draw=L] (85.01,91.05) -- (101.61,91.05);
	\definecolor{F}{rgb}{0,0,0}
	\path[line width=0.30mm, draw=L, fill=F] (89.97,91.02) circle (0.38mm);
	\path[line width=0.21mm, draw=L] (100.02,91.80) -- (98.93,91.05) -- (100.02,90.44);
	\path[line width=0.30mm, draw=L] (89.89,96.92);
	\path[line width=0.30mm, draw=L, dash pattern=on 2.00mm off 1.00mm] (90.23,98.60);
	\path[line width=0.21mm, draw=L] (90.01,91.39) -- (90.01,98.64);
	\path[line width=0.21mm, draw=L] (89.28,95.56) -- (90.03,94.47) -- (90.64,95.56);
	\draw(83.93,90.42) node[anchor=base west]{\fontsize{5.69}{6.83}\selectfont 1};
	\draw(101.68,90.42) node[anchor=base west]{\fontsize{5.69}{6.83}\selectfont 2};
	\path[line width=0.30mm, draw=L] (96.09,92.86);
	\path[line width=0.30mm, draw=L, dash pattern=on 2.00mm off 1.00mm] (96.43,94.54);
	\path[line width=0.30mm, draw=L, fill=F] (96.44,91.04) circle (0.38mm);
	\path[line width=0.30mm, draw=L] (96.37,96.94);
	\path[line width=0.30mm, draw=L, dash pattern=on 2.00mm off 1.00mm] (96.71,98.62);
	\path[line width=0.21mm, draw=L] (96.48,91.41) -- (96.49,98.66);
	\path[line width=0.21mm, draw=L] (95.76,94.48) -- (96.50,95.58) -- (97.12,94.48);
	\path[line width=0.21mm, draw=L] (87.84,91.76) -- (86.75,91.01) -- (87.84,90.40);
	\end{tikzpicture}}
\newcommand{\GraphThirtyFourOpposite}{\begin{tikzpicture}[baseline=90mm,x=1.00mm, y=1.00mm, inner xsep=0pt, inner ysep=0pt, outer xsep=0pt, outer ysep=0pt]
	\useasboundingbox  (83.00,89.92) rectangle +(20,10);
	\definecolor{L}{rgb}{0,0,0}
	\path[line width=0.30mm, draw=L] (89.61,92.84);
	\path[line width=0.30mm, draw=L, dash pattern=on 2.00mm off 1.00mm] (89.96,94.52);
	\path[line width=0.21mm, draw=L] (93.77,91.78) -- (92.68,91.03) -- (93.77,90.42);
	\path[line width=0.21mm, draw=L] (85.01,91.05) -- (101.61,91.05);
	\definecolor{F}{rgb}{0,0,0}
	\path[line width=0.30mm, draw=L, fill=F] (89.97,91.02) circle (0.38mm);
	\path[line width=0.21mm, draw=L] (100.02,91.80) -- (98.93,91.05) -- (100.02,90.44);
	\path[line width=0.30mm, draw=L] (89.89,96.92);
	\path[line width=0.30mm, draw=L, dash pattern=on 2.00mm off 1.00mm] (90.23,98.60);
	\path[line width=0.21mm, draw=L] (90.01,91.39) -- (90.01,98.64);
	\path[line width=0.21mm, draw=L] (89.28,94.47) -- (90.03,95.56) -- (90.64,94.47);
	\draw(83.93,90.42) node[anchor=base west]{\fontsize{5.69}{6.83}\selectfont 1};
	\draw(101.68,90.42) node[anchor=base west]{\fontsize{5.69}{6.83}\selectfont 2};
	\path[line width=0.30mm, draw=L] (96.09,92.86);
	\path[line width=0.30mm, draw=L, dash pattern=on 2.00mm off 1.00mm] (96.43,94.54);
	\path[line width=0.30mm, draw=L, fill=F] (96.44,91.04) circle (0.38mm);
	\path[line width=0.30mm, draw=L] (96.37,96.94);
	\path[line width=0.30mm, draw=L, dash pattern=on 2.00mm off 1.00mm] (96.71,98.62);
	\path[line width=0.21mm, draw=L] (96.48,91.41) -- (96.49,98.66);
	\path[line width=0.21mm, draw=L] (95.76,95.58) -- (96.50,94.48) -- (97.12,95.58);
	\path[line width=0.21mm, draw=L] (87.84,91.76) -- (86.75,91.01) -- (87.84,90.40);
	\end{tikzpicture}}
 \theoremstyle{plain}
  \newtheorem{thm}{Theorem}[section]
  \newtheorem{prop}[thm]{Proposition}
  \newtheorem{cor}[thm]{Corollary}
  \newtheorem{lemma}[thm]{Lemma}
  
 \theoremstyle{definition}
  \newtheorem{df}[thm]{Definition}
 \theoremstyle{remark}
  \newtheorem{rem}[thm]{Remark}
  \newtheorem{exa}[thm]{Example}
 \numberwithin{equation}{section}

\author{Eli Hawkins, Kasia Rejzner}
\title{The Star Product in\\ Interacting Quantum Field Theory}
\date{}
 
\begin{document}
 \sloppy
\maketitle
\begin{center}
\vspace{-4ex}
\emph{\small Department of Mathematics}\\
\emph{\small The University of York, United Kingdom}\\
{\small eli.hawkins@york.ac.uk}\\
{\small kasia.rejzner@york.ac.uk}\\\end{center}
\begin{abstract}
	We propose a new formula for the star product in deformation quantization of  Poisson structures  related in a specific way to a variational problem for a function $S$, interpreted as the action functional. Our approach is motivated by perturbative Algebraic Quantum Field Theory (pAQFT). We provide a direct combinatorial formula for the star product and we show that it can be applied to a certain class of infinite dimensional manifolds (e.g., regular observables in pAQFT). This is the first step towards understanding how pAQFT can be formulated such that the only formal parameter is $\hbar$, while the coupling constant can be treated as a number.

	In the introductory part of the paper, apart from reviewing the framework, we make precise several statements present in the pAQFT literature and  recast these in the language of (formal) deformation quantization. 
	Finally, we use our formalism to streamline the proof of perturbative agreement provided by Drago, Hack, and Pinamonti and to generalize some of the results obtained in that work to the case of a non-linear interaction. 
\end{abstract}
\tableofcontents

\section{Introduction}
Constructing interacting quantum field theory (QFT) models in 4 dimensions is one of the most important challenges facing  modern theoretical physics. Even though there is no final consensus on how the actual axiomatic framework underlying QFT should be formulated, most attempts at construction of models try to fit into one of the established axiomatic systems: Wightman-G{\r{a}}rding \cite{StreaterWightman}, Haag-Kastler \cite{HK,Haag} or Osterwalder-Schrader \cite{OS73}. On the other hand, most computations in QFT are done using less rigorous methods and  often rely on a perturbation theory expansion organized in terms of Feynman graphs. 

A new approach that combines the advantages of a mathematically sound axiomatic framework with perturbative methods has emerged in the last two decades; it is called perturbative algebraic quantum field theory (pAQFT). The foundations were laid in  \cite{DF,DF02,DF04,DF05,DFloop,BDF} and further results concerning fermionic fields and gauge theory were obtained in \cite{Rej11a,FR,FR3}. For reviews see \cite{Book} and \cite{Due19}.

One of the ingredients of pAQFT is formal deformation quantization --- the construction of a (quantum) {\it star product} for the algebra of observables \cite{BFFLS1,BFFLS2}. The first application of deformation quantization to QFT appears in the works of Dito \cite{Dit90,Dit93}.

There are two well-known general constructions of star products on finite-dimensional manifolds: The Fedosov construction \cite{Fedosov} starts from a symplectic manifold with some torsion-free connection; the Kontsevich construction  \cite{Kon} starts from a Poisson structure on an affine space. The former has been generalized to the infinite dimensional setting by Collini \cite{Collini}. In this approach (commonly referred to as the ``on-shell approach''), the equations of motion for the fields are imposed before quantization and the space of solutions is a sort of infinite-dimensional symplectic manifold.

Here, we choose to work in the off-shell approach, i.e. the equations of motion are imposed after quantization. The space of field configurations is a sort of infinite-dimensional \textit{Poisson} manifold, which suggests that something in the spirit of Kontsevich is called for. However, the Kontsevich formula does not carry over verbatim to the infinite dimensional case \cite{Dit15}. This motivates the search for an alternative formulation and a new formula for the star product that can be applied to infinite dimensional Poisson manifolds, at least in some class of interesting examples.

A possible way forward has been suggested in \cite{FR3}, where an interacting star product for a sufficiently regular class of functionals has been constructed by an indirect method. There, the theory was formulated in terms of formal power series in $\hbar$ and the coupling constant $\lambda$. In this (perturbative) approach, the starting point is an affine configuration space with an action functional that is split into a preferred quadratic (free) part and the remainder. The free part is quantized with the help of a Hadamard distribution (this notion mimics the properties of the 2-point function of the Minkowski vacuum state). This is then translated into a star product for the full theory by using quantum M{\o}ller operators and the algebraic adiabatic limit. 

As far as the relation to Kontsevich's approach is concerned, we note that the product introduced in \cite{FR3} cannot arise from Kontsevich's construction, because the first order (in $\hbar$) part of the star product is not proportional to the Poisson structure (instead, it is given in terms of the Hadamard function that differs from the Poisson structure by a symmetric bidistribution). A plausible generalization of Kontsevich's construction would be one that constructs the star product from this first order term and the affine structure. We will show here that such a construction cannot give the desired star product.

What then? In field theory, the Poisson structure is not a given structure, but follows from the action, hence it is not surprising that the action is an appropriate ingredient for a construction of the star product. The affine structure is an acceptable ingredient as well. The preferred quadratic part, however, is an undesirable ingredient, because we would like to achieve a nonperturbative construction.

The purpose of the present paper is to work towards a direct construction of the star product for infinite dimensional affine manifolds, where the Poisson structure follows from the given action. The direct formula we propose in this work agrees with the one of \cite{FR3} on the perturbative level, but it goes beyond that. What we achieve here is a construction of the interacting star product that depends only on the full action, the causal structure, and the affine structure of the configuration space. This is a desirable result from the pAQFT perspective, since in many situations the split into free and interacting theory is unnatural and the physical results should not depend on this split. This is the case, for example, in quantum gravity, as indicated in \cite{BFRej13}. Moreover, our construction is nonperturbative in the sense that it does not require a preferred quadratic action and only $\hbar$ is treated as a formal parameter. 

We start our construction by first transforming the free star product to remove the dependence on a Hadamard distribution. This means that we treat the quantization relative to the time-ordered product, rather than an (unphysical in the quantum context) pointwise product. Unfortunately, in the infinite dimensional case, this transformation only works if the interaction and observables are sufficiently regular (for field theories in 3 or more spacetime dimensions, this implies that all their functional derivatives have to be given by smooth compactly supported functions). Nevertheless, we believe that the construction and its combinatorial structure may be relevant for a nonperturbative construction of the star product for a more realistic action.  As a byproduct, we also show how to nonperturbatively construct the quantum M{\o}ller operator for regular observables from the classical M{\o}ller operator. (The existence of classical M{\o}ller operators and of retarded and advanced Green functions for regular interactions is implied by the results in the paper \cite{FV}, which is expected to appear soon.) Our main results are Theorems  \ref{Reduced star V} and \ref{non:pert:H}.

Our construction works in the finite-dimensional case without any restrictions and it would be interesting to investigate the relation between our star product and other known constructions, e.g., in geometric quantization. We will address this point in our future research. Another potential realm of applications is field theories in two spacetime dimensions, where the singularity of the Feynman propagator is logarithmic and we expect that a larger class of observables and interactions can be treated using our methods. Recently, a construction of Sine Gordon model in 2D was achieved by \cite{BR16} using pAQFT methods, so it would be interesting to apply our formula for the star product to that case and to other theories, where exact results are expected to hold.



The paper is organized as follows. In Sections 2--3 we review the construction of classical field theory in the pAQFT framework and provide a more rigorous proof of the result of \cite{DF02,BreDue} that the retarded M{\o}ller map intertwines between the Peierls brackets of the free and interacting theories. In Section 4 we discuss deformation quantization. We introduce several natural quantization maps, useful in pAQFT and prove a theorem that completely characterizes the ambiguity in constructing the time-ordering operator on regular functionals. This result mimics the Main Theorem of Renormalization proven for local functionals in \cite{BDF}. In Section 5 we introduce the formal S-matrix and the quantum M{\o}ller operator and in Proposition~\ref{quantum correction} we show how the latter can be constructed nonperturbatively on regular functionals from the  classical M{\o}ller operator. Next we prove a direct formula, eq.~\eqref{reduced star V}, for the interacting star product and show that it makes sense nonperturbatively in the coupling constant. In this section we also discuss the relation to the formula of Kontsevich, provide some useful formulae for the quantum M{\o}ller operator, and finally we discuss the principle of perturbative agreement. 

\section{Kinematical structure}\label{class}
In the framework of perturbative algebraic quantum field theory (pAQFT) one starts with the classical theory, which is subsequently quantized. We work in the Lagrangian framework, but there are some modifications that we need to make to deal with the infinite dimensional character of field theory. In this section we give an overview of mathematical structures that will be needed later on to construct models of classical and quantum field theories. Since we do not fix the dynamics yet, the content of this section describes the kinematical structure of our model. 
\subsection{The space of field configurations}\label{ConfigSpace}
We start with a globally hyperbolic spacetime $\Mcal=(M,\mathbf g)$. Next we introduce $\Ecal$, the space of field configurations. The choice of $\Ecal$ specifies what kind of objects our model describes (e.g., scalar fields, gauge fields, etc.). 
\begin{df}\label{conf} 
The configuration space $\Ecal$ on the fixed spacetime $\Mcal=(M,\mathbf g)$ is realized as the space of smooth sections $\Gamma(E\rightarrow M)$ of some vector bundle $E\xrightarrow{\pi} M$ over $M$.  \end{df}

\begin{df}
The space of compactly supported sections is denoted $\Ecal_c=\Gamma_c(E\rightarrow M)$. The space of sections of the dual bundle is denoted $\Ecal^*=\Gamma(E^*\rightarrow M)$.
\end{df}

\subsection{Functionals on the configuration space}
We model classical and quantum observables as smooth (in the sense of  \cite{Bas64,Ham,Mil,Neeb}) functionals on $\Ecal$. 

Consider some $F:\Ecal\to\CC$ and $\ph\in\Ecal$. We require $F$ to be smooth in the sense of Bastiani  \cite{Bas64,Ham,Mil,Neeb}, i.e.
\begin{df}
	Let $\Xcal$ and $\Ycal$ be topological vector spaces, $U \subseteq \Xcal$ an open set and $f:U \rightarrow \Ycal$ a map. The derivative of $f$ at $x\in U$ in the direction\index{derivative!on a locally convex vector space} of $h\in\Xcal$ is defined as
	\be\label{de}
	\left<f^{(1)}(x),h\right> \doteq \lim_{t\rightarrow 0}\frac{1}{t}\left(f(x + th) - f(x)\right)
	\ee
	whenever the limit exists. The function $f$ is called differentiable\index{infinite dimensional!calculus} at $x$ if $\left<f^{(1)}(x),h\right>$ exists for all $h \in \Xcal$. It is called continuously differentiable if it is differentiable at all points of $U$ and
	$f^{(1)}:U\times \Xcal\rightarrow \Ycal, (x,h)\mapsto f^{(1)}(x)(h)$
	is a continuous map. It is called a $\Ccal^1$-map if it is continuous and continuously differentiable. Higher derivatives are defined by
	\be
	\left<f^{(k)}(x),v_1\otimes\dots\otimes v_k\right>\doteq\left. \frac{\partial^k}
	{\partial t_1\dots\partial t_k}f(x+t_1 v_1+\dots + t_k v_k)\right|_{t_1=\dots=t_k=0},
	\ee
	and $f$ is $\Ccal^k$ if $f^{(k)}$ is jointly  continuous as a map $U\times \Xcal^k\rightarrow \Ycal$. We say that $f$ is smooth if it is $\Ccal^k$ for all $k\in\NN$. 
\end{df}
For a detailed discussion of Bastiani smoothness in the context of classical field theory, see e.g., \cite{BDGR}.
By definition, if  $F^{(1)}(\ph)$ exists, then it is an element of the complexified dual space ${\Ecal'}^{\sst\CC}\doteq \Ecal'\otimes\CC$. More generally, the $n$'th derivative defines an element $F^{(n)}(\ph)$ of the continuous complex dual of the completed projective tensor product $\Ecal^{\hat{\otimes}_\pi n}\cong \Gamma(E^{\boxtimes n}\rightarrow M^n)$, where $\boxtimes$ is the exterior tensor product of vector bundles, i.e. $F^{(n)}(\ph)\in {(\Ecal^{\hat{\otimes}_\pi n})'}^{\sst\CC}$ (hence it is a distribution with compact support).
\begin{df}\label{SpacetimeSup}
The \emph{spacetime support} of a function, $F:\Ecal\to S$ (where $S$ may be any set)  is defined by
\begin{align}\label{support}
\supp F\doteq\{ x\in M\mid\forall \text{open }U\ni x\ \exists \ph,\psi\in\Ecal, \supp\psi\subset U,F(\ph+\psi)\not= F(\ph)\}\ .\nonumber
\end{align}
Alternatively, one can write this as:
\[
\supp F=\overline{\bigcup\limits_{\ph\in\Ecal}\supp(F^{(1)}(\ph))}\,.
\]
\end{df}
\begin{df}\label{localfunctionals}
A functional $F\in\Ci(\Ecal,\CC)$ is called \emph{local}  if for each $\ph_0\in\Ecal$ there exists an open neighbourhood $O$ in $\Ecal$ and $k\in\NN$ such that for all $\ph\in O$ we have
	\be
	F(\ph)=\int_{M} \alpha(j^k_x(\ph))\ ,
	\ee
	where $j^k_x(\ph)$ is the $k$'th jet prolongation of $\ph$ and $\alpha$ is a map over $M$ from the jet bundle to the volume-form bundle. We denote the space of local functionals by $\Fcal_\loc$.
\end{df}
We equip the space $\Fcal_{\loc}$ of local functionals on the configuration space  with the pointwise product using the prescription
\be\label{pointwise}
(F\cdot G)(\ph)\doteq F(\ph)G(\ph)\,,
\ee
where $\ph\in\Ecal$. $\Fcal_{\loc}$ is not closed under this product, but we can consider instead the space $\Fcal$ of \emph{multilocal functionals}, which is defined as the algebraic closure of $\Fcal_{\loc}$ under the product \eqref{pointwise}. We can also introduce the involution operator $*$ on $\Fcal$ using complex conjugation, i.e.,
\[
F^*(\ph)\doteq \overline{F(\ph)}\,.
\]
In this way we obtain a commutative $*$-algebra. 

Functional derivatives of smooth functionals on $\Ecal$ are compactly supported distributions. We can distinguish certain important classes of functionals by analyzing the wavefront (WF) set properties of their derivatives. 

Local and multilocal functionals satisfy some important regularity properties. Firstly, for local functionals the wavefront set of $F^{(n)}(\ph)$ is orthogonal to the tangent bundle of the \emph{thin diagonal},
\[
\{(x,\dots,x)\in M^n\mid x\in M\} .
\]
In particular, $F^{(1)}(\ph)$ has empty wavefront set, and so is smooth for each fixed  $\ph\in\Ecal$. The latter is true also  for multilocal functionals, i.e., $F\in\Fcal$. Note that using the metric volume form $\mu_{\mathbf g}$ we can therefore identify $F^{(1)}(\ph)$ with an element of ${\Ecal^*}^{\sst\CC}$.

\begin{df}
	A functional $F\in\Ci(\Ecal,\CC)$ is called \emph{regular} ($F\in \Fcal_{\reg}$) if $F^{(n)}(\ph)$ has empty WF set for all $n\in\NN$, $\ph\in\Ecal$.
\end{df}
\begin{rem}
Some authors additionally require that a regular functional have compact support. This definition only implies that each derivative has compact support (the support of $F^{(1)}(\ph)$ can change with $\ph$, so $\supp F$ may not be compact).
\end{rem}
For a regular functional, $F^{(n)}(\ph)$ can be identified with an element of $\Gamma_c((E^*)^{\boxtimes n}\to M^n)^{\sst \CC}$.

\section{Classical theory}\label{ClassicalTheory}
\subsection{Dynamics}
Dynamics is introduced in the Lagrangian framework. We begin with recalling some crucial definitions after \cite{BDF}. Here, $\Dcal(M)$ denotes the space of smooth, compactly supported, real-valued functions (test functions).
\begin{df}\label{Lagr}
	A \textit{generalized Lagrangian} on a fixed spacetime $\Mcal=(M,\mathbf g)$ is a map $L:\Dcal(M)\rightarrow\Fcal_{\loc}$ such that
	\begin{enumerate}[i)]
		\item $L(f+g+h)=L(f+g)-L(g)+L(g+h)$ for $f,g,h\in\Dcal(M)$ with $\supp\,f\cap\supp\,h=\varnothing$ ({Additivity}).
		\item $\supp(L(f))\subseteq \supp(f)$ ({Support}).
		\item Let $\Iso(\Mcal)$ be the oriented and time oriented isometry group of the spacetime $\Mcal$. (For Minkowski spacetime $\Iso(\Mcal)$ is the proper orthochronous Poincar\'e group $\Pcal^\uparrow_+$.) We require that $L(f)(u^*\ph)=L(u_*f)(\ph)$ for every $u\in\Iso(\Mcal)$ ({Covariance}).
	\end{enumerate}
\end{df}
%

\begin{df}\label{ClassAction}
 An \emph{action} is an equivalence class of Lagrangians under the  equivalence relation \cite{BDF}
 \be\label{equi}
L_1\sim L_2 \quad\textrm{iff}\quad\supp ((L_1-L_2)(f))\subset \supp df\,.
 \ee 
 \end{df}
The physical meaning of \eqref{equi} is to identify Lagrangians that ``differ by a total divergence''. 

An action is something like a local functional with (possibly) noncompact support. It is used just as $L(1)$ would be if it existed.
\begin{df}\label{ELDerivative}
 The \emph{Euler-Lagrange derivative} (first variational derivative) of $L$ is a map $L':\Ecal\rightarrow {\Ecal'_c}$ defined by
 \[
 \left<L'(\ph),h\right>\doteq \left<L(f)^{(1)}(\ph),h\right>\,,
 \]
 where $h\in\Ecal_c$ and $f\in\Dcal(M)$ is chosen in such a way that $f= 1$ on $\supp h$.
 \end{df}
 Since $L(f)$ is a local functional, $L'$ doesn't depend on the choice of $f$. 
  Note that two Lagrangians equivalent under the relation \eqref{equi} induce the same Euler-Lagrange derivative, so dynamics is a structure coming from actions rather than Lagrangians. 
\begin{df}
The Euler-Lagrange derivative of an action, $S$, is $S'\doteq L'$ for any Lagrangian $L\in S$.
\end{df}  
  We are now ready to introduce the equations of motion ({\eom}'s).
\begin{df}
The equation of motion ({\eom}) corresponding to the action $S$ is
\be\label{eom}
S'(\ph)= 0\,,
\ee
understood as a condition on $\ph\in\Ecal$.
\end{df}
\begin{rem}
The space of solutions of \eqref{eom} may be pathological. Instead, in this algebraic setting, we should work with the quotient of $\Fcal$ or $\Fcal_\reg$ by the ideal generated by $S'$. This plays the \emph{role} of the algebra of functionals on the space of solutions.

However, in this paper we are concerned with ``off shell'' constructions, i.e., the quotient is not taken.
\end{rem}

 \begin{df}
The second variational derivative $S''$ of the action $S$ is defined by
\[
\left<S''(\ph),\psi_1\otimes \psi_2\right>\doteq \left<L(f)^{(2)}(\ph),\psi_1\otimes \psi_2\right>\,,
\]
where $L\in S$ and $f= 1$ on $\supp \psi_1\cup\supp \psi_2$. 
\end{df}
By definition,  $S''(\ph):\Ecal_c\times\Ecal_c\to\CC$  is a bilinear map, due to locality it can in fact be extended to a map $S''(\ph):\Ecal\times\Ecal_c\to\CC$, which then
induces a continuous linear operator $P_S(\ph):\Ecal^{\sst \CC}\rightarrow{\Ecal^*}^{\sst \CC}$.
Note that if $S$ is quadratic then $P_S\doteq P_S(\ph)$ is the same for all $\ph$ and $S'(\ph)=P_S\ph$. This is the case for the free scalar field, where $P_S=-(\Box+ m^2)$.

The crucial assumption in the pAQFT approach is that $P_S(\ph)$ is a Green-hyperbolic operator \cite{GreenBear}, i.e. that it admits unique retarded and advanced Green's functions (fundamental solutions) $\Delta_S^{\mathrm{R}}(\ph),\ \Delta_S^{\mathrm{A}}(\ph):{\Ecal^*_c}^{\sst \CC}\rightarrow\Ecal^{\sst \CC}$ defined by the requirements
\begin{align*}
P_S(\ph)\circ\Delta^{\mathrm{R}/\mathrm{A}}_S(\ph)&=\id\,,\\
\left.\Delta^{\mathrm{R}/\mathrm{A}}_S(\ph)\circ P_S(\ph)\right|_{\Ecal_c}&=\id\,,
\end{align*}
and the support properties
\begin{align*}
\supp \Delta^{\mathrm{R}}_S(\ph)(\psi)&\subset J^+(\supp \psi)\,,\\
\supp \Delta^{\mathrm{A}}_S(\ph)(\psi)&\subset J^-(\supp \psi)\,,
\end{align*}
where $\psi\in{\Ecal^*}^{\sst \CC}_c$. Note that, by the Schwartz kernel theorem, these operators can be written in terms of their integral kernels, which then satisfy appropriate support properties and
\be\label{AandR}
\Delta^{\mathrm{R}}_S(\ph)(y,x)=\Delta^{\mathrm{A}}_S(\ph)(x,y)\,.
\ee

The \emph{causal propagator} is
\be\label{causalprop}
\Delta_S(\ph)\doteq\Delta_S^{\mathrm{R}}(\ph)-\Delta_S^{\mathrm{A}}(\ph)\,.
\ee
Due to \eqref{AandR} the causal propagator is antisymmetric, i.e., its integral kernel satisfies
\[
\Delta_S(\ph)(y,x)=-\Delta_S(\ph)(x,y)\,.
\]
We equip the algebra of functionals  with a Poisson bracket called the \emph{Peierls bracket} \cite{Pei}. 
\begin{df}\label{PeierlsBracket}
The \emph{Peierls bracket} of $F,G\in\Fcal$ is defined by
\be\label{Pei}
\Pei{F}{G}_S (\ph)\doteq \left<\Delta_S(\ph),F^{(1)}(\ph)\otimes G^{(1)}(\ph)\right>\,,
\ee
where $\Delta_S(\ph)$ is used as a complex bidistribution in $\Gamma'_c((E^*)^{\boxtimes 2}\rightarrow M^2)^{\sst \CC}$ (this can be understood as a completion of $\Ecal_c\otimes\Ecal_c$, i.e. $\Delta_S$ is sort of a section of the an appropriate completion of $T^2\Ecal$, i.e. a bivector field).
\end{df}
\begin{rem}
The  algebra of multi-local functionals, $\Fcal$, is not closed under the Peierls bracket. It works better on the algebra of microcausal functionals, $\Fcal_\mc$, defined later on. For a quadratic action, the Peierls bracket gives a Poisson bracket on the algebra of  regular functionals, $\Fcal_\reg$.
\end{rem}

In this paper, we will mainly consider an ``action'' of the form
\[
S = S_0 + \la V
\]
where $S_0$ is a quadratic (in $\ph$) action and $V\in\Fcal_\reg$ is compactly supported. This is not an action in the sense of Definition~\ref{ClassAction}, because (unless it is quadratic) $V$ cannot be expressed as an equivalence class of Lagrangians, cannot be cut off with a test function, and is nonlocal. Nevertheless,  it can be used in much the same way as an action. For example, the Euler-Lagrange derivative should be understood as $S'=S_0'+\la V^{(1)}$. 
\begin{rem}
The $\la V$ is used as a regularized interaction term. The fact that $V$ has compact support is a kind of infrared regularization, and the fact that it is a regular (rather than local) functional is an ultraviolet (UV) regularization. This type of IR regularization is the usual technique used in Epstein Glaser renormalization \cite{EG}.
\end{rem}

With this, the definitions of the propagators must be modified slightly (this is equivalent to the definition given in Lemma 1 in \cite{BreDue}):
\[
\supp \Delta_S^{\mathrm{R/A}}(\ph)(f) \subseteq J^\pm(\supp f \cup \supp V) \,.
\]

If $\la$ is a formal parameter, it is easy to see that 
\begin{align}
\Delta^{\mathrm A}_S &= \sum_{n=0}^\infty (-\la)^n \Delta^{\mathrm A}_{S_0} \left( V^{(2)} \Delta^{\mathrm A}_{S_0}\right)^n \label{regular advanced}\\
&= \Delta^{\mathrm A}_{S_0}-\la \Delta^{\mathrm A}_{S_0} V^{(2)} \Delta^{\mathrm A}_{S_0} + \la^2 \Delta^{\mathrm A}_{S_0} V^{(2)} \Delta^{\mathrm A}_{S_0} V^{(2)} \Delta^{\mathrm A}_{S_0} - \dots \nonumber \,.
\end{align}
satisfies this definition, and there is an analogous formula for $\Delta_S^{\mathrm R}$.

\begin{rem}
To reduce clutter, we will often omit the $\circ$ symbol when composing linear operators, as in eq.~\eqref{regular advanced}.
\end{rem}

\subsection{Classical M{\o}ller maps off-shell}\label{Moll}
To avoid  functional analytic difficulties we define all the structures only for regular functionals $\Fcal_{\reg}$ or local functionals $\Fcal_{\loc}$. We mostly work perturbatively, i.e., we use formal power series in $\la$, which plays the role of a coupling constant.

From now on, we fix a free (quadratic) action $S_0$ and consider ``actions'' of the form $S=S_0+\la V$ with $V\in\Fcal_\reg$. For two such actions, the retarded and advanced M{\o}ller maps ${\mathtt r}_{S_1,S_2},{\mathtt a}_{S_1,S_2}:\Ecal_{S_2} \to \Ecal_{S_1}$ are defined on shell by the requirements that for any $\ph\in\Ecal_{S_2}$, 
$\ph-{\mathtt r}_{S_1,S_2}(\ph)$
has past-compact support and $\ph-{\mathtt a}_{S_1,S_2}(\ph)$ has future-compact support. We will be using the retarded maps, but everything we say adapts easily to the advanced maps.

We will need the extension of this off-shell. Following \cite{DF02}, the off-shell retarded M{\o}ller map is defined by the conditions
\be\label{compo}
{\mathtt r}_{S_1,S_2}\circ {\mathtt r}_{S_2,S_3}={\mathtt r}_{S_1,S_3}
\ee
and 
\be
\label{first order}
\left.\frac{d}{d\la} {\mathtt r}_{S+\la V,S}(\ph)\right|_{\la=0} = -\Delta^{\mathrm R}_S(\ph) V^{(1)}(\ph) \,.
\ee

These retarded M{\o}ller maps are perturbatively well defined, and restriction to solutions gives the on-shell M{\o}ller maps \cite{DF02}. To simplify the notation we abbreviate ${\mathtt r}_{\la V}\doteq{\mathtt r}_{S_0+\la V,S_0}$.

M{\o}ller maps act on functionals by pullback. 
\be
(r_{\la V}F)(\ph)\doteq F\circ {\mathtt r}_{\la V}(\ph)\,,\\
\ee
where $F\in \Fcal_{\reg}$, $\ph\in\Ecal$. 

\begin{lemma}
\label{YF lemma}
If $\la$ is a formal parameter, then the retarded M{\o}ller map satisfies the \emph{Yang-Feldman} equation
\be
\label{YF equation}
{\mathtt r}_{\la V}(\ph)=\ph-\la\Delta_{S_0}^{\mathrm{R}}V^{(1)}({\mathtt r}_{\la V}(\ph))\,.
\ee
Conversely, a map satisfying the Yang-Feldman equation \eqref{YF equation} must be the retarded M{\o}ller map.
\end{lemma}
\begin{proof}
Using eq.~\eqref{compo} and \eqref{first order}
\begin{align*}
\frac{d}{d\la}{\mathtt r}_{\la V}(\ph) &= \left.\frac{d}{d\mu}\mathtt r_{(\la+\mu)V}(\ph)\right|_{\mu=0} 
= \left.\frac{d}{d\mu}\mathtt r_{S_0+(\la+\mu)V,S_0+\la V}\circ \mathtt r_{\la V}(\ph)\right|_{\mu=0} \\
&=-\Delta_{S_0+\la V}^{\rm R}V^{(1)}({\mathtt r}_{\la V}(\ph))
\end{align*}
Apply $P_{S_0}+\la V^{(2)}(\ph)$ (where $P_{S_0}$ is the differential operator induced by $S_0''$).
\[
\left(P_{S_0}+\la V^{(2)}({\mathtt r}_{\la V}(\ph))\right)\frac{d}{d\la}{\mathtt r}_{\la V}(\ph) =-V^{(1)}({\mathtt r}_{\la V}(\ph))
\]
Apply $\Delta_{S_0}^{\mathrm R}$ (which is independent of $\ph$).
\[
\left(\id+\la\Delta_{S_0}^{\mathrm{R}} V^{(2)}({\mathtt r}_{\la V}(\ph))\right)\frac{d}{d\la}{\mathtt r}_{\la V}(\ph) =-\Delta_{S_0}^{\mathrm{R}}V^{(1)}({\mathtt r}_{\la V}(\ph))
\]
Rearranging gives
\begin{align*}
\frac{d}{d\la}{\mathtt r}_{\la V}(\ph) 
&=-\Delta_{S_0}^{\mathrm{R}}V^{(1)}({\mathtt r}_{\la V}(\ph))-\la\Delta_{S_0}^{\mathrm{R}} V^{(2)}({\mathtt r}_{\la V}(\ph))\frac{d}{d\la}{\mathtt r}_{\la V}(\ph)\\
&=\frac{d}{d\la}\left(-\la \Delta_{S_0}^{\mathrm{R}}V^{(1)}({\mathtt r}_{\la V}(\ph))\right) \,.
\end{align*}
Since $\mathtt r_0=\id$, integrating gives \eqref{YF equation}.

The Yang-Feldman equation is equivalent to saying that the inverse of the M{\o}ller map is
\be
\label{inverse Moller}
{\mathtt r}^{-1}_{\la V}(\ph) = \ph + \la \Delta_{S_0}^{\mathrm{R}} V^{(1)}(\ph) \,.
\ee
This gives ${\mathtt r}^{-1}_{\la V}$ as a formal power series in $\la$. The constant term is just the identity map. Therefore this is invertible as a formal power series. Its inverse is unique, so the M{\o}ller map is that unique inverse. 

More explicitly, $\mathtt r_{\la V}(\ph) \approx \ph$ to $0$'th order. Applying eq.~\eqref{YF equation} iteratively improves this approximation, so that $\mathtt r_{\la V}(\ph) \approx \ph - \Delta^{\mathrm R}_{S_0} V^{(1)}(\ph)$ to first order and the sequence of approximations converges as a formal power series ($\la$-adically) to $\mathtt r_{\la V}$.
\end{proof}

Nonperturbatively, the Yang-Feldman equation may be taken as a definition of the M{\o}ller map. There is certainly no problem with defining ${\mathtt r}^{-1}_{\la V}$ by eq.~\eqref{inverse Moller}, but its inverse ${\mathtt r}_{\la V}$ may not actually exist. An idea how to use the Nash-Moser inverse function theorem on locally convex topological vector spaces to tackle this problem has been proposed in \cite{BFR}.

We will now provide an alternative proof of the result of \cite{DF02} stating that $r_{\la V}$ intertwines between the free and interacting Poisson brackets. The advantage of our proof is that it is explicitly performed to all orders and can be generalized to the nonperturbative setting, while the argument in \cite{DF02} is essentially a proof of the infinitesimal version of the statement.
\begin{df}
Denote the derivative of the inverse M{\o}ller map by
\[
\rho \doteq \left(\mathtt r^{-1}_{\la V}\right)^{(1)}(\ph) : \Ecal\to\Ecal \,.
\]
The transpose $\rho^{\mathsf T}:\Ecal^* \to \Ecal^*$ is defined by reversing the arguments in the integral kernel for $\rho$.
\end{df}

\begin{lemma}
\label{Propagator Moller}
The derivative of the inverse M{\o}ller map is 
\be
\label{Moller derivative}
\rho = \id + \la \Delta_{S_0}^{\mathrm{R}} V^{(2)}(\ph) \,,
\ee 
and
\be
\label{propagator Moller}
\rho\circ \Delta_S^{\mathrm{R}}(\ph)\circ \rho^{\mathsf T}\big|_{\Ecal_c} 
= \Delta^{\mathrm R}_{S_0} + \la \Delta^{\mathrm R}_{S_0}V^{(2)}(\ph)\Delta_{S_0}^{\mathrm{A}} \,.
\ee
\end{lemma}
\begin{proof}
Equation \eqref{Moller derivative} follows immediately from eq.~\eqref{inverse Moller}.

The image of $\Delta^{\mathrm R}_{S_0}:\Ecal^*_c\to\Ecal$ is the set of $\psi\in\Ecal$ such that $\supp\psi$ is past-compact and $\supp P_{S_0}\psi$ is compact. Because $P_{S}(\ph)-P_{S_0} = \la V^{(2)}(\ph)$ has compact support, the image of $\operatorname{Im} \Delta^{\mathrm R}_S(\ph)=\operatorname{Im} \Delta^{\mathrm R}_{S_0}$. Also note that $P_{S_0}\circ \Delta^{\mathrm R}_{S_0} = \id_{{\Ecal^*_c}^{\sst \CC}}$ implies $\Delta^{\mathrm R}_{S_0} \circ P_{S_0} = \id_{\operatorname{Im} \Delta^{\mathrm R}_{S_0}}$.

With this in mind --- and hiding the $\ph$ arguments --- we have 
\begin{align*}
\rho\circ\Delta_S^{\mathrm{R}} &= \left(\id_{\operatorname{Im} \Delta^{\mathrm R}_{S_0}} + \la \Delta_{S_0}^{\mathrm{R}} V^{(2)}\right) \Delta_S^{\mathrm{R}} \\
&=\left(\Delta^{\mathrm R}_{S_0} P_{S_0} + \la \Delta_{S_0}^{\mathrm{R}} V^{(2)}\right) \Delta_S^{\mathrm{R}} \\
&= \Delta^{\mathrm R}_{S_0} \left(P_{S_0} + \la V^{(2)}\right) \Delta_S^{\mathrm{R}}
= \Delta^{\mathrm R}_{S_0}P_S \Delta_S^{\mathrm{R}} = \Delta^{\mathrm R}_{S_0} \circ \id_{{\Ecal^*_c}^{\sst \CC}}\\
&= \Delta^{\mathrm R}_{S_0} \,.
\end{align*}
Composing this with the transpose $\rho^{\mathsf T} = \id_{\Ecal^*_c} + \la V^{(2)}(\ph)\Delta_{S_0}^{\mathrm{A}}$ gives
\begin{align*}
\rho\circ \Delta_S^{\mathrm{R}} \circ\rho^{\mathsf T} 
&= \Delta^{\mathrm R}_{S_0} \left(\id_{\Ecal^*_c}  + \la V^{(2)}\Delta_{S_0}^{\mathrm{A}} \right)\\
&= \Delta^{\mathrm R}_{S_0} + \la \Delta^{\mathrm R}_{S_0}V^{(2)}\Delta_{S_0}^{\mathrm{A}} \,.
\end{align*}
\end{proof}

\begin{prop}\label{IntertwiningProp}
\cite[Prop.~2]{DF02}
	Let $F,G,V\in \Fcal_{\loc}$ or  $F,G,V\in \Fcal_{\reg}$.
	The retarded M{\o}ller operator $r_{\la V}$ intertwines the Peierls	brackets, i.e.,
	\[
	\Pei{r_{\la V}F}{r_{\la V}G}_{S_0}=r_{\la V}(\Pei{F}{G}_{S_0+\la V}) \,.
	\]
\end{prop}
\begin{proof}
Note that the last term of eq.~\eqref{propagator Moller} is symmetric, so subtracting the transpose of this equation gives simply
\[
\rho\circ\Delta_S(\ph) \circ\rho^{\mathsf T} 
= \Delta_{S_0} \,.
\]

It is simpler to prove the equivalent property with $r_{\la V}^{-1}$. First note that the derivative of $r_{\la V}^{-1}F = F\circ{\mathtt r}^{-1}_{\la V}$ is
\[
(r_{\la V}^{-1}F)^{(1)}(\ph) = F^{(1)}({\mathtt r}^{-1}_{\la V}\ph) \circ\rho \,.
\]
This gives
\begin{align*}
\Pei{r_{\la V}^{-1}F}{r_{\la V}^{-1}G}_S (\ph) &= \left<\Delta_S(\ph), [F^{(1)}({\mathtt r}^{-1}_{\la V}\ph) \circ\rho] \otimes  [G^{(1)}({\mathtt r}^{-1}_{\la V}\ph) \circ\rho]\right> \\
&= \left<\rho\circ \Delta_S(\ph) \circ\rho^{\mathsf T} , F^{(1)}({\mathtt r}^{-1}_{\la V}\ph) \otimes G^{(1)}({\mathtt r}^{-1}_{\la V}\ph) \right>\\
&= \left <\Delta_{S_0}, F^{(1)}({\mathtt r}^{-1}_{\la V}\ph) \otimes G^{(1)}({\mathtt r}^{-1}_{\la V}\ph) \right>\\
&= \left(r_{\la V}^{-1} \Pei{F}{G}_{S_0}\right)(\ph) \,.
\end{align*}
\end{proof}

\section{Deformation quantization}\label{DeformationQuantization}
Suppose that $\Fcal_*$ is some space of functionals that is closed under the pointwise product and some Poisson bracket. 
%
Formal deformation quantization \cite{BFFLS1} of the  Poisson algebra $(\Fcal_*,\Pei{\cdot}{\cdot})$ means constructing an associative algebra $(\Fcal_*[[\hbar]],\star)$, where the  product $\star$ is given by a power series
\be
\label{generic star product}
F\star G=\sum\limits_{n=0}^\infty\hbar^n B_n(F,G)\,,
\ee
in which each $B_n$ is a bidifferential operator (in the sense of calculus on $\Ecal$) and in particular
\begin{align*}
B_0(F,G)&=F\cdot G\,,\\
B_1(F,G)-B_1(G,F)&=i\Pei{F}{G}\,.
\end{align*}


\subsection{Exponential star products}\label{sec:ExpProd}
The Peierls bracket of a quadratic action such as $S_0$ is a constant bidifferential operator, in the sense that it does not depend on $\ph$. Some of the $\star$-products that we need to consider are simple in a similar way.

\begin{df}
Given a sequence of cones $\Lambda_n\subset T^*M$ for $n=1.2,\dots$, we define a space of functionals by
\[
\Fcal_* \doteq \{F\in\Ci(\Ecal,\CC) \mid \WF(F^{(n)}(\ph))\subset \Lambda_n \forall\ph\in\Ecal,\ n\in\NN\} .
\]
The most important choices here are:
\begin{itemize}
\item
Regular functionals ($*=\reg$) are defined by $\Lambda_n=\varnothing$.
\item
\label{MicrocausalFunctionals}
\emph{Microcausal functionals} ($*=\mc$) are defined by
\be\label{cone}
	\Lambda_n\doteq T^*M^n\setminus\{(x_1,\dots,x_n;k_1,\dots,k_n)\mid (k_1,\dots,k_n)\in (\overline{V}_+^n \cup \overline{V}_-^n)_{(x_1,\dots,x_n)}\}\,.
\ee
\end{itemize}
%
	\end{df}

This has an obvious generalization to functionals of two variables.
\begin{df}\label{manyvar}
Let $\Fcal_*^2$ be the set of functionals $F\in\Ci(\Ecal\times\Ecal,\CC)$
such that for all $\ph_1,\ph_2\in \Ecal$, $n_1,n_2\in\NN_0$,   and for $n\doteq n_1+n_2$,
\[
\frac{\delta^{n_1+n_2}}{\delta\ph_1^{n_1}\delta\ph_2^{n_2}} F(\ph_1,\ph_2)\,,
\]
as a distribution in  $\Gamma'(E^{\boxtimes n}\to M^n)^{\sst \CC}$, has  WF set contained in $\Lambda_n$.
\end{df}

Let $m:\Fcal_*\otimes\Fcal_*\to \Fcal_*$ denote the pointwise product. In this notation, $F\cdot G = m(F\otimes G)$. Note that $\Fcal_*\otimes\Fcal_*\subset \Fcal_*^2$, and $m$ is pullback by the diagonal. It is clear from definition \ref{manyvar} that $m$ extends  to a map  $m:\Fcal_*^2\to\Fcal_*$. 

Suppose $K\in {\Gamma_c'}({E^*}^{\boxtimes 2}\rightarrow M^2)^{\sst\CC}$ is such that the differential operator defined by 
\be
\label{DofK}
D\doteq \left<K,\tfrac{\delta^2}{\delta \ph_1\delta\ph_2}\right>\,,
\ee
is a map $D:\Fcal_*^2\to \Fcal_*^2$. Note that on tensor products
\be
\label{biderivation}
[D(F\otimes G) ](\ph_1,\ph_2)\doteq \left<\bidis,F^{(1)}(\ph_1)\otimes G^{(1)}(\ph_2)\right>\,.
\ee
\begin{rem}
Note the nontrivial compatibility condition implicit here  between the choice  of $\Fcal_*$ (i.e. the choice of $\Lambda_n$) and the singularity structure of $\bidis$.
\end{rem}
\begin{df}\label{def:exponential prod}
The \emph{exponential product} given by $\bidis$ on $\Fcal_*[[\hbar]]$ is defined by
\be
\label{exponential form}
F\star G \doteq m\circ e^{\hbar D}(F\otimes G)  \,.
\ee
\end{df}
\begin{prop}
Any exponential product, $\star$,  is  associative.
\end{prop}
This is proven in \cite[Prop.~II.4]{Kai}.  This is a simple generalization of the finite rank case originally proven in \cite[Thm.~8]{Gerst3}. See also \cite{MMOY07,Waldmann}.

The involution ${}^*$ (given by complex conjugation) extends to $\Fcal_*[[\hbar]]$ if we just let $\hbar^*=\hbar$. However, this is an antiautomorphism of $\star$ if and only if $\bidis(y,x)=\overline{\bidis(x,y)}$.

In terms of eq.~\eqref{generic star product}, if we expand \eqref{exponential form} in powers of $\hbar$, then the first term is just $B_0=m$, and the second term is $B_1= m\circ D$, i.e.,
\[
B_1(F,G)(\ph) = \left<\bidis, F^{(1)}(\ph)\otimes G^{(1)}(\ph)\right>\,.
\]
From this, we can see that $\star$ gives a deformation quantization with respect to the Peierls bracket induced by $S_0$ if and only if 
\be
\label{antisymmetric part}
\bidis(x,y)-\bidis(y,x) = i\Delta_{S_0}(x,y) \,.
\ee
So, only the antisymmetric part of $\bidis$ is relevant to compatibility with the Peierls bracket. The freedom in choosing the symmetric part has been exploited in the literature \cite{BF0} to construct QFT models for free fields on globally hyperbolic manifolds.

\begin{df}
\label{Alpha}
Given a bidistribution, $Y \in {\Gamma_c'}({E^*}^{\boxtimes 2}\rightarrow M^2)^{\sst\CC}$, let $\Dcal_Y \doteq \left<Y,\frac{\delta^2}{\delta\ph^2}\right>$ and $\alpha_{Y} \doteq e^{\frac\hbar2\Dcal_Y}$.
\end{df}
\begin{prop}
\label{Star equivalence}
Consider  $\bidis_1,\bidis_2\in{\Gamma_c'}^{\sst\CC}({E^*}^{\boxtimes 2}\rightarrow M^2)$, whose difference, $Y\doteq\bidis_2-\bidis_1$ is symmetric, and which determine products $\star_1$ and $\star_2$.  If $\Fcal_*$ is in the domain of all powers of $\Dcal_Y$, then  
\[
\alpha_{Y} : (\Fcal_*[[\hbar]],\star_1)\to (\Fcal_*[[\hbar]],\star_2)
\]
is an isomorphism.
\end{prop}
\begin{proof}
Firstly, the hypothesis that $\Fcal_*$ is in the domain of all powers of $\Dcal_Y$ means that $\alpha_{Y} : \Fcal_*[[\hbar]]\to\Fcal_*[[\hbar]]$ is well defined.

Because of the symmetry of $Y$, applying $\Dcal_Y$ to a product gives 
\[
\Dcal_Y (FG) = \Dcal_Y(F)G + F\Dcal_Y(G) + 2m\circ (D_2-D_1)(F\otimes G) \,.
\]
More concisely,
\[
\Dcal_Y\circ m = m\circ\left(\Dcal_Y\otimes\id+\id\otimes\Dcal_Y+2D_2-2D_1\right)\,,
\]
where $D_i$ is given by $K_i$ in eq.~\eqref{DofK}.
In other words, $m$ intertwines those two operators. This implies that it intertwines their powers and their exponentials, therefore
\[
\alpha_{Y}\circ m = m\circ e^{\hbar(\frac12\Dcal_Y\otimes\id+\frac12\id\otimes\Dcal_Y+D_2-D_1)}
\]
since the various differential operators commute. Composing this identity on the right with $e^{\hbar D_1}$ gives
\begin{align*}
\alpha_{Y} \circ m \circ e^{\hbar D_1} &= \alpha_{Y}\circ m_1 \\
&= m\circ e^{\hbar(\frac12\Dcal_Y\otimes\id+\frac12\id\otimes\Dcal_Y+D_2)}
= m_2\circ (\alpha_{Y}\otimes\alpha_{Y}) \,.
\end{align*}
This means that 
\[
\alpha_{Y} : (\Fcal_*[[\hbar]],\star_1)\to (\Fcal_*[[\hbar]],\star_2)
\]
 is a homomorphism. Finally, $\alpha_{Y}$ is a formal power series with leading term the identity map, therefore it is invertible and hence an isomorphism.
\end{proof}
\begin{rem}
In the finite-dimensional setting, a proof of this result was given in \cite{MMOY07}, although it may have been known well before then. In the context of quantum field theory, it was discussed in \cite{BDF}.
\end{rem}

Since only the antisymmetric part of $\bidis$ really matters, the simplest choice is to take $\bidis$ to be antisymmetric. In that case, eq.~\eqref{antisymmetric part} requires that $\bidis=\frac{i}2\Delta_{S_0}$. 
\begin{df}
The \emph{Moyal-Weyl product} (denoted $\starz$) is the exponential product defined on $\Fcal_\reg[[\hbar]]$ by $\frac{i}2\Delta_{S_0}$.
\end{df}

Unfortunately, this $\starz$ does not extend to a larger space of functionals than $\Fcal_\reg$. This is the fault of $\Delta_{S_0}$, whose wavefront set is
\be
\WF(\Delta_{S_0})=\{(x,k;x',-k')\in \dot{T}^*M^2\mid(x,k)\sim(x',k')\}\,,
\ee
where $\sim$ means that both $(x,k)$ and $(x',k')$ belong to some null geodesic strip, i.e., a curve of the form $(\gamma,\kappa):I\to T^*M$, for some interval $I \subset \RR$,  where $\gamma$ is a null geodesic  and $\kappa$ is given by $\kappa(\lambda) = g(\dot{\gamma}(\lambda),.)$.

The problem is that the pullback of a tensor power of $\Delta_{S_0}$ has a  WF set that at singular points contains the whole cotangent bundle \cite{BDF,DFloop}. This means that the second order term of $\starz$ is only well defined on regular functionals. We can obtain a better behaved star product by a choice of $\bidis$ that has a smaller WF set. 

With this in mind, recall that (as shown in \cite{Rad}) there exists a real, symmetric, distributional bisolution to the field equation, $H$, such that 
\be\label{WFsplit}
\Delta_{S_0}^+\doteq \tfrac{i}{2}\Delta_{S_0}+H
\ee
has WF set
\be\label{spectrum}
\WF(\Delta_{S_0}^+)=\{(x,k;x,-k')\in \dot{T}^*M^2\mid(x,k)\sim(x',k'), k\in \partial J^+_x\}\,,
\ee
where $\partial J^+_x\subset T^*_xM$ is the set of future-pointing null vectors.
Such a $H$ is called a \emph{Hadamard distribution}.

This is better, because sums of future-pointing vectors are always future-pointing, and do not give the entire cotangent space. This leaves room for $F^{(n)}$ and $G^{(n)}$ to have nontrivial WF sets but still gives a well-defined star product of $F$ and $G$.

This is the motivation behind the definition of microcausal functionals in Definition~\ref{MicrocausalFunctionals}.


\begin{df}
Let $\Had$ denote the set of Hadamard distributions.

Given $H\in\Had$, the \emph{Wick} product (denoted $\starH$) is the exponential product on $\Fcal_\mc[[\hbar]]$ given by $\Delta_{S_0}^+$. 
\end{df}
Hence $(\Fcal_\mc[[\hbar]],\starH)$ is a deformation quantization of $\Fcal_\mc$ with respect to the Peierls bracket \cite{BDF} with ${}^*$ as an involution. 
By Prop.~\ref{Star equivalence}, 
\[
\alpha_H : (\Fcal_\reg[[\hbar]],\starz)\to(\Fcal_\reg[[\hbar]],\starH) \subset (\Fcal_\mc[[\hbar]],\starH) 
\]
is an isomorphism.

The choice of a Hadamard distribution  is far from unique. It is only determined up to the addition of any smooth real symmetric bisolution. Prop.~\ref{Star equivalence} shows that the products given by any two choices are equivalent, as 
\[
\alpha_{H'-H}:(\Fcal_\mc[[\hbar]],\starH)\to(\Fcal_\mc[[\hbar]],\star_{H'})
\]
is an isomorphism. These isomorphisms are coherent in the sense that $\alpha_{H''-H'}\circ\alpha_{H'-H}=\alpha_{H''-H}$.

This means that we can think of just one abstract algebra, $\fA$. 
Technically, $\fA$ is the limit of the above collection of algebras over $\Had$, viewed as a category with a unique morphism from any object to any other. Explicitly:
\begin{df}
\label{fA def}
\[
\fA \doteq \{A:\Had\to\Fcal_\mc[[\hbar]] \mid \forall H,H'\in\Had,\ A_{H'} = \alpha_{H'-H}A_H\}
\]
We denote the product of $A$ and $B\in\fA$ as $A\bullet B$, where
\[
(A\bullet B)_H \doteq A_H\starH B_H \,.
\]
The involution is defined by
\[
(A^*)_H \doteq (A_H)^* \,.
\]
\end{df}

Our point of view differs slightly from  the one taken in \cite{BDF}. In that work, $\fA$ denotes a particular realization of our abstract algebra, obtained by completion of the space $\Fcal_{\reg}[[\hbar]]$.

We will not need to use this definition directly very much. The point is that this algebra is isomorphic to all those constructed using Hadamard distributions, but it does not require a choice of Hadamard distribution.

The different star products come from different ways of identifying $\fA$ with $\Fcal_\mc[[\hbar]]$ as vector spaces. 
In the next section we will explain how these different identifications are related to the choice of quantization maps.

\subsection{Quantization maps}\label{sec:QuantMaps}
In this section we discuss quantization maps and formalize constructions known from \cite{BDF} and \cite{DF}. 

Here are some key features of $\fA$:
\begin{itemize}
	\item
	$\fA$ is a free module of $\CC[[\hbar]]$.
	\item
	There is a surjective homomorphism $\Pcal:\fA\to\Fcal_\mc$ (evaluation at the classical limit).
	\item
	$\ker\Pcal=\hbar\,\fA$.
	\item
	This implies that the commutator of any $A,B\in\fA$ satisfies  $\Pcal([A,B]_\bullet)=[\Pcal(F),\Pcal(G)]=0$, so $[A,B]_\bullet\in\hbar\fA$. This allows us to state the property that
	\be
	\label{Poisson property}
	\Pcal\left(\tfrac1{i\hbar}[A,B]_\bullet\right)=\Pei{\Pcal(F)}{\Pcal(F)}_{S_0}\,.
	\ee
\end{itemize}
The space of microcausal functionals $\Fcal_\mc$ is not necessarily the only choice of classical algebra, but it will suffice here.

Note that a $\CC$-linear map $\Fcal_*\to\fA$ extends to a $\CC[[\hbar]]$-linear map $\Fcal_*[[\hbar]]\to\fA$. Indeed, these are equivalent in the sense that any $\CC[[\hbar]]$-linear map is just the linear extension of its restriction to $\Fcal_*$. For this reason, we will abuse notation by denoting such maps by the same symbol.
\begin{df}
\label{Quantization map}
	A \emph{quantization map} is a linear map $\Qcal:\Fcal_*\to\fA$ such that 
	\begin{itemize}
		\item
		$\Pcal\circ\Qcal=\id:\Fcal_*\to\Fcal_*$, and
		\item
		the image of the $\CC[[\hbar]]$-linear extension of $\Qcal$ is closed under the noncommutative product in $\fA$.
	\end{itemize}
\end{df}
In Physics terms, a quantization map is a choice of operator ordering.
\begin{rem}
It would also be reasonable to require t${}^*$-linearity:
\be
\label{reality}
\Qcal(F^*)=\Qcal(F)^*\,.
\ee
Some --- but not all --- of our quantization maps satisfy this condition.
\end{rem}
\begin{prop}
	Let $\Fcal_*\subseteq\Fcal_\mc$  be a Poisson subalgebra. A quantization map $\Qcal:\Fcal_*\to\fA$,  induces a star product on $\Fcal_*[[\hbar]]$ by the condition
	\be
	\label{induced product}
	\Qcal(F \star G) = \Qcal(F)\bullet \Qcal(G) \,.
	\ee
\end{prop}
\begin{proof}
	First, we need to show that $\Qcal:\Fcal_*[[\hbar]]\to\fA$ is injective. Suppose that there is a nonzero element of the kernel with leading order term $\hbar^m F$. This means that $\Qcal(\hbar^mF)\in\hbar^{m+1}\fA$, but $\Qcal(\hbar^mF)=\hbar^m\Qcal(F)$, so $\Qcal(F)\in\hbar\fA$ and $\Pcal\circ\Qcal(F)=0$. However, $\Pcal\circ\Qcal(F)=F\neq0$, so this is a contradiction.
	
	So, $\Qcal$ is a bijection to its image, and its image is an associative algebra. In this way, eq.~\eqref{induced product} defines an associative product, $\star$, on $\Fcal_*[[\hbar]]$. For $F,G\in\Fcal_*$, $F\star G\in\Fcal_*[[\hbar]]$ has an expansion of the form \eqref{generic star product} simply because this is a formal power series. 
	
	By the compatibility of $\Qcal$ with $\Pcal$,
	\[
	B_0(F,G) = \Pcal\circ\Qcal(F\star G) = \Pcal\left(\Qcal(F)\bullet\Qcal(G)\right) = F\cdot G\,,
	\]
	and by eq.~\eqref{Poisson property},
	\begin{align*}
	B_1(F,G)-B_1(G,F) &= \Pcal\circ\Qcal\left(\tfrac1\hbar[F\star G- G\star F]\right) \\ 
	&= \Pcal\left(\tfrac1\hbar[\Qcal(F),\Qcal(G)]_\bullet\right)\\ 
	&= i\Pei{F}{G}_{S_0}\,.
	\end{align*}
\end{proof}

\begin{rem}
	In a less formal version, $\Qcal$ would be a map from an algebra of bounded classical observables to the algebra of sections of a continuous field of C${}^*$-algebras. This can still induce a star product, but eq.~\eqref{induced product} needs to become an asymptotic expansion.
\end{rem}

Let us now return to the exponential product discussed in
section~\ref{sec:ExpProd}. 
\begin{df}
Given a Hadamard distribution $H\in\Had$, 
\[
\Qcal_H:\Fcal_{\mc}\rightarrow \fA
\]
is defined  (in terms of Def.~\ref{fA def}) by
\[
(\Qcal_HF)_{H'} \doteq \alpha_{H'-H}(F) \,.
\]
\end{df}
This is a quantization map that induces the exponential product $\starH$. It is also ${}^*$-linear.
\begin{df}
\[
\Qcal_\Weyl:\Fcal_\reg\to \fA 
\]
is defined by 
\[
(\Qcal_\Weyl F)_H \doteq \alpha_H F \,.
\]
\end{df}
This is a quantization map which induces the Moyal-Weyl product, $\starz$ on $\Fcal_\reg[[\hbar]]$. It is also ${}^*$-linear.

 These quantization maps are related by
 \[
 \Qcal_\Weyl = \Qcal_H\circ\alpha_H
 \]
 and
 \[
 \Qcal_{H} = \Qcal_{H'}\circ\alpha_{H'-H} \,.
 \] 

Let $\fA_\reg$ be the image of the extension of $\Qcal_\Weyl$ to formal power series.
Because the extensions of these quantization maps to formal power series are injective, $\Qcal_\Weyl$ induces a commutative ``pointwise'' product on $\fA_\reg$, and $\Qcal_H$ induces one on $\fA$. 
If $\Delta^+_{S_0}$ is the 2-point function of some choice of ``vacuum'' state, then the commutative product induced by $\Qcal_H$ is the normal-ordered (Wick) product, and  $\Qcal_H$ can be interpreted as the corresponding normal-ordering map.  Of course, this depends upon a choice of $H$.

\subsection{Nets of Algebras}
Although it is not our main focus here, in AQFT, we really want to construct a net of algebras, rather than just a single algebra. 
\begin{df}
Let $\Kcal(\Mcal)$ be the set of open, precompact, causally convex subsets of $M$. Define this to be a category whose morphisms are the inclusions of subsets.

Also, for $\Ocal\in\Kcal(\Mcal)$, we abuse notation by letting
\[
\Fcal_*(\Ocal) \doteq \{F\in\Fcal_* \mid \supp F \subseteq \Ocal\} \,.
\]
\end{df}
This is a functor from $\Kcal(\Mcal)$ to the category of subspaces of $\Fcal_*$.

In this way, the concrete construction of deformation quantization gives a net of algebras, i.e., a functor from $\Kcal(\Mcal)$ to a category of algebras:
\[
\Ocal\mapsto (\Fcal_\mc(\Ocal)[[\hbar]],\starH) \,.
\]

More abstractly, we have a functor (also denoted $\fA$) from $\Kcal(\Mcal)$ to the category of subalgebras of $\fA$.

For a given quantization map, $\Qcal:\Fcal_*\to\fA$, we can define $\Qcal_\Ocal$ as the restriction of $\Qcal$ to $\Fcal_*(\Ocal)$. The quantization maps that we have been considering satisfy the additional property that 
\be
\label{respect support}
\Qcal_\Ocal:\Fcal_*(\Ocal)\to\fA(\Ocal) \,.
\ee
It follows immediately that these are the components of a natural transformation from the functor $\Fcal_*$ to $\fA$ composed with the forgetful functor to vector spaces.

\subsection{Time-ordered products}\label{ToP}
In this section we discuss the time-ordered product. One uses this structure in pAQFT to construct the $S$-matrix and the interacting fields (see e.g., \cite{BDF}). We will come back to these in Section \ref{cpt}. Here, we want to review some basic properties of this product, emphasizing the importance of the time-ordering map  $\Qcal_\TT$ 
that establishes the equivalence between the time-ordered product and the pointwise product. 

\begin{df}
The relation $\prec$ on $\Kcal(\Mcal)$ means ``not later than'' i.e., $\Ocal_1 \prec \Ocal_2$ means that there exists a Cauchy surface to the future of $\Ocal_1$ and to the past of $\Ocal_2$.
\end{df}
We want a time-ordered version of the non-commutative product $\bullet$ to satisfy for $\Ocal_1,\Ocal_2\in\Kcal(\Mcal)$, $A\in\fA(\Ocal_1)$ and $B\in\fA(\Ocal_2)$,
\be\label{orderingQT}
A\circT B=\begin{cases}
	A\bullet B&\textrm{if }\Ocal_2\prec\Ocal_1\,,\\
	B\bullet A&\textrm{if }\Ocal_1\prec\Ocal_2\,.
\end{cases}
\ee 

Moreover, we require that there exists a quantization map $\Qcal_{\TT}:\Fcal_*\to\fA$, satisfying the support condition \eqref{respect support} and such that
\[
\Qcal_{\TT}F \circT \Qcal_{\TT}G=\Qcal_{\TT}(F\cdot G)\,,
\]
for $F,G\in\Fcal_*$.

There is certain freedom in how we choose $\circT$ for a given $\bullet$. We will further restrict this freedom when discussing concrete realizations of $\bullet$ in the following section

\subsubsection{Relation to Weyl quantization}
In this section we discuss how to construct
 $\Qcal_{\TT}$ (and hence $\circT$) starting from another quantization map. 
 
It is convenient to start with $\Qcal_\Weyl$ and first construct $\Qcal_{\TT}$ on $\Fcal_{\reg}$ by setting
\be\label{eq:QT:QW:QH}
\Qcal_\TT\doteq \Qcal_\Weyl\circ\TT\,,
\ee
for some $\TT:\Fcal_{\reg}[[\hbar]]\rightarrow \Fcal_{\reg}[[\hbar]]$.

Definition \ref{Quantization map} of ``quantization map'' implies that $\TT=\id \mod \hbar$. The support condition \eqref{respect support} implies that $\supp \TT F$ is contained in the causal completion of $\supp F$.

We impose further conditions on $\TT$:
\begin{itemize}
\item
$\TT:\Fcal_{\reg}[[\hbar]]\rightarrow \Fcal_{\reg}[[\hbar]]$ is a differential operator, 
\item
$\TT F=F$ for $F$ linear,
\end{itemize}
This last property is equivalent to requiring that there exists a natural endomorpihsm of the functor $\Fcal_\reg[[\hbar]]$ whose component at $\Ocal$ is the restriction of $\Tcal$ to $\Fcal_\reg(\Ocal)[[\hbar]]$.

A natural question to ask is what is the freedom in choosing $\Tcal$? A similar problem arises in Epstein-Glaser renormalisation \cite{EG}, where one constructs $n$-fold time-ordered products recursively (as multilinear maps on local functionals). There, the non-uniqueness of these maps is characterized by the Main Theorem of Renormalization (for various versions of this result see \cite{SP82,Pin01,HW02,DF04,BDF}). The following theorem provides a solution to this problem for $\Tcal$ restricted to regular functionals (and without requiring the field independence). 
\begin{thm}
\label{T freedom}
Any two operators $\TT$ and $\tilde{\TT}$ satisfying these conditions are related by 
\[
\tilde\TT = \TT\circ e^X
\]
where 
\[
XF(\ph)=\sum_{n=2}^\infty \left<a_n(\ph),F^{(n)}(\ph)\right>\,,
\]
with each $a_n(\ph)$ a formal power series in $\hbar$ with coefficients in symmetric distributions on $M^n$ such that:
\begin{itemize}
\item
$a_n(\ph)$ is a multiple of $\hbar$;
\item
$a_n(\ph)$ is supported on the thin diagonal of $M^n$;
\item
$a_n(\ph)$ depends at most linearly on $\ph$;
\item
$a_n^{(1)}(\ph)$ is supported on the thin diagonal of $M^{n+1}$.
\end{itemize}
\end{thm}
\begin{proof}
First, define $X=\log ({\TT}^{-1} \circ \tilde\TT)$. Because $\TT$ and $\tilde\TT$ are differential, $X$ must be differential. This means that it can be written in the form
\[
XF(\ph)=\sum_{n=0}^\infty \left<a_n(\ph),F^{(n)}(\ph)\right>\,,
\]
where $a_n$ is valued in formal power series of symmetric distributions.
If $F$ is any linear functional, then $\TT F = \tilde\TT F = F$, so $XF=0$, therefore $a_0=a_1=0$.

The ratio $\TT^{-1}\circ\tilde\TT = e^X$ is almost a homomorphism of the pointwise product in the sense that if $F$ and $G$ have causally separated support ($\supp G \prec \supp F$ or $\supp F\prec\supp G$) then
\[
e^X(F\cdot G) = (e^X F)\cdot(e^X G) \,.
\]
Consequently, $X$ is almost a derivation in the sense that if $F$ and $G$ have causally separated support, then
\be
\label{almost derivation}
X(F\cdot G) = XF\cdot G + F\cdot XG \,.
\ee

Any finite set of distinct points of $M$ can be listed in an order consistent with the causal partial order, and there is a time slicing of $\Mcal$ (into Cauchy surfaces) consistent with this order. Let $x_1,\dots,x_n\in M$ be such a list of points. There exist neighborhoods of these points that are causally separated by Cauchy surfaces of this time slicing. Consider any $p_1,\dots,p_n\in\Ecal_c^*$ supported on these respective neighborhoods, and any $\ph_0\in\Ecal$. Define linear functionals
\[
F_i(\ph) = \left<\ph-\ph_0,p_i\right> \,.
\]
For any numbers $m_1,\dots,m_n\in\NN$, apply $X$ to the product $F_1^{m_1}\dots F_n^{m_n}$ and use eq.~\eqref{almost derivation}. Evaluating at $\ph_0$, this shows that 
\[
0 = \left<a_{m_1+\dots+m_n}(\ph_0),p_1^{\otimes m_1}\otimes \dots\otimes p_n^{\otimes m_n}\right> \,.
\]
Because the section $p_i$ can take any values around $x_i$, this shows that $a_{m_1+\dots+m_n}(\ph_0)$ is not supported at
\[
(\underbrace{x_1,\dots,x_1}_{m_1},\dots,\underbrace{x_n,\dots,x_n}_{m_n}) \,.
\]
Up to symmetry, this is any point outside the thin diagonal.

Finally, consider a  point $x\in M$. If $F$ is supported on a causally complete neighborhood $U\ni x$, then 
\[
\left<a_n(\ph),F^{(n)}(\ph)\right>
\]
must be supported on $U$. Take $F$ homogeneous of degree $n$, so that $F^{(n)}(\ph)=f\in\Gamma_c(M^n,(E^*)^{\boxtimes n})^{\sst \CC}$ and $\supp f\subset U^n$. Now consider $\psi\in\Ecal$ such that $\supp \psi\cap U=\varnothing$. Since $\Tcal F$ has to be supported in $U$, we conclude that
\[
\left<a_n(\ph)-a_n(\ph+\psi),f\right>=0\,.
\] 
We can now use this fact to conclude that the derivative $a_n^{(1)}(\ph)$, seen as a distribution on $M^{n+1}$, has to be supported on the diagonal. Since $X$ maps regular functionals to regular functionals, this implies that $a_n(\ph)$ can depend at most linearly on $\ph$.
\end{proof}

In pAQFT one also requires the stronger condition of \textit{field independence}, which is usually phrased as
\[
\frac{\delta }{\delta \ph} (\Tcal F)=\Tcal\, \frac{\delta F }{\delta \ph}\,.
\]
This just means that $\Tcal$ should have constant coefficients. If $\TT$ and $\tilde\TT$ in Theorem~\ref{T freedom} satisfy this, then each $a_n(\ph)$ is independent of $\ph$.

Just as a quantization map induces a star product from $\bullet$, it also induces a commutative product from $\circT$. For $\Qcal_\Weyl$ we denote this commutative product on $\Fcal_\reg[[\hbar]]$ as $\T$. This is defined analogously to eq.~\eqref{induced product} by 
\be
\Qcal_\Weyl(F\T G)= \Qcal_\Weyl(F)\circT\Qcal_\Weyl(G) \,.
\ee
This is equivalent to\[
F\T G\doteq \TT(\TT^{-1}F\cdot\TT^{-1}G)\,.
\]
Because $\Qcal_\Weyl$ satisfies \eqref{respect support},  \eqref{orderingQT} implies  that $\T$ satisfies the condition 
\be\label{ordering}
F\T G=\begin{cases}
	F\starz G&\textrm{if }\supp G\prec\supp F\,,\\
	G\starz F&\textrm{if }\supp F\prec\supp G\,.
\end{cases}
\ee

There is a natural choice of $\TT$ satisfying field independence, given in terms of the Dirac propagator, $\Delta_{S_0}^{\mathrm{D}}\doteq\frac{1}{2} (\Delta_{S_0}^{\mathrm{R}}+\Delta_{S_0}^{\mathrm{A}})$. Set
\be
\label{DiracT}
\TT\doteq \alpha_{i\Delta_{S_0}^{\mathrm D}}=e^{\frac{i\hbar}{2}\Dcal_{\mathrm{D}}}\,, 
\ee
where $\Dcal_{\mathrm{D}}=\left<\Delta_{S_0}^{\mathrm D},\frac{\delta^2}{\delta\ph^2}\right>$.

\begin{prop}
 The time-ordered product $\T$ defined in terms of $\TT$ by eq.~\eqref{DiracT} equals the exponential product given by Definition \ref{def:exponential prod} with $K=i\Delta_{S_0}^{\mathrm{D}}$.
\end{prop}
\begin{proof}
The Leibniz rule for differentiation implies that
\[
\Dcal_{K}\circ m=m\circ (\Dcal_K\otimes \id+\id\otimes \Dcal_K+2D_K)\,,
\] 
which is in fact a coproduct structure. Hence
\begin{multline*}
\TT(\TT^{-1}F\cdot\TT^{-1}G)=e^{\frac{i\hbar}{2}\Dcal_{\mathrm{D}}}\circ m\left(^{-\frac{i\hbar}{2}\Dcal_{\mathrm{D}}} F\otimes ^{-\frac{i\hbar}{2}\Dcal_{\mathrm{D}}}G\right)\\
=m\circ e^{\frac{i\hbar}{2}\Dcal_{\mathrm{D}}\otimes \id+\id\otimes \frac{i\hbar}{2}\Dcal_{\mathrm{D}}+i\hbar D_{\mathrm{D}}}\left(^{-\frac{i\hbar}{2}\Dcal_{\mathrm{D}}} F\otimes ^{-\frac{i\hbar}{2}\Dcal_{\mathrm{D}}}G\right)\\
=m\circ e^{i\hbar D_{\mathrm{D}}}(F\otimes G)\,.
\end{multline*}
	\end{proof}



Because both $\starz$ and $\T$ are exponential products, they are related by an exponential factor.  With obvious notation following eq.~\eqref{DofK}, $D_\Delta = D_{\mathrm R}-D_{\mathrm A}$ and $D_{\mathrm D} = \frac12\left(D_{\mathrm R}+D_{\mathrm A}\right)$, so $\frac12 D_\Delta-D_{\mathrm D} = -D_{\mathrm A}$, and therefore the relation is
\be
\label{starT relation}
m_{\starz} = \mT\circ e^{-i\hbar D_{\mathrm A}} \,.
\ee

From this relation, it is easy to see why $\T$ satisfies eq.~\eqref{ordering}. Consider the case that $\supp G \prec \supp F$. By definition, $\Delta_{S_0}^{\mathrm A}(x,y) = 0$ when $x\nleq y$, so $D_{\mathrm A}(F\otimes G)=0$ and
\begin{align*}
F\starz G &= \mT\circ e^{-i\hbar D_{\mathrm A}}(F\otimes G) = \mT(F\otimes G) = F\T G\,.
\end{align*}

Coming back to formula \eqref{eq:QT:QW:QH}, using $\Tcal$, we have now constructed a quantization map $\Qcal_{\TT}$ on $\Fcal_{\reg}[[\hbar]]$ with the desired properties. We denote by $\starb$ the star product on $\Fcal_{\reg}[[\hbar]]$ induced by $\Qcal_\TT$ from the nocommutative product $\bullet$ on $\fA$; this product is given by the multiplication map
\[
m\circ e^{-i\hbar D_A} \,,
\]
i.e.
\[
F\starb G= m\circ e^{-i\hbar D_A}(F\otimes G)\,.
\]
The disadvantage of using $\Qcal_\TT$ as the quantization map is that it does not extend naturally to $\Fcal_\mc$ (unlike $\Qcal_H$). It is also not ${}^*$-linear like $\Qcal_H$ and $\Qcal_\Weyl$; instead it induces an involution on $F\in\Fcal_\reg[[\hbar]]$, given by the formula 
\[
F \mapsto\TT^{-1}[(\TT F)^*]=\alpha_{-2i\Delta_{S_0}^{\mathrm D}}(F^*)
\]
and compatible with $\starb$.

\subsubsection{Relation to normal ordering}
We can also start with $\Qcal_H$ and let 
\[
\Qcal_\TT \doteq \Qcal_H\circ\TT_H \,.
\]
The relation between $\Qcal_\Weyl$ and $\Qcal_H$ implies that
\[
\TT_H = \alpha_{\Delta_{S_0}^{\mathrm F}} \,,
\]
where $\Delta_{S_0}^{\mathrm F}$ is the Feynman propagator, defined as $\Delta_{S_0}^{\mathrm F} = i\Delta_{S_0}^{\mathrm D} + H$. 

Using $\TT_H$ we can introduce another exponential product on $\Fcal_\reg[[\hbar]]$, denoted by $\dTH$ and given by setting $K=\Delta_{S_0}^{\mathrm F}$ in definition~\ref{def:exponential prod}. Clearly, this product is again commutative and equivalent to the pointwise product. 

Since $\Delta_{S_0}^+ - \Delta_{S_0}^{\mathrm F} = -i\Delta_{S_0}^{\mathrm A}$, the multiplication maps for $\starH$ and $\dTH$ are related by,
\[
m_{\starH}=m_{\TT_H}\circ e^{-i\hbar D_{\mathrm A}} \,,
\]
which is exactly the same as the relation between $\starz$ and $\T$; consequently, by identical reasoning, they are also related by eq.~\eqref{ordering}, and we can say that $\dTH$ is the time-ordered product associated to $\starH$.

Unfortunately, although $\starH$ is defined on $\Fcal_\mc[[\hbar]]$, this construction only defines $\dTH$ and $\TT_H$ on $\Fcal_\reg[[\hbar]]$. Renormalization in pAQFT is a matter of extending these consistently to local functionals. Again, there is some freedom in constructing such an extension \cite{BDF,HW02}, physically interpreted as \textit{renormalization freedom}. To fix this freedom, one imposes certain conditions more restrictive than those we imposed on time ordered products of regular functionals, this is discussed extensively in \cite{Due19}.

We summarize the different exponential product used throughout this section in Table~\ref{tab:prod}. Compared to the existing literature, the only new ingredient is the non-commutative product $\starb$, related to more commonly used star products by:
\begin{align*}
F\starz G &= \TT(\TT^{-1}F\starb\TT^{-1}G)\,,\\ F\starH G &= \TT_H(\TT_H^{-1}F\starb\TT_H^{-1}G)\,.
\end{align*}
\begin{table}
\begin{center}
\begin{tabular}{ccc}
Product & Symbol & Distribution\\
\hline
Pointwise & $\cdot$ & $0$\\
Weyl & $\starz$ & $\frac i2\Delta_{S_0} = \frac i2\left(\Delta_{S_0}^{\mathrm R}-\Delta_{S_0}^{\mathrm A}\right)$\\
Time ordered & $\T$ & $i\Delta_{S_0}^{\mathrm D} = \frac{i}2\left(\Delta_{S_0}^{\mathrm R}+\Delta_{S_0}^{\mathrm A}\right)$ \\
Wick & $\starH$ & $\Delta_{S_0}^+ = \frac i2\left(\Delta_{S_0}^{\mathrm R}-\Delta_{S_0}^{\mathrm A}\right)+H$\\
Time ordered & $\dTH$ & $\Delta_{S_0}^{\mathrm F} = \frac{i}2\left(\Delta_{S_0}^{\mathrm R}+\Delta_{S_0}^{\mathrm A}\right) + H$\\
& $\starb$ & $-i \Delta_{S_0}^{\mathrm A}$
\end{tabular}
\end{center}
\caption{Some of the exponential products, and the distributions used to construct them.}\label{tab:prod}
\end{table}

\section{Interaction and M{\o}ller operators}\label{cpt}
\subsection{The M{\o}ller operator in the abstract algebra}
Consider a theory with action $S=S_0+\la V$, where $V\in \Fcal_{\reg}$ is compactly supported and $\la$ is the coupling constant, treated from now on as a formal parameter (similarly to $\hbar$). 

We wish to construct a formal deformation quantization of $(\Fcal_*,\Pei{\cdot}{\cdot}_S)$. Proposition~\ref{IntertwiningProp} shows that $r_{\la V}$ intertwines the Peierls brackets for $S$ and $S_0$. If $\star$ is a formal deformation quantization for $S_0$, then this gives an obvious formal deformation quantization for $S$. Simply define $\star_r$ by 
\be
\label{naive}
r_{\la V}(F\star_r G) = r_{\la V}F \star r_{\la V}G \,.
\ee
This construction was proposed by Dito \cite{Dit90,Dit93} in the case of an interacting theory on Minkowski spacetime that is equivalent to a free theory. This is simple, but it is not the best choice. We will explain the reasons for this in Section~\ref{Naive}, after describing what we claim is a better construction.  

Firstly, we need to introduce some further definitions. A central object in pAQFT for constructing interacting theories is the \textit{formal $S$-matrix}.
\begin{df}\label{nrS}
The \emph{formal $S$-matrix} is the map $\Scal:\la\fA_{\reg}[[\la]]\rightarrow \fA_{\reg}[[\la/\hbar]]$ given by the time-ordered exponential
\be\label{Smatrix}
\Scal(A)=e_{\circT}^{iA/\hbar}=\Qcal_{\TT}\big( e^{i (\Qcal_{\TT}^{-1}A)/\hbar}\big) \,.
\ee
\end{df}
\begin{rem}
	Here $V$ plays the role of an interaction term where some infrared (IR) and ultraviolet (UV) regularizations have been implemented. The former guarantees compact support and the latter regularity. An example of a regular functional is $\ph\mapsto \int f(x_1,\dots,x_3)\ph(x_1)\dots\ph(x_3)$, where $f$ is a compactly supported density on $M^3$ and $\ph\in\Ci(M,\RR)$. The physical interaction is recovered in the limit 
	\[f\rightarrow \delta(x_1-x_2)\dots\delta(x_1-x_3)\,d^4x_1\dots d^4x_3\,.
	\]
	 Instead of taking this limit directly, in pAQFT one proceeds in two steps. First the map $\Scal$ is extended to a larger subset of $\fA$, to deal with the potential UV divergences (see for example \cite{BDF} and \cite{Book} for a review). Next, one takes \textit{the algebraic adiabatic limit} to deal with the IR problem. For more details see \cite{Fredenhagen2015,Book} and \cite{Eli16} for an alternative formulation.
\end{rem}
The physical interpretation is that $\Scal$ becomes the scattering matrix of the theory, in the adiabatic limit (i.e., when the infrared regularization is removed), if it exists. One uses $\Scal$ to construct interacting fields using the quantum M{\o}ller operator given by the formula of Bogoliubov (see e.g., \cite{BDF} and \cite{Book} for a review):
\begin{align*}
\Rcal_{A}(B)&= \left.-i\hbar\frac{d}{dt}\left(\Scal(A)^{\minus}\bullet\Scal(A+t B)\right)\right|_{t=0} \\
&= \left(\Scal(A)\right)^{-1}\bullet\left(\Scal(A)\circT B\right) \\
&= \left(e_{\circT}^{iA/\hbar}\right)^{-1}\bullet\left(e_{\circT}^{iA/\hbar}\circT B\right) \,.
\end{align*}
A priori, this is defined for $A,B\in\lambda\fA_\reg[[\lambda]]$. However, the last expression makes sense for $B\in\fA_\reg$ and defines a map $\Rcal_{A} : \fA_\reg\to\fA_\reg[[\la/\hbar]]$. Better still, we will show below that $\Rcal_{A} : \fA_\reg\to\fA_\reg[[\la]]$ --- i.e., no negative powers of $\hbar$ survive.

For the interaction $V\in\Fcal_\reg$, we choose $A=\la\,\Qcal_{\TT}(V)$ in this formula. The interacting product $\bint$ on $\fA_\reg$ is then defined by
\[
B \bint C \doteq \Rcal_{A}^{-1}\left(\Rcal_{A}(B)\bullet\Rcal_{A}(C)\right) \,.
\]

\subsection{The M{\o}ller operator in terms of functionals}
For the purpose of practical computations, it is easier to work with functionals rather than the abstract algebra $\fA_\reg$, so we should identify $\fA_\reg$ with $\Fcal_\reg[[\hbar]]$ by using a quantization map, but there are several to choose from. This leads to abundance of products used in the literature in this context. However, the fundamental structures can be described on the level of the abstract algebra, and there one needs only two products: $\bullet,\circT$, while $\bint$ can be constructed out of the two, as described above.

The most obvious way to work is based on the quantization map $\Qcal_{\textrm{Weyl}}$. With this identification, the formal S-matrix becomes
$\Scal_0\doteq \Qcal_{\textrm{Weyl}}^{-1}\circ\Scal\circ \Qcal_{\textrm{Weyl}}$, so one is often interested in
\[
 \Scal_0 (\la\,\TT V)= \TT(e^{i\la V/\hbar})\equiv e_{\TT}^{i\la \TT V/\hbar}\,.
\]
The M{\o}ller operator becomes $R_{0,\la V} \doteq \Qcal_{\textrm{Weyl}}^{-1}\circ\Rcal_{\la \Qcal_\TT(V)}\circ \Qcal_{\textrm{Weyl}}$, which is
\begin{align}\label{RV}
R_{0,\la V}(F)
&=\left(e_{\sst{\TT}}^{i\la\,\TT V/\hbar}\right)^{\starz\minus}\starz\left(e_{\sst{\TT}}^{i\la\,\TT V/\hbar}\T  F\right)\,.
\end{align}
If we instead use $\Qcal_H$, then the formula for $R_{H,\la V}$ is the same, but with $\starH$ and $\TT_H$.

All these formulae for the interacting star product and interacting fields are difficult to work with, because they each use two different products, neither of which is the natural pointwise product on the space of functionals.

\subsection{The M{\o}ller operator in the time-ordered identification}
Instead, the most convenient quantization map for computations is actually $\Qcal_\TT$. With this identification, the formal S-matrix becomes simply 
\[
\Qcal_\TT^{-1}\circ \Scal\circ\Qcal_\TT(F) = e^{iF/\hbar} \,.
\]
The M{\o}ller operator becomes $R_{\TT,\la V} \doteq \Qcal_\TT^{-1}\circ\Rcal_{\la \Qcal_\TT(V)}\circ \Qcal_\TT$, which is
\[
R_{\TT,\la V}(F) = \left(e^{i\la V/\hbar}\right)^{\starb-1}\starb \left(e^{i\la V/\hbar}\cdot F\right) \,.
\]
The inverse of this is particularly simple:
 \[
R_{\TT,\la V}^{-1}(G)= e^{-i\la V/\hbar}  \left(e^{i\la V/\hbar}\starb G\right) \,.
\]

Note that any formula for $R_{\TT,\la V}$ can be converted to a formula for $R_{H,\la V}$ (and \emph{vice versa}) by $R_{H,\la V}\circ\TT_H = \TT_H\circ R_{\TT,\la V}$.

Recall that the $n$'th derivative of a functional is valued in the complex dual of the projective tensor product $\Ecal^{\hat{\otimes}_\pi n}\cong \Gamma(E^{\boxtimes n}\rightarrow M^n)$.
\begin{df}
Any $\chi\in \Ecal^{\hat{\otimes}_\pi n}_\CC$ defines a linear map
\[
\chi\triangleright {}: \Ci(\Ecal,\CC)\to\Ci(\Ecal,\CC)
\]
by
\[
(\chi\triangleright F)(\ph) = \left<F^{(n)}(\ph),\chi\right> .
\]
\end{df}
Note that for $\psi\in\Ecal$, this defines a derivation. In general, this extends to linear combinations and defines differential operators. The idea is that the abelian Lie group $\Ecal$ acts on functionals by translation. This $\Ecal$ is its own Lie algebra and in that capacity acts on functionals by derivations. The projective symmetric tensor product algebra  $\hat S_\pi\Ecal$ is a completion of the universal enveloping algebra and acts on functionals by differential operators. The $\triangleright$ notation is borrowed from the Hopf algebra literature.

A convenient way of describing elements of this algebra --- and thus differential operators on functionals is as functions on $\Ecal'$. We can define $K\in\hat S_\pi\Ecal$ by giving a formula for $K(w)$ where $w\in\Ecal'$. Formally, $K\triangleright F$ is defined by replacing $w^{\otimes n} $ with $F^{(n)}$. Of course, we extend all of this to formal power series, and we use functionals valued in $\hat S_\pi\Ecal$.

\begin{prop}
\label{RfromJ}
\be
\label{RTformula}
R_{\TT,\la V}^{-1}(G)(\ph) = \left(J(\ph)\triangleright G\right)(\ph)
\ee
where
\be
\label{generating function}
J(\ph;w) \doteq \exp\left(\la\frac{V(\ph-i\hbar \Delta_{S_0}^{\mathrm A}w)-V(\ph)}{-i\hbar}\right)
\ee
for $\phi\in\Ecal$ and $w\in\Ecal^*_c$.
\end{prop}
\begin{rem}
Equation \eqref{RTformula} is subtle, so it is worth expanding upon. First, $J$ is a functional from $\Ecal$ to formal power series with coefficients in $\hat S_\pi\Ecal$. In eq.~\eqref{RTformula}, $J$ is evaluated at $\ph$ to give a formal power series with coefficients in $\hat S_\pi\Ecal$. This then defines a differential operator which acts on $G$ to give a functional. Finally, this functional is evaluated at $\ph$.
\end{rem}
\begin{proof}
First, for some arbitrary functional $F$, consider $F\starb{}$ as an operator acting on $G$:
\begin{align*}
(F\starb G)(\ph) &= \sum_{n=0}^\infty \frac{(-i\hbar)^n}{n!} \left<F^{(n)}(\ph),(\Delta^{\mathrm A}_{S_0})^{\otimes n} G^{(n)}(\ph)\right> \\
&= \left(K(\ph)\triangleright G\right)(\ph)
\end{align*}
where
\[
K(\ph;w) \doteq \sum_{n=0}^\infty \frac{(-i\hbar)^n}{n!} \left<F^{(n)}(\ph),(\Delta^{\mathrm A}_{S_0}w)^n\right> = F(\ph-i\hbar\Delta^{\mathrm A}_{S_0}w) \,.
\]

For $F=e^{i\la V/\hbar}$, this becomes
\[
K(\ph;w) = \exp\left(\la\frac{V(\ph-i\hbar \Delta_{S_0}^{\mathrm A}w)}{-i\hbar}\right) 
 = e^{i\la V(\ph)/\hbar} J(\ph;w) \,,
\]
and so
\[
R_{\TT,\la V}^{-1}(G)(\ph) = e^{-i\la V(\ph)/\hbar}\left(K(\ph)\triangleright G\right)(\ph) = \left(J(\ph)\triangleright G\right)(\ph) \,.   \qedhere
\]
\end{proof}

The original proof of the next result is due to \cite{DFloop} and was presented for the  case of local functionals. Here we are working only with regular functionals, which allows a simpler proof. (Recall that $\fA$ is an algebra over $\CC[[\hbar]]$.)
\begin{cor}
For $A\in\lambda\fA_\reg$,
$\Rcal_A:\fA[[\la]]\to\fA[[\la]]$, and for $V\in\Fcal_\reg$, $R_{0,\la V},R_{\TT,\la V}: \Fcal_\reg[[\hbar,\la]] \to \Fcal_\reg[[\hbar,\la]]$. That is, the M{\o}ller operator contains no negative powers of $\hbar$.
\end{cor}
\begin{proof}
In the expression
\[
\frac{V(\ph-i\hbar \Delta_{S_0}^{\mathrm A}w)-V(\ph)}{-i\hbar} \,,
\]
the numerator is of order $\hbar$, so this cancels the $\hbar$ in the denominator, thus $J$ does not contain any negative powers of $\hbar$. This means that $R_{\TT,\la V}^{-1}: \Fcal_\reg[[\hbar,\la]] \to \Fcal_\reg[[\hbar,\la]]$. 

The fact that $J\approx 1$ to $0$'th order in $\la$ means that $R_{\TT,\la V}^{-1}$ is a formal power series with the identity map in leading order. As a formal power series, it can thus be inverted, and $R_{\TT,\la V}$ exists. Finally, the other forms of the M{\o}ller operator are equivalent to this one.
\end{proof}

Let us now discuss the classical limit. Recall that $\Pcal:\fA\to \Fcal_\mc,\fA_\reg\to\Fcal_\reg$ is the evaluation at the classical limit; it corresponds to setting $\hbar=0$.
\begin{prop}
\label{Classical Moller}
The classical limit of the quantum M{\o}ller operator is the classical M{\o}ller operator:
\[
\Pcal\circ\Rcal_{\la \Qcal_\TT(V)} = r_{\la V}\circ \Pcal \,.
\]
Equivalently, for $F\in\Fcal_\reg$,
	\[
	R_{H,\la V}(F)\big|_{\hbar=0}=r_{\la V}(F)\,.
	\]	 
\end{prop}
\begin{proof}
Again, it is easiest to prove the equivalent statement for $R_{\TT,\la V}$.

To $0$'th order in $\hbar$, 
\[
\frac{V(\ph-i\hbar \Delta_{S_0}^{\mathrm A}w)-V(\ph)}{-i\hbar} 
\approx \left<V^{(1)}(\ph),\Delta_{S_0}^{\mathrm A}w\right> \,,
\]
so
\be
\label{J0}
J(\ph;w) \approx J_0(\ph;w)\doteq\exp \left(\la \left<V^{(1)}(\ph),\Delta_{S_0}^{\mathrm A}w\right>\right) \,.
\ee
Also note that $\TT_H$ is the identity operator to $0$'th order in $\hbar$, so $R_{H,\la V} \approx R_{\TT,\la V}$.

Denote 
\[
\tilde{r}^{\minus}_{\la V}(G)\doteq  R_{0,\la V}^{-1}(G)\Bigr|_{\hbar=0} = R_{\TT,\la V}^{-1}(G)\Bigr|_{\hbar=0} = \Pcal\circ \Rcal_{\la\Qcal_\TT(V)}\circ\Qcal_\TT(G) \,.
\]
 This gives
\begin{align*}
\tilde{r}^{\minus}_{\la V}(G)(\ph) &= (J_0(\ph)\triangleright G)(\ph) \\
&= \sum_{k=0}^\infty \frac{\la^n}{n!}\left<\left(V^{(1)}(\ph)\right)^n, (\Delta_{S_0}^{\mathrm{A}})^{\otimes n} G^{(n)}(\ph)\right> \\
&= \sum_{k=0}^\infty \frac{1}{n!}\left<\left(\la \Delta_{S_0}^{\mathrm{R}}V^{(1)}(\ph)\right)^n, G^{(n)}(\ph)\right>
\end{align*}
This is just a Taylor series expansion, so
\[
\tilde{r}^{\minus}_{\la V}(G)(\ph) = G\left(\ph+\la\Delta_{S_0}^{\mathrm{R}}V^{(1)}(\ph)\right) \,.
\]
In other words, $\tilde{r}^{\minus}_{\la V}$ is the pullback by the operator that acts on $\Ecal[[\la]]$ as
\be\label{inversetilder}
\tilde{\mathtt r}^{\minus}_{\la V}(\ph)=\ph+\la\Delta_{S_0}^{\mathrm{R}}V^{(1)}(\ph)
\ee
Rearranging this gives an equation satisfied by the inverse map,
\[
\tilde{\mathtt r}_{\la V}(\ph)=\ph-\la\Delta_{S_0}^{\mathrm{R}}V^{(1)}(\tilde{\mathtt r}_{\la V}(\ph))\,,
\]
which is the Yang-Feldman equation, and so by Lemma~\ref{YF lemma}, $\tilde{\mathtt r}_{\la V} = {\mathtt r}_{\la V}$.
\end{proof}
\begin{df}
Let $\starHint$ denote the interacting product on $\Fcal_{\reg}[[\hbar,\la]]$ in the identification given by $\Qcal_H$, i.e.,
\[
F\starHint G \doteq \Qcal_H^{-1}\left(\Qcal_HF\bint\Qcal_HG\right) 
=   R_{H,\la V}^{-1}(R_{H,\la V}F\starH R_{H,\la V} G)\,.
\]
Denote the $\starHint$-commutator by
\[
[F,G]_{\starHint}=F\starHint G -G\starHint F\,.
\]
\end{df}
It is now easy to see that the theory defined by $\starHint$ is indeed a quantization of the classical theory defined by the Poisson bracket $\Pei{\;\cdot\;}{\;\cdot\;}_{S_0+\la V}$ given in Definition \ref{PeierlsBracket}.
\begin{prop}
	Let $F,G,V\in \Fcal_{\reg}$, then
\[
\frac{1}{i\hbar}[F,G]_{\starHint}\Bigr|_{\hbar=0}=\Pei{F}{G}_{S_0+\la V}\,,
\]	 
thus $(\Fcal_\reg[[\hbar,\la]],\starHint)$ is a formal deformation quantization of $(\Fcal_\reg[[\la]],\Pei{\cdot}{\cdot}_{S_0+\la V})$.
\end{prop}
\begin{proof}
This is a straightforward consequence of Propositions \ref{IntertwiningProp} and \ref{Classical Moller}.
The quantum M{\o}ller operator intertwines $\starH$ with $\starHint$ (by definition), so it intertwines their commutators. To first order in $\hbar$,
\begin{multline*}
r_{\la V}\left([F,G]_{\starHint}\right) \approx R_{H,\la V} [F,G]_{\starHint}
= [R_{H,\la V}F,R_{H,\la V}G]_{\starH} \\
\approx i\hbar \Pei{r_{\la V}F}{r_{\la V}G}_{S_0} = i\hbar\, r_{\la V} \Pei{F}{G}_{S_0+\la V} \,. 
\end{multline*}
\end{proof}
\begin{rem}
The equivalent, more abstract, statement is (like eq.~\eqref{Poisson property}) that for any $A,B\in\fA_\reg$,
\[
\Pcal\left(\tfrac1{i\hbar}[A,B]_{\bint}\right) = \Pei{\Pcal A}{\Pcal B}_{S_0+\la V} \,.
\]
\end{rem}

Next we show that the quantum M{\o}ller operator can be constructed nonperturbatively, provided the classical M{\o}ller operator is known exactly. To this end, we will extract the ``classical part'' of the  quantum M{\o}ller operator and see what remains. It will then become clear that, at a fixed order in $\hbar$, the remaining ``purely quantum part'' contains only finitely many terms in its coupling constant expansion.
\begin{prop}
\label{quantum correction}
Again, let $r_{\la V}$ be the classical M{\o}ller operator.
Define $\triangleright_r$, $J_1$, and $\Upsilon_{\la V}$ by
\[
K \triangleright_r F \doteq r_{\la V}^{-1}(K\triangleright r_{\la V}F)\,,
\]
\[
J_1(\ph;w) \doteq \exp\left(\la\frac{V(\ph-i\hbar \Delta_{S_0}^{\mathrm A}w)-V(\ph)+i\hbar\left<V^{(1)}(\ph),\Delta_{S_0}^{\mathrm A}w\right>}{-i\hbar}\right)\,,
\]
and
\[
(\Upsilon_{\la V}F)(\ph) \doteq (J_1(\ph) \triangleright_r F)(\ph) \,.
\]
The inverse quantum M{\o}ller operator can be computed as
\be
\label{R factorization}
R_{\TT,\la V}^{-1} = \Upsilon_{\la V}\circ r_{\la V}^{-1} \,.
\ee
\end{prop}
\begin{proof}
The quantum correction $J_1$ is chosen so that $J=J_0J_1$, where $J_0$ is defined in eq.~\eqref{J0}. Proposition~\ref{RfromJ} gives
\begin{align*}
(R_{\TT,\la V}^{-1}F)(\ph) &= (J(\ph)\triangleright F)(\ph) \\
&= [J_0(\ph)\triangleright (J_1(\ph)\triangleright F)](\ph) \,.
\end{align*}
This last expression contains $\ph$ 3 times. Note that the second $\ph$ is just a constant as far as $J_0(\ph)\triangleright$ is concerned; it is not differentiated, and indeed this is also equal to  $ [J_1(\ph)\triangleright (J_0(\ph)\triangleright F)](\ph)$.

Proposition~\ref{Classical Moller} showed that $(r_{\la V}^{-1}G)(\ph) = (J_0(\ph)\triangleright G)(\ph)$. 
Setting $G(\ph)= (J_1(\ph)\triangleright F)(\ph)$ gives
\begin{align*}
(R_{\TT,\la V}^{-1}F)(\ph) &= [r_{\la V}^{-1}(J_1(\ph)\triangleright F)](\ph) \\
&= (J_1(\ph) \triangleright_r r_{\la V}^{-1}F)(\ph)
= (\Upsilon_{\la V}\circ r_{\la V}^{-1}F)(\ph) \,. 
\end{align*}
\end{proof}
\begin{rem}\label{compare:w:beta}
Rearranging \eqref{R factorization} results in $r_{\la V}^{-1}R_{\TT,\la V} = \Upsilon_{\la V}^{-1}$, so $\Upsilon_{\la V}^{-1}$ can be seen an the ``purely quantum'' part of the M{\o}ller operator. Such a map has been introduced in \cite{DHP} for quadratic interactions (where it is called $\beta$ and in our notation we have $\beta\doteq r_{\la V}^{-1}\circ R_{0,\la V}$), so our current discussion is a natural generalization of that result. The only missing step in that comparison is to transform the quantum M{\o}ller map from the $\Qcal_\TT$-identification to the $\Qcal_\Weyl$ or $\Qcal_H$-identification (i.e., to go from $R_{\TT,\la V}$ to $R_{0,\la V}$ or $R_{H,\la V}$). We will come back to this in Section \ref{PPA}.
\end{rem}
\begin{rem}
Note that the fraction in the definition of $J_1$ is of order at least $\hbar$, so that every occurrence of $\la$ is accompanied by a factor of $\hbar$. If $J_1$ is expanded in powers of $\hbar$, then each coefficient contains only finitely many powers of $\la$, so $\la$ no longer needs to be treated as a formal parameter. \emph{This gives a nonperturbative definition of the quantum M{\o}ller operator, provided that the classical M{\o}ller operator exists nonperturbatively.} This existence will be shown in the forthcoming paper \cite{FV}.
\end{rem}

\subsection{The naive product}\label{Naive}
This shows the relationship between the interacting product and the naive product constructed from $\starb$ using the classical M{\o}ller operator. From eq.~\eqref{naive}, the naive product $\startr$ is given by
\be
r_{\la V}^{-1}F  \startr r_{\la V}^{-1}G = r_{\la V}^{-1}(F\starb G) \,.
\ee
The interacting product $\starbint$ in the identification given by $\Qcal_\TT$
is defined by 
\be
\label{startint}
R_{\TT,\la V}^{-1}(F)\starbint R_{\TT,\la V}^{-1}(G) = R_{\TT,\la V}^{-1}(F\starb G) \,.
\ee

Equation \eqref{R factorization} shows that
\[
\Upsilon_{\la V}F \starbint \Upsilon_{\la V}G = \Upsilon_{\la V}(F\startr G) \,.
\]
Note that $\Upsilon_{\la V}$ is the identity map modulo $\hbar$.
This means that  $\Upsilon_{\la V}$, the ``purely quantum'' part of the inverse quantum M{\o}ller operator, is the ``gauge'' equivalence \cite{Kon} relating the naive product $\startr$ to the preferred choice $\starbint$.

Now we are ready to consider the relative merits of the interacting product $\starbint$ and the naive product, $\startr$, defined by eq.~\eqref{naive}. The advantage of $\starbint$ is related to locality and the adiabatic limit. Ultimately the interacting star product will be used to construct local nets of algebras. Consider some bounded causally convex region $\Ocal\subset M$ and the M{\o}ller operator restricted to functionals supported in $\Ocal$. 
Using a local (rather than regular) interaction, it was shown in \cite{BF0,Fredenhagen2015} (in a slightly different setting) that perturbing the interaction by a compactly supported functional, supported outside $\Ocal$ only changes the restriction of the M{\o}ller operator by an inner automorphism. This means that for $F$ and $G$ supported in $\Ocal$, $F\startr G$ (and $F\starHint G$) only depend on the interaction in $\Ocal$. In other words, the net of algebras depends only locally on the interaction. We can therefore first introduce an infrared cutoff for the interaction and then remove it using the algebraic adiabatic limit construction \cite{BF0,r21796}. (Similarly, the retarded and advanced M{\o}ller operators are related by inner automorphisms, so they define exactly the same interacting products.)

The naive products defined from the classical M{\o}ller operator by eq.~\eqref{naive} do not share this property.
To first order in $\la$,
\[
r_{\la V}(F)(\ph) = F(\ph) - \la\left<V^{(1)}(\ph),\Delta^{\mathrm A}_{S_0}F^{(1)}(\ph)\right>+\dots \,.
\]
To first order in $\hbar$,
\[
(F\starb G)(\ph) = F(\ph)G(\ph) -i\hbar\left<F^{(1)}(\ph),\Delta^{\mathrm A}_{S_0}G^{(1)}(\ph)\right>+\dots\,.
\]
This gives a formula for $\startr$ up to first order in $\la$ and in $\hbar$:
\begin{multline*}
    (F\startr G)(\ph) = F(\ph)G(\ph) -i\hbar\left<F^{(1)}(\ph),\Delta^{\mathrm A}_{S_0}G^{(1)}(\ph)\right> \\ + i\la\hbar\left<V^{(2)}(\ph),\Delta^{\mathrm A}_{S_0}F^{(1)}(\ph)\otimes\Delta^{\mathrm A}_{S_0}G^{(1)}(\ph)\right>\\ + i\la\hbar\left<V^{(2)}(\ph),\Delta^{\mathrm R}_{S_0}F^{(1)}(\ph)\otimes\Delta^{\mathrm A}_{S_0}G^{(1)}(\ph)\right> + \dots\,.
\end{multline*}
The third term is problematic, even if $V$ is a local functional. Because of the advanced propagators, this product will change if $V$ is changed anywhere to the past of the support of $F$ and $G$. The same is true of the naive product constructed from $\starH$.


\section{Graphical computations}
To simplify and organize the computations, we will represent our structures in terms of graphs.
\begin{df}
$\Gcal(n)$ is the set of isomorphism classes of directed graphs with $n$ vertices labelled $1,\dots, n$ (and possibly unlabelled vertices with valency $\geq1$). 

Also denote
\[
\Gcal \doteq \bigcup_{n\in\NN} \Gcal(n) \,.
\]
We make $\Gcal$ a category by defining a \emph{morphism} $u:\alpha\to\beta$ to be a function that 
\begin{itemize}
\item
maps vertices of $\alpha$ to vertices of $\beta$ and edges of $\alpha$ to edges or vertices of $\beta$,
\item
 respects sources and targets of edges,
 \item
  maps labelled vertices to labelled vertices, 
  \item
  and preserves the order of labelled vertices.
\end{itemize}

For $\gamma\in\Gcal$: $e(\gamma)$ is the number of edges; $v(\gamma)$ is the number of unlabelled vertices; $\Aut(\gamma)$ is the group of automorphisms.
\end{df}
\begin{rem}
This definition of morphism differs from the standard definition of a map of graphs \cite[Sec.~II.7]{MacLane}. It is chosen to give the concept of extension that will be useful below. The meaning of isomorphism is the same.
\end{rem}
\begin{df}
A graph $\gamma\in\Gcal(n)$ determines an $n$-ary multidifferential operator, $\vec\gamma$, on functionals as follows: 
\begin{itemize}
\item
An edge represents $\Delta_{S_0}^{\mathrm A}(x,y)$ with the direction from $y$ to $x$ --- i.e., such that this is only nonvanishing when the edge points from the future to the past;
\item
if the labelled vertex $j$ has valency $r$, this represents the order $r$ derivative of the $j$'th argument;
\item
likewise, an unlabelled vertex of valency $r$ represents $V^{(r)}$.
\end{itemize}
\end{df}
\begin{df}
In diagrams, $\Delta_{S_0}^{\mathrm A}$ will be denoted by a dashed line with an arrow, so all graphs in $\Gcal$ will be drawn with dashed lines for the edges.
\end{df}

\subsection{Non-interacting product}
\begin{df}
$\Gcal_1(2)\subset\Gcal(2)$ is the subset of graphs with no unlabelled vertices in which all edges go from $2$ to $1$.
\end{df}

In terms of these, the non-interacting product can be expressed as
\be
\label{graph free product}
F\starb G = \sum_{\gamma\in\G_1(2)} \frac{(-i\hbar)^{e(\gamma)}}{\abs{\Aut \gamma}} \vec\gamma(F,G) .
\ee
\begin{exa}
The $\hbar^3$ term is given by the graph
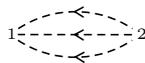
\[
\GraphFive
\]
which is the unique graph in $\Gcal_1(2)$ with 3 edges.  Its automorphism group is $S_3$, which has order $6$ and gives the correct coefficient, $\frac16(-i\hbar)^3$.
\end{exa}

\subsection{Perturbative calculations}
\subsubsection{Inverse M{\o}ller operator} 
\begin{df}
$\Gcal_2(1)\subset\Gcal(1)$ is the set of  graphs  such that every edge goes from $1$ to an unlabelled vertex.
\end{df}
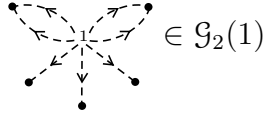
\begin{exa}
\[
\GraphSeven \in \Gcal_2(1)
\]
has automorphism group $S_3 \times (\ZZ_2\ltimes\ZZ_2^2)$.
\end{exa}
\begin{lemma}
\be
\label{graph R}
R_{\TT,\la V}^{-1}(F) = \sum_{\gamma\in\Gcal_2(1)}\frac{(-i\hbar)^{e(\gamma)-v(\gamma)}\la^{v(\gamma)}}{\abs{\Aut(\gamma)}} \vec{\gamma}(F)
\ee
\end{lemma}
\begin{proof}
Taking a Taylor expansion of $V$ about $\ph$ in eq.~\eqref{generating function} shows that
\begin{align*}
J(\ph;w) &= \exp\left(\la \sum_{n=1}^\infty \frac{(-i\hbar)^{n-1}}{n!}\left<V^{(n)}(\ph),(\Delta^{\mathrm A}_{S_0}w)^{\otimes n}\right>\right) \\
&= \prod_{n=1}^\infty \exp\left(\frac{\la(-i\hbar)^{n-1}}{n!}\left<V^{(n)}(\ph),(\Delta^{\mathrm A}_{S_0}w)^{\otimes n}\right>\right) \\
&= \prod_{n=1}^\infty \sum_{k=0}^\infty \frac1{k!} \left(\frac{\la(-i\hbar)^{n-1}}{n!}\left<V^{(n)}(\ph),(\Delta^{\mathrm A}_{S_0}w)^{\otimes n}\right>\right)^k
\end{align*}
The term $\left<V^{(n)}(\ph),(\Delta^{\mathrm A}_{S_0}w)^{\otimes n}\right>$ gives the same differential operator as the graph with $n$ edges directed from the vertex $1$ to a single unlabelled vertex. We have here a sum over all possible products of such terms. This corresponds to all possible graphs in $\Gcal_2(1)$. 

The automorphisms of a graph $\gamma\in\Gcal_2(1)$ do not permute the unlabelled vertices with different valencies, thus the automorphism group is a Cartesian product of the group of automorphisms for each of those subgraphs. 

If $\gamma$ has $k$ unlabelled vertices of valency $n$, then an automorphism can permute each of those vertices, and for each such vertex, it can permute the $n$ edges leading to it. That subgroup of automorphisms is thus a semidirect product of $S_k$ and $(S_n)^k$. This has order $k! (n!)^k$. 
\end{proof}

\subsubsection{Composition of operators}
The following is analogous to the concept of an extension of groups.
\begin{df}
For $\alpha\in\Gcal(n)$ and $\gamma\in\Gcal(m)$, an \emph{extension} of $\gamma$ by $\alpha$ at $j\in\gamma$ is a pair of an injective and a surjective morphism
\[
\alpha \stackrel{u}\hookrightarrow \beta \stackrel{v}\twoheadrightarrow \gamma
\]
with $\beta\in\Gcal(n+m-1)$, such that $u(\alpha)=v^{-1}(j)$, and the restriction of $v$ to the compliment of $u(\alpha)$ is injective.

Two extensions are \emph{equivalent} if there exists a commutative diagram between them with the identity on $\alpha$ and $\gamma$, and an automorphism on $\beta$.
\end{df}
\begin{exa}
There is an extension
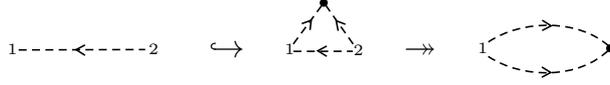
\[
\GraphEight \quad \hookrightarrow\quad \GraphSixteen \quad\twoheadrightarrow \quad\GraphNine 
\]
which maps the first graph to the bottom edge of the second graph, and then collapses that subgraph to the vertex $1$ of the third graph.
\end{exa}
\begin{df}
The \emph{partial composition} of multilinear maps is denoted by $\circ_j$ and means the composition of one map into the $j$'th argument of another.
\end{df}
\begin{lemma}
\label{Partial Composition}
Let $\alpha\in\Gcal(n)$ and $\gamma\in\Gcal(m)$ and $j=1,\dots,m$. The partial composition  at $j$ is
\[
\vec\gamma\circ_j\vec\alpha = \sum_{\alpha \hookrightarrow \beta \twoheadrightarrow \gamma} \vec\beta
\]
where the sum is over equivalence classes of extensions of $\gamma$ by $\alpha$ at $j$.

Equivalently, 
\[
\frac1{\abs{\Aut\gamma}\abs{\Aut\alpha}} \vec\gamma\circ_j\vec\alpha
= \sum_{\exists\;\alpha \hookrightarrow \beta \twoheadrightarrow \gamma} \frac{d_\beta}{\abs{\Aut\beta}} \vec\beta
\]
where the sum is over the set of $\beta\in\Gcal(n+m-1)$ such that there exists such an extension of $\gamma$ by $\alpha$ at $j$,  and  $d_\beta$ is the number of subgraphs of $\beta$ isomorphic to $\alpha$ such that the quotient is isomorphic to $\gamma$ with the subgraph mapped to $j$.
\end{lemma}
\begin{proof}
The composition $\gamma\circ_j\alpha$ can be thought of as inserting $\alpha$ in place of $j\in\gamma$. If $r$ is the valency of $j\in\gamma$, then this is taking an order $r$ derivative of $\vec\alpha$. The edges represent $\Delta_{S_0}^{\mathrm A}$, which is constant, so the product rule tells us that this derivative is given by a sum over all possible ways of attaching the $r$ edges to vertices of $\alpha$ (instead of to $j\in\gamma$). Any such attachment gives a graph $\beta$ with a subgraph identified with $\alpha$ and the quotient identified with $\gamma$. This gives precisely the set of equivalence classes of extensions as defined above.

For the second expression, we need to write this as a sum over the graphs $\beta$ that can appear as extensions, but one graph can appear in inequivalent extensions, so we need to understand the number of extensions with a given $\beta$.

An extension $\alpha \hookrightarrow \beta \twoheadrightarrow \gamma$ certainly determines a subgraph of $\beta$ that is isomorphic to $\alpha$. If two extensions determine the same subgraph, then it is easy to construct an automorphism of $\beta$ that gives an equivalence between the extensions. (All of the edges of $\beta$ are either in the subgraph or map to edges of $\gamma$.) So, having the same subgraph is a weaker condition than equivalence of extensions.

By definition, $d_\beta$ is the number of possible images of $\alpha$ in extensions. Any two extensions with the same image are related by an automorphism of $\alpha$ and an automorphism of $\gamma$, therefore the number of extensions is
\[
d_\beta \abs{\Aut\alpha}\abs{\Aut\gamma} \,.
\]
An equivalence of extensions is always given by an automorphism of $\beta$, so an equivalence class consists of $\abs{\Aut\beta}$ extensions, therefore the number of equivalence classes of extensions is
\[
\frac{d_\beta \abs{\Aut\alpha}\abs{\Aut\gamma}}{\abs{\Aut\beta}} \,. \qedhere
\]
\end{proof}
\begin{exa}
\[
\GraphEight \circ_2 \GraphThirteen = \GraphFourteen + 2\,\GraphFifteen
\]
because the last graph occurs in two inequivalent extensions.
\end{exa}
\begin{exa}
\[
\GraphSixteen \circ_1 \GraphOneOpposite = \GraphSeventeen + \GraphSquare+\GraphEighteenOpposite+\GraphSquareRotated
\]
Each of these graphs has trivial automorphism group, except for the last, for which it has order $2$. However, there are also $2$ ways of mapping $\GraphOneOpposite$ into it to give the quotient $\GraphSixteen$. This is why the coefficient of that term is also $1$.
\end{exa}

\subsubsection{Interacting product}
\begin{df}
$\Gcal_2(2)\subset\Gcal(2)$ is the subset of  graphs such that every edge either goes from a labelled vertex to an unlabelled vertex or from $2$ to $1$.
\end{df}
\begin{lemma}
\label{G2Lemma}
\be
\label{G2 formula}
R_{\TT,\la V}^{-1}(F\starb G) = \sum_{\gamma\in\Gcal_2(2)}\frac{(-i\hbar)^{e(\gamma)-v(\gamma)}\la^{v(\gamma)}}{\abs{\Aut(\gamma)}} \vec{\gamma}(F,G)
\ee
\end{lemma}
\begin{proof}
This follows by inserting eq.~\eqref{graph free product} into eq.~\eqref{graph R} and using the product rule.
\end{proof}

\begin{df}
$\Gcal_3(n)\subset\Gcal(n)$ is the set of  graphs such that:
\begin{itemize}
\item
Every unlabelled vertex has at least one ingoing edge and one outgoing edge;
\item
there are no directed cycles;
\item
for $1\leq j<k\leq n$, there does not exist any directed path from $j$ to $k$.
\end{itemize}
\end{df}
In particular, this implies that $1$ is a sink (has only ingoing edges) and $n$ is a source (has only outgoing edges).

We are now ready to write down the graphical expansion of the interacting star product. Note that the graphs we are using are constructed from the free propagator $\Delta^{\rm A}_{S_0}$ and from the derivatives of the interaction term $V$. Later on, we will re-express things in terms of the propagator $\Delta^{\rm A}_{S}$ of the interacting theory and the derivatives of the full action $S$.
\begin{thm}
\label{Interacting product}
\be
\label{graph star V}
F\starbint G = \sum_{\gamma\in\Gcal_3(2)}\frac{(-i\hbar)^{e(\gamma)-v(\gamma)}(-\la)^{v(\gamma)}}{\abs{\Aut(\gamma)}} \vec{\gamma}(F,G)
\ee
\end{thm}
\begin{rem}
The right hand side is the same as in eq.~\eqref{G2 formula}, except that  $\Gcal_3$ has replaced $\Gcal_2$ and $\la$ has become $-\la$.
\end{rem}
\begin{df}
$\Gcal_4(2)\subset\Gcal(2)$ is the set of graphs obtained by extending graphs in $\Gcal_3(2)$ by graphs in $\Gcal_2(1)$ at $1$ and $2$.
\end{df}
\begin{proof}
For this proof, let $F \star_{\sst ?} G$ denote the right hand side of eq.~\eqref{graph star V}, so that we need to prove $\star_{\sst ?}=\starbint$. That is, we need to prove that $R_{\TT,\la V}^{-1}(F\starb G) = R_{\TT,\la V}^{-1}(F)\star_{\sst ?} R_{\TT,\la V}^{-1}(G)$, and eq.~\eqref{G2 formula} has already computed the left side.

By Lemma \ref{Partial Composition} and the definition of $\Gcal_4$,  $R_{\TT,\la V}^{-1}(F)\star_{\sst ?} R_{\TT,\la V}^{-1}(G)$ can be computed as some sum over $\Gcal_4(2)$. We will show that cancellation reduces this to the sum over $\Gcal_2(2)$ in eq.~\eqref{G2 formula}.

We first need to check that $\Gcal_2(2)\subseteq \Gcal_4(2)$, so suppose that $\gamma\in \Gcal_2(2)$. By definition, all edges of $\gamma$ either go from $1$ or $2$ to an unlabelled vertex or from $2$ to $1$. Let $\alpha\subset\gamma$ be the subgraph of edges from $1$ (to unlabelled vertices) and denote the quotient graph as $\gamma/\alpha$. Let $\beta\subset\gamma$ be the subgraph of edges going from $2$ to (unlabelled) vertices not in $\alpha$. All other edges of $\gamma$ must go from $2$ to vertices of $\alpha$, therefore 
\[
(\gamma/\alpha)/\beta \in \Gcal_1(2)\subset \Gcal_3(2)
\]
and so $\gamma\in\Gcal_4(2)$.

Eqs.~\eqref{graph star V} and \eqref{graph R} give
\[
R_{\TT,\la V}^{-1}(F)\star_{\sst ?} R_{\TT,\la V}^{-1}(G) = \sum_{\alpha,\bet\in\Gcal_2(1)}\sum_{\delta\in\Gcal_3(2)} \frac{(-i\hbar)^{e-v} (-1)^{v(\delta)} \la^{v}}{\abs{\Aut\alpha}\abs{\Aut\beta}\abs{\Aut\delta}} \vec\delta(\vec\alpha F,\vec\beta G) 
\]
where $e=e(\alpha)+e(\beta)+e(\delta)$ and $v=v(\alpha)+v(\beta)+v(\delta)$. Applying Lemma \ref{Partial Composition}, this becomes
\begin{equation*}
R_{\TT,\la V}^{-1}(F)\star_{\sst ?} R_{\TT,\la V}^{-1}(G) = 
 \sum_{\gamma\in\Gcal_4(2)}  \frac{(-i\hbar)^{e(\gamma)-v(\gamma)}  \la^{v(\gamma)}}{\abs{\Aut\gamma}} \vec\gamma(F,G) 
 \sum_{\alpha,\beta}(-1)^{v(\gamma)-v(\alpha)-v(\beta)} ,
\end{equation*}
where the last sum is over $\alpha,\beta\subset\gamma$, such that $1\in\alpha$, $2\in\beta$, $\alpha,\beta\in\Gcal_2(1)$, and $(\gamma/\alpha)/\beta \in \Gcal_3(2)$.

Note that $\alpha\subset\gamma$ is the subgraph of all edges outgoing from $1\in\gamma$, so it is uniquely determined by $\gamma$. On the other hand, although $\beta$ must contain all edges from $2$ to sinks that are not in $\alpha$, it \emph{may} contain any edge that goes from $2$ to any other unlabelled vertex not in $\alpha$ (see Example~\ref{G4 example}). The sum over $\beta$ is over the binary choices of including or not including each of these edges, thus 
\[
\sum_\beta (-1)^{v(\beta)} = 0
\]
\emph{if there are any such edges}. This reduces the expression for $R_{\TT,\la V}^{-1}(F)\star_{\sst ?} R_{\TT,\la V}^{-1}(G)$ to a sum over $\gamma$ without any such edges.

We are therefore interested in graphs $\gamma\in\Gcal_4(2)$ that do not contain any such ambiguous edge. This means that any edge from $2\in\gamma$ must go to a vertex of $\alpha\subset\gamma$ or to a sink. By the definition of $\Gcal_4$, $1\in \gamma/\alpha$ is a sink, so any vertex of $\alpha\subset\gamma$ other than $1$ is a sink. In short, any edge from $2\in\gamma$ must go to $1$ or a sink. 
This implies that all unlabelled vertices are sinks, and so any edge goes from $1$ or $2$ to an unlabelled vertex or from $2$ to $1$. In other words, $\gamma\in\Gcal_2(2)$. In that case, all vertices of $\gamma$ are in $\alpha$ or $\beta$, so $v(\gamma)=v(\alpha)+v(\beta)$. Therefore,
\[
 R_{\TT,\la V}^{-1}(F)\star_{\sst ?} R_{\TT,\la V}^{-1}(G) =
 \sum_{\gamma\in\Gcal_2(2)} \frac{(-i\hbar)^{e(\gamma)-v(\gamma)} \la^{v(\gamma)}}{\abs{\Aut\gamma}} \vec\gamma(F,G) \,.
\]
With Lemma~\ref{G2Lemma}, this shows that $\star_{\sst ?}$ satisfies eq.~\eqref{startint}, which is the defining property of $\starbint$.
\end{proof}

\begin{exa}
\label{G4 example}
Consider the graph
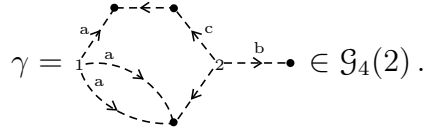
\[
\gamma=\GraphTwenty \in\Gcal_4(2)\,.
\]
The subgraph $\alpha$ consists of the edges labelled (a) here (and the adjacent vertices). The edge (b) must be in the subgraph $\beta$, but the edge (c) may or may not be in $\beta$. The other edges cannot be in $\beta$.  The sum over the 2 possible choices of $\beta$ with or without (c) gives $0$, so that $\gamma$ does not contribute to the formula for $R_{\TT,\la V}^{-1}(F)\starbint R_{\TT,\la V}^{-1}(G)$.
\end{exa}

\subsection{Nonperturbative expression for an interacting product}
In this section we will show that the interacting star product can be written in terms of the full propagator $\Delta^{\rm A}_{S}$ and derivatives of $S$ and, provided that $\Delta^{\rm A}_{S}$ is known exactly, the result is a formal power series in $\hbar$, but the coupling constant $\la$ can be treated as a number.
\begin{df}
$\Gcal_5(n)$ is the set of isomorphism classes of directed graphs with labelled vertices $1,\dots,n$ such that:
\begin{itemize}
\item
Each unlabelled vertex has valency \textbf{at least 3} and is neither a source nor a sink;
\item
there exist no directed cycles;
\item
for $1\leq j < k \leq n$, there does not exist a directed path from $j$ to $k$.
\end{itemize}
Also define
\[
\Gcal_5\doteq\bigcup_{n\in\NN}\Gcal_5(n) \,.
\]
\end{df}
\begin{rem}
	A crucial consequence of the first condition in the definition of $\Gcal_5(n)$ is that derivatives of $V$ appearing in the graphical expansion are at least 3rd derivatives, so one can replace these with derivatives of $S$ ($S_0$ is quadratic).
\end{rem}
\begin{df}
A graph $\gamma\in\Gcal_5(n)$ defines an $n$-ary multidifferential operator, $\acts\gamma$, as follows: 
\begin{itemize}
\item
An edge represents $\Delta^{\mathrm A}_S$.
\item
The labelled vertex $j$ represents a derivative of the $j$'th argument.
\item
An unlabelled vertex represents a variational derivative of $S$.
\end{itemize}
\end{df}
\begin{df}
In our diagrams, $\Delta^{\mathrm A}_S$ will be denoted by a solid line with an arrow. For this reason, graphs in $\Gcal_5$ will be drawn with solid lines for the edges. (This helps to distinguish $\Gcal_5$ from $\Gcal$.)
\end{df}
\begin{rem}
$S''$ is the linearized equation of motion operator, thus $S''\Delta^{\mathrm A}$ is the identity operator. This means that if there were a bivalent vertex in $\gamma$, then $\acts\gamma$ would be the same as if that vertex were removed, i.e., 
\[
\SolidLine = \SolidLineWithVertex \,.
\]
For this reason, $\Gcal_5(n)$ can (and should) be thought of as a quotient of $\Gcal_3(n)$. This leads to the appropriate definition of morphisms.
\end{rem}
\begin{df}
For $\alpha,\beta\in\Gcal_5$, a \emph{morphism} $u:\alpha\to\beta$ is a function that 
\begin{itemize}
\item
maps vertices of $\alpha$ to vertices of $\beta$ and edges of $\alpha$ to \textbf{directed paths} in $\beta$,
\item
respects sources and targets of edges,
\item 
maps labelled vertices to labelled vertices, 
\item
and preserves the order of labelled vertices.
\end{itemize}
\end{df}
\noindent The concept of an extension in $\Gcal_5$ follows from this definition.

\vspace{1ex}
The following theorem is the main result of this section. It delivers an explicit formula for the interacting star product that is nonperturbative in the coupling constant.
\begin{thm}
\label{Reduced star V}
For $S=S_0+\la V$,
\be
\label{reduced star V}
F \starbint G = \sum_{\gamma\in\Gcal_5(2)} \frac{(-1)^{v(\gamma)}(-i\hbar)^{e(\gamma)-v(\gamma)}}{\abs{\Aut\gamma}} \acts\gamma(F,G) \,.
\ee
\end{thm}
\begin{proof}
For $\gamma\in\G_5(n)$, the operator $\acts\gamma$ can be expressed as a sum of operators given by graphs in $\Gcal_3(n)$, with the following dictionary:
\begin{itemize}
\item
Because any unlabeled vertex in $\gamma$ has valency at least $3$ (say, $r$), $S^{(r)} = \la V^{(r)}$, so this vertex corresponds to a vertex in $\Gcal_3$ and a factor of $\la$.
\item
By eq.~\eqref{regular advanced}, an edge in $\gamma$ corresponds to a sum over all possible chains of edges and bivalent vertices in $\Gcal_3$, with a factor of $-\la$ for every vertex, i.e.,
\[
\SolidLine = \GraphTen \;-\; \la \GraphEleven \;+\; \la^2 \GraphTwelve \;- \dots \,.
\] 
\end{itemize}

Adding bivalent vertices along edges of graphs in $\Gcal_5(n)$ will give all graphs in $\Gcal_3(n)$.

Applying this dictionary to \eqref{reduced star V} gives \eqref{graph star V}. 
\end{proof}

\begin{rem}
	Note that after re-expressing everything in terms of full propagators, at a fixed order in $\hbar$ there are only finitely many terms in the $\la$-expansion. Hence the result is exact in the coupling constant, provided that $\Delta^{\rm A}_{S_0}$ can be constructed.
\end{rem}
Theorem \ref{Interacting product} shows that eq.~\eqref{reduced star V} gives the interacting product for a free action plus a regular perturbation (up to time ordering). This implies in particular that the product is associative. It is worth understanding why it is associative in greater generality.
\begin{thm}
\label{Associativity}
Suppose that  $S$ is an action and $K(\ph)(x,y)$ is any Green's function for the linearized equation of motion. If $\star$ is defined by the right hand side of  eq.~\eqref{reduced star V}, with $K$ in place of $\Delta^{\mathrm A}_S$, then $\star$ is an associative product on the domain of definition of the associativity condition.
\end{thm}
\begin{proof}
As usual, the dependence on $\ph$ will not be written explicitly. 

The defining property of a Green's function is $S''K=\id_{\Ecal_c}$. Differentiating this gives
\[
0 = S^{(3)}K + S''K^{(1)} \,,
\]
and multiplying on the left by $K$ gives
\[
0 = K S^{(3)}K + K S''K^{(1)} \,.
\]
In this equation, $K S''$ is acting on the image of $K$, where it acts as the identity, so 
\[
K^{(1)} = - K S^{(3)} K \,. \]

We would first like to compute $(F_1\star F_2)\star F_3$.

To compute this, first observe that for $\alpha,\gamma\in\Gcal_5(2)$, 
\[
\acts\gamma\circ_1\acts\alpha = \sum_{\alpha\hookrightarrow\beta\twoheadrightarrow\gamma}  (-1)^{v(\alpha)-v(\beta)+v(\gamma)} \acts\beta
\]
where the sum is over equivalence classes of extensions at $1$. 
 These are precisely the ways of breaking up $1\in\gamma$ and attaching the edges to $\alpha$ (possibly by adding new vertices) to form a graph $\beta$.
 
The axioms of $\Gcal_5$ imply that there exists a partial order, $\preceq$, on the vertices of $\beta\in\Gcal_5(3)$ that is generated by the edges and $1\succ2\succ3$.

This preimage of $1\in\gamma$ is completely determined by the structure of $\beta$ alone. It is the complete subgraph whose vertices satisfy $\succeq 2$. For this reason, $d_\beta=1$ in the sense of Lemma~\ref{Partial Composition}. The number of equivalence classes of extensions is thus
\[
\frac{\abs{\Aut\alpha}\abs{\Aut\gamma}}{\abs{\Aut\beta}} \,.
\]

This uniqueness shows that any $\beta$ can only appear in one term of the expansion of $(F_1\star F_2)\star F_3$. Conversely, any $\beta\in\Gcal_5(3)$ does occur in this expansion.

Together, this shows that 
\be
(F_1\star F_2)\star F_3 = \sum_{\beta\in\Gcal_5(3)} \frac{(-1)^{v(\beta)}(-i\hbar)^{e(\beta)-v(\beta)}}{\abs{\Aut\beta}} \acts\beta(F_1,F_2,F_3) \,.
\ee
An essentially identical calculation (using the subgraph determined by $\preceq2$) shows that $F_1\star(F_2\star F_3)$ is given by the same formula, thus
\[
(F_1\star F_2)\star F_3= F_1\star(F_2\star F_3) \,,
\]
provided that both sides are defined.
\end{proof}

\begin{exa}
Consider the graph 
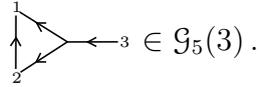
\[
\GraphTwentyOne \in \Gcal_5(3) \,.
\]
This occurs in a unique extension at $1$,
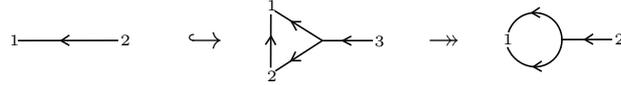
\[
\SolidLineOneTwo \quad\hookrightarrow\quad 
\GraphTwentyOne
\quad\twoheadrightarrow\quad \SolidLineLoop
\]
corresponding to a term in $(F_1\star F_2)\star F_3$, and in a unique extension at $2$
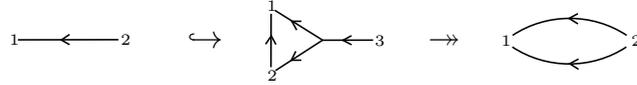
\[
\SolidLineOneTwo \quad\hookrightarrow\quad 
\GraphTwentyOne
 \quad\twoheadrightarrow\quad \GraphTwoSolid
\]
corresponding to a term in $F_1\star(F_2\star F_3)$.
\end{exa}

\begin{rem}
We have not proven that eq.~\eqref{reduced star V} gives the correct interacting product for a local action. However, Theorem~\ref{Associativity} makes this a plausible conjecture.
\end{rem}

\subsubsection{Low order terms}
Explicitly, the product $\starbint$ is given up to order $\hbar^3$ as
\[
m_{\starbint} = m + \hbar B^{\starbint}_1 + \hbar^2 B^{\starbint}_2 +\hbar^3 B^{\starbint}_3 + \dots
\]
where 
\be
B^{\starbint}_1 = -i\SolidLineOneTwo\,,
\ee
\be
B^{\starbint}_2 = \frac{-1}2 \GraphTwoSolid + \frac12  \SolidLineLoop+ \frac12 \SolidLineLoopReverse - \frac12 \SolidLineLoopDouble \,,
\ee
and
\begin{multline}
B^{\starbint}_3 = \frac{i}6 \GraphTwoSolidTripple
- \frac{i}2 \GraphTwentyTwo
- \frac{i}2 \GraphTwentyTwoOpposite
+ \frac{i}2 \GraphTwentyTwoMiddle\\
 +i \GraphTwentyTwoWithLine
- \frac{i}4 \GraphTwoSolidDouble
+ \frac{i}4 \SolidLineLoopTwo
- \frac{i}6 \SolidLineLoopMiddle\\
 + \frac{i}2 \GraphTwentyTwoLine
+ \frac{i}2 \GraphTwentyTwoLineOpposite
- \frac{i}2 \GraphTwentyTwoLineMiddle
-i \GraphTwentyThree \\
+ \frac{i}4 \GraphTwoSolidDoubleLine
- \frac{i}4 \SolidLineLoopTwoLine
- \frac{i}6 \SolidLineLoopMiddleLeft\\ 
+ \frac{i}2 \GraphTwentyTwoLineLeft
+ \frac{i}2 \GraphTwentyTwoOppositeLeft
- \frac{i}2 \GraphTwentyTwoLineMiddleLeft
-i \GraphTwentyThreeLeft \\
+ \frac{i}4 \GraphTwoSolidDoubleLineLeft
- \frac{i}4 \SolidLineLoopTwoLineLeft
+ \frac{i}6 \SolidLineLoopMiddleSymm\\
 - \frac{i}2 \GraphTwentyTwoLineSymm 
 - \frac{i}2 \GraphTwentyTwoLineOppositeSymm 
 + \frac{i}2 \GraphTwentyTwoLineMiddleSymm \\
 +i \GraphTwentyThreeSymm
- \frac{i}4 \GraphTwoSolidDoubleLineSymm\\
 + \frac{i}4 \SolidLineLoopTwoLineSymm\,.
\end{multline}

\subsubsection{Is there a Kontsevich-type formula?}\label{section:Kon}
In his famous paper on deformation quantization \cite{Kon}, Kontsevich presented a formula for constructing a $\star$-product from an arbitrary Poisson structure on a finite-dimensional vector space. Every term is a polynomial in the Poisson structure and its derivatives. In that construction,  $B_1$  is antisymmetric and proportional to the Poisson structure, so this can be thought of as constructing a $\star$-product from its first order term.

The first order term of $\starbint$ is $\Delta^{\mathrm A}_S$, which is not antisymmetric. This suggests a question. In analogy with Kontsevich's formula, can $\starbint$ be constructed from $\Delta^{\mathrm A}_S$ and its functional derivatives?

To address this, we first need some notation. Kontsevich's formula uses a sum over graphs in which vertices represent the Poisson structure. In his graphs, every unlabelled vertex has 2 outgoing edges. Because $\Delta^{\mathrm A}_S$ has no symmetry, in our generalization, it will be necessary to distinguish these as left and right edges.

\begin{df}
A \emph{K-graph} is a directed graph in which:
\begin{itemize}
\item
There are labelled vertices $1$ and $2$, and possibly unlabelled vertices;
\item
every edge is labelled as ``left'' or ``right'';
\item
every unlabelled vertex has 2 outgoing edges, one left and one right.
\end{itemize}
As with other graphs that we have considered, a K-graph determines a bidifferential operator. The vertices $1$ and $2$ represent the arguments. The unlabelled vertices represent $\Delta^{\mathrm A}_S$. The edges represent derivatives.

In diagrams, these will be drawn with solid arrows on the right edges.
\end{df}

For example,
\[
B^{\starbint}_1  
= -i\GraphTwentyFour\,.
\]
(Equality means equality of operators.)

To see how K-graphs can be expressed in  terms of graphs in which edges represent $\Delta^{\mathrm A}_S$, consider the following examples involving \textit{parts of graphs}:
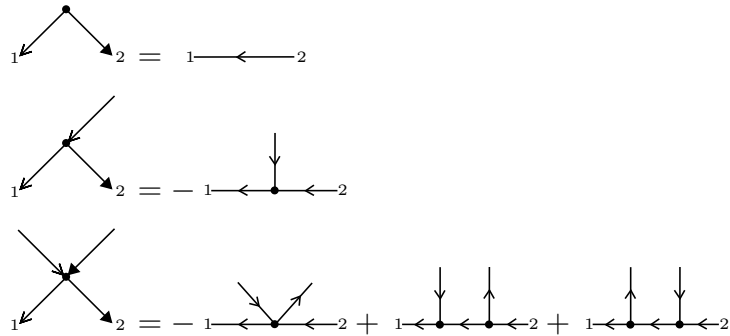
\begin{align*}
\GraphKOne&=\SolidLineOneTwo\\
\GraphKTwo&=-\GraphThirtyTwo\\
\GraphKThree&=-\GraphThirtyThree+\GraphThirtyFour+\GraphThirtyFourOpposite
\end{align*}
Using the rules above applied to parts of graphs, one can easily treat an arbitrary K-graph.

The second order term of $\starbint$ can indeed be constructed in this way:
\begin{multline*}
B^{\starbint}_2 = \frac{-1}2 \GraphTwentyFive - \frac12 \GraphTwentySix - \frac12 \GraphTwentySixOpposite\\ - \frac12 \GraphTwentySeven \,.
\end{multline*}
In general, the operator given by a K-graph can also be given by a sum of graphs in which edges represent $\Delta^{\mathrm A}_S$ and vertices represent derivatives of $S$, but not \emph{vice versa}. A K-graph in which each unlabelled vertex has at most one ingoing edge will give a single term; otherwise, it will give several terms.


We can come fairly close to expressing $B^{\starbint}_3$ in terms of K-graphs. The only problem is with terms in which the labeled vertices both have valency~1. We can nearly reproduce this part of $-iB^{\starbint}_3$ as
\begin{multline*}
\frac16\GraphTwentyEight+\frac14\GraphTwentyNine\\+\frac16\GraphThirty+\frac16\GraphThirtyOpposite\\
= \frac16\SolidLineLoopMiddleSymm - \frac12 \GraphTwentyTwoLineSymm - \frac12 \GraphTwentyTwoLineOppositeSymm\\
+ \frac23 \GraphTwentyTwoLineMiddleSymm 
+ \frac{11}{12} \GraphTwentyThreeSymm\\
- \frac14 \GraphTwoSolidDoubleLineSymm + \frac14 \SolidLineLoopTwoLineSymm
\end{multline*}

This leaves a discrepancy of 
\be
\label{remainder}
\frac16\GraphTwentyTwoLineMiddleSymm  - \frac1{12}  \GraphTwentyThreeSymm \,.
\ee
Note that these two terms each have 4 unlabelled vertices. A K-graph at order $\hbar^3$ has precisely 3 unlabelled vertices, and when it is translated at least one term has 3 or fewer unlabelled vertices. In order to express \eqref{remainder} as a combination of K-graphs, we must in particular find combinations of K-graphs in which terms with 3 or fewer unlabelled vertices cancel.
A tedious search shows that those combinations cannot reproduce \eqref{remainder}, and thus the third order term of $\starbint$ cannot be constructed from $\Delta^{\mathrm A}_S$ and its derivatives. This means that $\starbint$ is not given by anything analogous to Kontsevich's formula.

\subsubsection{Interacting Wick product}\label{section:Ren}
These calculations have used the identification of $\fA_\reg$ with $\Fcal_\reg[[\hbar]]$ defined by the quantization map $Q_\TT$. In order to compare results and to try to extend this to all of $\fA$, we should use the identification defined by $\Qcal_H$. This gives a product that is related to the abstract interacting product $\bint$ in the same way that the Wick product is related to the abstract non-interacting product $\bullet$.

An expression for the interacting product $\starHint$ in that identification is obtained from $\starbint$  by
\be
\label{starHint}
F \starHint G = \TT_H(\TT_H^{-1}F \starbint \TT_H^{-1}G)
\ee
for $F,G\in \Fcal_\reg[[\hbar,\la]]$.

The next step is to extend the domain of definition of $\starHint$ to local non-linear arguments. 
The potential problem with this is that the expansion of the product $\starbint$ in terms of Feynman diagrams contains loops involving the advanced propagator, and light-cone divergences could be present.
\begin{rem}
Note that the advanced propagator, as a bi-distributional kernel $\Delta^{\mathrm A}_S(x,y)$ has singular support given by the condition $(x-y)^2=0$ (lightcone equation). To build pointwise products of distributions one can use H{\"o}rmander's criterion \cite{Hoer1}, based on the concept of the wavefront set. Roughly speaking, one can multiply distributions if the sum of their wavefront sets (seen as subsets of the cotangent bundle) does not include the zero section. This condition is violated for all point on the whole past lightcone if we try to take the square of $\Delta^{\mathrm A}_S(x,y)$. For details details about the wavefront sets of fundamental solutions, see for example \cite{Rad}.
\end{rem}

\begin{df}
$\Gcal_6(n)$ is the set of isomorphism classes of graphs with directed and undirected edges and labelled vertices $1,\dots,n$ such that:
\begin{itemize}
\item
Each unlabeled vertex has valency at least 3, including at least one ingoing and one outgoing edge;
\item
there exist no directed cycles;
\item
for $1\leq j<k \leq n$, there does not exist a directed path from $j$ to $k$.
\end{itemize}
\end{df}
\begin{df}
A graph $\gamma\in \Gcal_6(n)$ defines an $n$-ary multidifferential operator, $\acts{\gamma}$, as follows:
\begin{itemize}
\item
A directed edge represents $\Delta_S^{\mathrm A}$;
\item
an undirected edge represents $\Delta_{S_0}^{\mathrm F}$;
\item
the vertex $j$ represents a derivative of the $j$'th argument;
\item
an unlabelled vertex represents a derivative of $S$.
\end{itemize}
\end{df}
\begin{rem}
Note that the Feynman propagator $\Delta_{S_0}^{\mathrm F} = \frac{i}2\left(\Delta_{S_0}^{\mathrm R}+\Delta_{S_0}^{\mathrm A}\right) + H$ is defined by the free action $S_0$ and the Hadamard distribution $H$. It is not a natural object from the point of view of the full action $S$. 
%
This suggests that there should be an alternative to $\Qcal_H$ that is more adapted to $S$, but at the moment we don't have a concrete proposal.
\end{rem}
\begin{lemma}
\[
\TT_H = \sum_{\gamma\in \Gcal_6(1)} \frac{\hbar^{e(\gamma)}}{\abs{\Aut\gamma}}\acts\gamma
\]
\end{lemma}
\begin{proof}
$\TT_H \doteq \alpha_{\Delta_{S_0}^{\mathrm F}}=e^{\frac\hbar2 \DFp}$ and $\DFp\doteq \Dcal_{\Delta^{\mathrm F}_{S_0}}$ is the operator given by the graph
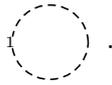
\[
\GraphThirtyOne\,.
\]

A graph in $\Gcal_6(1)$ has no directed edges and no unlabelled vertices; it is just a bouquet of undirected loops. Consider the unique graph with $e(\gamma)=m$ loops. Because the functional derivative of $\Delta_{S_0}^{\mathrm F}$ is $0$, this graph gives the operator $\acts\gamma=\DFp^m$. Its automorphism group is $\Aut \gamma = S_m\rtimes\mathbb{Z}_2^m$. So,
\[
\frac{\hbar^{e(\gamma)}}{\abs{\Aut\gamma}}\acts\gamma = \frac{\hbar^m}{2^m m!}\DFp^m \,.
\qedhere
\]
\end{proof}

\begin{df}
$\Gcal_7(n)\subset\Gcal_6(n)$ is the subset of graphs with no loops at labelled vertices (i.e., no edge begins and ends at the same labelled vertex).

$\Gcal_8(n)\subset\Gcal_6(n)$ is the subset of graphs with no loops.
\end{df}
\begin{thm}\label{non:pert:H}
\be
\label{HV product}
F \starHint G = \sum_{\gamma\in \Gcal_7(2)} \frac{(-i)^{v(\gamma)+d(\gamma)}\hbar^{e(\gamma)-v(\gamma)}}{\abs{\Aut\gamma}} \acts{\gamma}(F,G) 
\ee
where $d(\gamma)$ is the number of directed edges. In particular, this is a finite sum at each order in $\hbar$.
\end{thm}
\begin{proof}
First, observe that
\[
\TT_H(F\starbint G) = \sum_{\gamma\in \Gcal_6(2)} \frac{(-i)^{v(\gamma)+d(\gamma)}\hbar^{e(\gamma)-v(\gamma)}}{\abs{\Aut\gamma}} \acts{\gamma}(F,G) \,.
\]
Replacing $F$ and $G$ with $\TT_H^{-1}F$ and $\TT_H^{-1}G$ serves to cancel out all terms with loops at $1$ or $2$.
\end{proof}
To extend to local $V$, the usual procedure is to extend  $\DFp$ to a map  $\widetilde{\DFp}$ that coincides with $\DFp$ on regular functionals and vanishes on local ones (see e.g., \cite[Section 6.2.1]{Book}). This implies that $\TT_H$ acts as identity on local functionals.
\begin{prop}
If $V\in\Fcal_\loc$, then the sum in \eqref{non:pert:H} can be taken over $\Gcal_8(2)$.
\end{prop}
\begin{proof}
If $V\in\Fcal_\loc$, then $\widetilde{\DFp}(S)=0$, so any loop at an unlabelled vertex gives $0$.
\end{proof}

By direct inspection of the graphs, it is not clear whether the expression above can be renormalized or not. The problem is related to the presence of free Feynman propagators together with interacting advanced propagators. This is potentially an issue, since $\Delta_{S_0}^{\rm F}-i\Delta_{S}^{\rm A}$ doesn't have the right WF set properties (in contrast to $\Delta_{S_0}^{\rm F}-i\Delta_{S_0}^{\rm A}$).

As mentioned before, this is caused by the fact that in constructing the interacting product we left the time-ordered product unchanged, since the time-ordering of $\starHint$ results again in the same commutative product $\dTH$.

Although we began with a perturbative construction using a free action $S_0$, it is only the time-ordered product that remembers $S_0$ and we would like this dependence to be completely removed in the interacting theory. 

We hope that the results of this paper will allow us in the future to find a better version of  the interacting Wick product, while keeping $\starbint$ unchanged.

\subsection{Formulae for the M{\o}ller operators}
In this section we prove some combinatorial formulae for the quantum M{\o}ller operator, which can be used to streamline computations and might be the starting point for investigating renormalization in the future.

In our terms, a \emph{corolla} is a graph $\gamma\in\G(1)$  such that there is a single edge from 1 to each unlabelled vertex. 
\begin{cor}
\[
r_{\la V}^{-1}(F) = \sum_{\gamma\;\mathrm{corolla}} \frac{\la^{v(\gamma)}}{\abs{\Aut(\gamma)}} \vec{\gamma}(F)
\]
\end{cor}
\begin{proof}
This follows from eq.~\eqref{graph R} by setting $\hbar=0$. The graphs $\gamma\in\Gcal_2(1)$ with $e(\gamma)=v(\gamma)$ are precisely the corollas.
\end{proof}

In our terms, a \emph{tree} is a connected graph $\gamma\in\G(1)$ such that 1 is a source, and each unlabelled vertex has precisely one ingoing edge.
\begin{lemma}
\be
\label{r graph}
r_{\la V}(F) = \sum_{\gamma\;\mathrm{tree}} \frac{(-\la)^{v(\gamma)}}{\abs{\Aut(\gamma)}} \vec{\gamma}(F)
\ee
\end{lemma}
\begin{proof}
Any extension of a corolla by a tree is also a tree, so the composition of $r_{\la V}^{-1}$ with \eqref{r graph} can be computed as a sum over trees. 

Consider a tree, $\beta$. What will the coefficient of this term be? For every subset of leaves (valency 1 vertices) of $\beta$, there is an extension $\alpha\hookrightarrow\beta\twoheadrightarrow\gamma$, where $\alpha$ is $\beta$ without those leaves, and $\gamma$ is the corolla made from those leaves. Note how the sign of the term depends upon the number of leaves that are removed.

Lemma~\ref{Partial Composition} shows that if $\beta$ has $m$ leaves, then the coefficient of this term is a multiple of 
\[
\sum_{k=0}^{m} \binom{m}{k}(-1)^k = 0
\]
unless $m=0$. Therefore the only term in the  composition of $r_{\la V}$ with \eqref{r graph} is given by the unique tree with no leaves; the composition is the identity.
\end{proof}

\begin{df}
$\Gcal_9(1)$ is the set of trees such that no unlabelled vertex has precisely one outgoing edge.
\end{df}
Any tree can be obtained from such a graph by adding vertices along edges. As in Theorem~\ref{Reduced star V}, summing over these gives \emph{interacting} advanced propagators. This gives the formula
\[
r_{\la V}(F) = \sum_{\gamma\in\Gcal_9(1)} \frac{(-1)^{v(\gamma)}}{\abs{\Aut(\gamma)}} \acts{\gamma}(F)
\]
where edges represent $\Delta^{\mathrm A}_S$ and vertices represent derivatives of $\la V$.

In a similar way, inverting eq.~\eqref{graph R} gives a graphical formula for $R_{\TT,\la V}$. 
\begin{df}
$\Gcal_{10}(1) \subset \Gcal(1)$ is the set of graphs such that:
\begin{itemize}
\item
Every unlabelled vertex has at least one incoming edge;
\item
1 is a source (has no incoming edges);
\item
there are no directed cycles.
\end{itemize}
\end{df}
With this,
\[
R_{\TT,\la V}(F) = \sum_{\gamma\in\Gcal_{10}(1)} \frac{(-i\hbar)^{e(\gamma)-v(\gamma)}(-\la)^{v(\gamma)}}{\abs{\Aut(\gamma)}} \vec{\gamma}(F)\,.
\]

Next, this leads to a formula for $R_{H,\la V}$. If $V$ is local (and we set $\TT_H(V)=V$) then this amounts to replacing all products by time-ordered products. In terms of graphs, this is given by attaching undirected $\Delta^{\mathrm F}_{S_0}$-edges to graphs from $\Gcal_{10}(1)$.
\begin{df}
$\Gcal_{11}(1)$ is the set of isomorphism classes of graphs with directed and undirected edges and a labelled vertex 1, such that
\begin{itemize}
\item
Every unlabelled vertex has at least one incoming edge;
\item
1 is a source;
\item
there are no directed cycles;
\item
there are no loops.
\end{itemize} 
\end{df}
With this, for $V$ local,
\[
R_{H,\la V}(F) = \sum_{\gamma\in\Gcal_{11}(1)} \frac{(-i)^{d(\gamma)-v(\gamma)}(-\la)^{v(\gamma)}\hbar^{e(\gamma)-v(\gamma)}}{\abs{\Aut(\gamma)}} \vec{\gamma}(F)
\]
where undirected edges represent $\Delta^{\mathrm F}_{S_0}$, and $d(\gamma)$ is the number of directed edges.

\begin{df}
$\Gcal_{12}(1)\subset \Gcal_{11}(1)$ is the subset of graphs such that that no unlabelled vertex has one incoming edge, one outgoing edge, and no unlabelled edge.
\end{df}
Any graph in $\Gcal_{11}(1)$ can be obtained by adding vertices along directed edges of a graph in $\Gcal_{12}(1)$. In this way, the formula can be reexpressed using the interacting advanced propagator,
\[
R_{H,\la V}(F) = \sum_{\gamma\in\Gcal_{12}(1)} \frac{(-i)^{v(\gamma)+d(\gamma)}\hbar^{e(\gamma)-v(\gamma)}}{\abs{\Aut(\gamma)}} \acts{\gamma}(F) 
\]
where directed edges represent $\Delta^{\mathrm A}_{S}$, undirected edges represent $\Delta^{\mathrm F}_{S_0}$, and unlabelled vertices represent derivatives of $\la V$.

Finally, consider the graphs in $\Gcal_{12}(1)$ without trees branching off of them. These are characterized by the lack of univalent vertices, i.e., leaves.
\begin{df}
$\Gcal_{13}(1)$ is the set of isomorphism classes of graphs with directed and undirected edges and a labelled vertex 1, such that
\begin{itemize}
\item
Every unlabelled vertex has at least one incoming edge and one other edge;
\item
no unlabelled vertex has only one incoming and one outgoing edge;
\item
1 is a source (has no incoming edges);
\item
there are no directed cycles;
\item
there are no loops.
\end{itemize} 
We define an operator
\[
\Omega_{\la V}(F) := \sum_{\gamma\in\Gcal_{13}(1)} \frac{(-i)^{v(\gamma)+d(\gamma)}\hbar^{e(\gamma)-v(\gamma)}}{\abs{\Aut(\gamma)}} \acts{\gamma}(F) \,.
\]
\end{df}

Note that any graph giving a term of order $\hbar^m$ has at most $4m$ vertices and $5m$ edges. There are finitely many such graphs, so $\lambda$ does not need to be a formal parameter in this formula.

Finally, observe that any graph in $\Gcal_{12}(1)$ can be obtained as an extension of a tree in $\Gcal_9(1)$ by a graph in $\Gcal_{13}(1)$. Consequently, the composition $r_{\la V}\circ\Omega_{\la V}$ is a sum over $\Gcal_{12}(1)$ with the same coefficients. Therefore
\[
R_{H,\la V} = r_{\la V}\circ \Omega_{\la V}
\]
which is a nonperturbative formula for $R_{H,\la V}$.

Unfortunately, in order to apply the standard methods of Epstein-Glaser renormalization \cite{DF04}, this formula needs to be expanded in $\lambda$ again, to show cancellation of the lightcone divergences. However, we hope that the nonperturbative formula can nevertheless lead to a well defined object,  if we use a different renormalization method (e.g., through some regularization scheme). This will be investigated in our future work.

\section{Perturbative Agreement}\label{PPA}
Consider the case that $V$ is quadratic, so that $V^{(2)}(\ph)$ is independent of $\ph$, and higher derivatives vanish. In such a situation one can treat the interacting theory exactly and a natural question to ask is how this compares with perturbative treatment. This issue has been discussed in the literature \cite{HW05,DHP} under the name \textit{perturbative agreement}. One way to look at it is to compare the interacting star product obtained by means of quantum M{\o}ller operators with the star product constructed directly from the advanced Green function for the quadratic action $S_0+\la V$. In this case, Theorem~\ref{Reduced star V} becomes:
\begin{cor}
If $V\in\Fcal_\reg$ is quadratic, then $\starbint$ is the exponential product defined by $-i\Delta^{\mathrm A}_{S_0+\la V}$.
\end{cor}
\begin{proof}
An unlabelled vertex of $\gamma\in\Gcal_5(2)$ has valency $r\geq3$ and represents the derivative $S^{(r)}$, which vanishes because $S$ is quadratic, therefore $\acts\gamma=0$ unless $\gamma$ has no unlabelled vertices.
This means that eq.~\eqref{reduced star V} simplifies to something like eq.~\eqref{graph free product}:
\[
F\starbint G = \sum_{\gamma\in\G_1(2)} \frac{(-i\hbar)^{e(\gamma)}}{\abs{\Aut \gamma}} \acts\gamma(F,G) \,. \qedhere
\]
\end{proof}
In \cite{HW05,DHP} the principle of perturbative agreement (PPA) is expressed as a compatibility condition for time-ordered products corresponding to $S_0$ and $S_0+\la V$. In order to prove it in our current setting we need to pass from the $\Qcal_\TT$-identification to the $\Qcal_\Weyl$-identification.

For the the remainder of this section, we \emph{do not} need $V$ to be regular.
\begin{lemma}
If $V\in\Fcal_\mc$ is quadratic then 
\be
\label{quadratic Rbar}
R_{\TT,\la V} = \alpha_{\la\delta}\circ r_{\la V}\,,
\ee
and 
\be
\label{quadratic Moller}
R_{0,\la V} = \TT\circ\alpha_{\la\delta}\circ r_{\la V}\circ\TT^{-1} 
\ee
where $\alpha$ is defined in Proposition~\ref{Star equivalence}, $r_{\la V}$ is the classical M{\o}ller operator, and 
\[
\delta = i\Delta^{\mathrm R}_{S_0} V^{(2)} \Delta^{\mathrm A}_{S_0} 
= i \DashedLineOneTwo \,.
\]
\end{lemma}
\begin{proof}
Because $V$ is quadratic,
\begin{align*}
V(\ph-i\hbar \Delta_{S_0}^{\mathrm A}w) &= V(\ph) - i\hbar \left< V^{(1)}(\ph),\Delta^{\mathrm A}_{S_0}w\right> - \frac{\hbar^2}{2} \left<V^{(2)},(\Delta^{\mathrm A}_{S_0}w)^{\otimes 2}\right> \\
&= V(\ph) - i\hbar \left< V^{(1)}(\ph),\Delta^{\mathrm A}_{S_0}w\right> + \frac{i}2\hbar^2 \left<\delta,w^{\otimes 2}\right> \,.
\end{align*}
In this way, eq.~\eqref{generating function} simplifies to 
\begin{align*}
J(\ph;w) &= \exp\left(\left<\la  V^{(1)}(\ph),\Delta^{\mathrm A}_{S_0}w\right> - \frac{\hbar\la}2 \left<\delta,w^{\otimes 2}\right>\right) \\
&= J_0(\ph;w) \exp\left( - \frac{\hbar\la}2 \left<\delta,w^{\otimes 2}\right>\right)
\,,
\end{align*}
i.e., $J_1(\ph;w)=\exp\left( - \frac{\hbar\la}2 \left<\delta,w^{\otimes 2}\right>\right)$.

A consequence of these simplifications is that when any of these acts (via $\triangleright$) on a regular functional, then the result is still regular. Because of this, the proof of Proposition~\ref{RfromJ} still works.

Now, $(J_0(\ph)\triangleright G)(\ph) = (r_{\la V}^{-1}G)(\ph)$, and $e^{-\hbar\la\delta/2} \triangleright G = \alpha_{-\la \delta}G$. This shows that $R_{\TT,\la V}^{-1}(G) = r_{\la V}^{-1}\circ \alpha_{-\la\delta} (G)$.
Inverting this gives eq.~\eqref{quadratic Rbar}.

Equation \eqref{quadratic Moller} just follows from the relationship between $\Qcal_\TT$ and $\Qcal_\Weyl$.
\end{proof}
\begin{rem}
Equation \eqref{quadratic Rbar} is very similar to eq.~\eqref{R factorization}. The difference is that $R_{\TT,\la V}$ is factored in a different order into classical and ``purely quantum'' part, since eq.~\eqref{R factorization} implies $R_{\TT,\la V} =  r_{\la V}\circ \Upsilon_{\la V}^{-1}$.
\end{rem}
\begin{rem}
	Note that our result holds  for a quadratic $V$ that is local and compactly supported, and that removing the IR regularization is a non-trivial step.
\end{rem}
 After transforming the quantum M{\o}ller map to the $\Qcal_\Weyl$-identification, \eqref{quadratic Rbar} allows us to compute the map $\beta\doteq r_{\la V}^{-1} R_{0,\la V}$ of \cite{DHP}. Note the close resemblance between $\beta$ and $\Upsilon_{\la V}^{-1}$, already pointed out in Remark \ref{compare:w:beta}. We are now ready to provide a streamlined version of the proof of Theorem 5.3 of \cite{DHP}.

\begin{thm}
For $V$ quadratic, the quantum M{\o}ller operator can be expressed in terms of the classical M{\o}ller operator as \cite{DHP}
\be
\label{Italian R}
R_{0,\la V} = r_{\la V}\circ \alpha_{i(\Delta_{S}^{\mathrm D}-\Delta_{S_0}^{\mathrm D})} \,.
\ee
\end{thm}
\begin{proof}
Because $S$ is quadratic, $\mathtt r^{-1}_{\la V}$ is a linear map. Its derivative is constant. Using the notation $\rho$ for this derivative again,
\[
\rho = \left(\mathtt r^{-1}_{\la V}\right)^{(1)}(\ph) = \mathtt r^{-1}_{\la V} = \id+\la\Delta^{\mathrm R}_{S_0} V^{(2)} \,.
\]

Equation \eqref{propagator Moller} simplifies to
\begin{align*}
\rho\circ \Delta^{\mathrm R}_S \circ \rho^{\mathsf T} &= \Delta^{\mathrm R}_{S_0} + \la \Delta^{\mathrm R}_{S_0} V^{(2)} \Delta^{\mathrm A}_{S_0} \\
&= \Delta^{\mathrm R}_{S_0} - i\la\delta \,.
\end{align*}
Taking the transpose gives the same identity for the advanced propagators, and adding them shows that
\[
\rho\circ \Delta^{\mathrm D}_S \circ \rho^{\mathsf T} = \Delta^{\mathrm D}_{S_0} - i\la\delta \,.
\]

%


The point of this is that
\[
\alpha_{i\Delta^{\mathrm D}_S}\circ r_{\la V}^{-1} =  r_{\la V}^{-1} \circ \alpha_{i\rho\circ\Delta^{\mathrm D}_S\circ \rho^{\mathsf T}} 
= r_{\la V}^{-1} \circ \alpha_{i\Delta^{\mathrm D}_{S_0}+\la\delta}
\,.
\]
Remembering that $\TT = \alpha_{i \Delta_{S_0}^{\mathrm D}}$, this gives
\[
r_{\la V}\circ \alpha_{i\Delta_{S}^{\mathrm D}} = \TT\circ\alpha_{\la \delta} \circ r_{\la V} \,.
\]
Equation \eqref{quadratic Moller} immediately becomes
\[
R_{0,\la V} = r_{\la V}\circ \alpha_{i\Delta_{S}^{\mathrm D}} \circ \TT^{-1}
\]
and using the definition of $\TT$ again gives eq.~\eqref{Italian R}
\end{proof}
Relation \eqref{Italian R} is in fact the condition of perturbative agreement, since  $\alpha_{i(\Delta_{S}^{\mathrm D}-\Delta_{S_0}^{\mathrm D})}$ is the map that intertwines between the time-ordered product corresponding to $S_0$ and the one corresponding to $S$. Another way of expressing this relation is
\[
R_{0,\la V}\circ \Tcal=r_{\la V}\circ \Tcal_S\,,
\]
where $\Tcal$ is the time-ordering map corresponding to $S_0$ and $ \Tcal_S$ is the analogous map corresponding to $S$.
\section{Conclusions}
In this paper we have derived an explicit formula for the deformation quantization of a general class of infinite dimensional Poisson manifolds. We have also investigated the relation between our formula and the Kontsevich formula. Although the latter hasn't been generalized to the infinite dimensional setting (as for now), one can check if our formula could in principle be derived from one that involves only the propagator (in our case, the advanced propagator) and its derivatives. By direct inspection of the graphs that appear in the third order of our expansion of the interacting star product, we have shown that our expressions cannot be derived by only using the Kontsevich-type graphs. The extra information that we need (apart from the knowledge of the propagator and its derivatives) is the action $S$. This is, however, always provided in the models we are working with, since they arise from classical field theory formulated in a Lagrangian setting (the Poisson structure is constructed using the action $S$). It would be interesting to understand if it is even possible to construct a star product in the infinite dimensional setting without using some additional structure; otherwise, our formula may be the best analogue of Kontsevich's formula in this setting. We want to investigate this problem in our future work.

 Our results hold for a restricted class of functions on the manifold in question and in order to generalize these results, one needs to perform renormalization. 
 We want to investigate this in future research. There are two possible ways forward. One is to use regularization and perform computations in concrete examples, to see how the divergences can be removed. This is expected to work for the quantum M{\o}ller operator itself (since it was constructed by perturbative methods, e.g., by \cite{DF04}) and could also shed some light on the singularities of the interacting product. For the latter, another strategy is to try to correct $\starHint$ by modifying the normal ordering quantization map $\Qcal_H$.
 
 We also want to see how the present results are compatible with the algebraic adiabatic limit \cite{BDF,JMPReview}  and the general framework proposed in \cite{Eli16}.
 
\section*{Acknowledgements}
This research was partially supported by KR's EPSRC grant \verb|EP/P021204/1|.
\bibliographystyle{alpha}

\end{document}